\documentclass[11pt, letterpaper]{article}
\pdfoutput=1

\usepackage{fullpage}
\usepackage{amsmath,amsfonts,amsthm,mathrsfs,xspace,graphicx,mathtools}
\usepackage{dsfont}
\usepackage{nicefrac}
\usepackage{endnotes}
\usepackage{mdframed}
\usepackage{setspace}
\usepackage{xcolor}
\definecolor{linkblue}{HTML}{001487}
\usepackage[colorlinks=true,allcolors=linkblue]{hyperref}
\usepackage[scaled=0.86]{helvet}
\usepackage{mathptmx}
\DeclareMathAlphabet{\mathcal}{OMS}{ntxm}{m}{n}
\let\mathbb\relax
\let\mathbb\mathds
\usepackage{amssymb,latexsym}
\usepackage{tikz}
\usepackage{authblk}

\usepackage{url}
\usepackage[stable]{footmisc}
\usepackage[nameinlink,capitalize,noabbrev]{cleveref}
\usepackage{longfbox}

\usepackage{enumitem}
\setenumerate[0]{label={\normalfont (\roman*)}}

\newtheorem{theorem}{Theorem}[section]
\newtheorem*{theorem*}{Theorem}
\newtheorem{proposition}[theorem]{Proposition}

\newtheorem{lemma}[theorem]{Lemma}

\newtheorem{corollary}[theorem]{Corollary}

\theoremstyle{remark}
\newtheorem{remark}[theorem]{Remark}

\theoremstyle{definition}
\newtheorem{definition}[theorem]{Definition}

\newtheorem{protocol}{Protocol}

\numberwithin{equation}{section}
\newcommand\numberthis{\addtocounter{equation}{1}\tag{\theequation}}

\newcommand{\ket}[1]{|#1\rangle}
\newcommand{\bra}[1]{\langle#1|}
\newcommand{\proj}[1]{\ket{#1}\!\bra{#1}}
\DeclarePairedDelimiterX\braket[2]{\langle}{\rangle}{#1 \delimsize\vert #2}

\newcommand{\tr}[1]{\mbox{\rm Tr}\!\left[ #1 \right]}
\newcommand{\Tr}{\mbox{\rm Tr}}

\newcommand{\pr}[1]{{\rm Pr}\!\left[ #1 \right]}
\newcommand{\1}{\mathbb{1}}

\newcommand{\C}{\ensuremath{\mathbb{C}}}
\newcommand{\N}{\ensuremath{\mathbb{N}}}
\newcommand{\R}{\ensuremath{\mathbb{R}}}

\newcommand{\Setup}{\ensuremath{\mathsf{Setup}}\xspace}

\let\H\relax
\newcommand{\H}{\mathcal{H}}
\DeclareMathOperator*{\E}{\mathbb{E}}

\newcommand{\ot}{\ensuremath{\otimes}}
\newcommand{\deq}{\coloneqq}

\newcommand{\mA}{\ensuremath{\mathcal{A}}}

\newcommand{\mG}{\ensuremath{\mathcal{G}}}
\newcommand{\mD}{\ensuremath{\mathcal{D}}}
\newcommand{\mI}{\ensuremath{\mathcal{I}}}
\newcommand{\mF}{\ensuremath{\mathcal{F}}}
\newcommand{\mK}{\ensuremath{\mathcal{K}}}
\newcommand{\mL}{\ensuremath{\mathcal{L}}}
\newcommand{\mM}{\ensuremath{\mathcal{M}}}
\newcommand{\mS}{\ensuremath{\mathcal{S}}}

\newcommand{\mX}{\ensuremath{\mathcal{X}}}
\newcommand{\mY}{\ensuremath{\mathcal{Y}}}
\newcommand{\kf}{\ensuremath{\mathcal{K}_{{1}}}}
\newcommand{\kg}{\ensuremath{\mathcal{K}_{{0}}}}

\newcommand{\bit}{\{0,1\}}
\newcommand{\rand}{\raisebox{-1pt}{\ensuremath{\,\xleftarrow{\raisebox{-1pt}{$\scriptscriptstyle\$$}}\,}}}

\newcommand{\setft}[1]{\textnormal{#1}}

\newcommand{\pos}{\setft{Pos}}

\DeclareMathOperator{\poly}{poly}
\DeclareMathOperator{\negl}{negl}

\newcommand{\bits}{\ensuremath{\{0, 1\}}}

\newcommand{\supp}{\setft{supp}}
\newcommand{\Gen}{\textsc{Gen}}

\newcommand{\QCP}{\mathsf{QCP}}
\newcommand{\Adv}{\mathsf{Adv}}
\newcommand{\CC}{\mathsf{CC}}

\newcommand{\CPTP}{\mathsf{CPTP}}

\newcommand{\id}{\textsf{id}}

\newcommand{\eps}{\varepsilon}

\newcommand{\Samp}{\textsc{Samp}}

\newcommand{\Chk}{\ensuremath{\textsc{Chk}}}

\newcommand{\KeyGen}{\ensuremath{\mathsf{KeyGen}}\xspace}
\newcommand{\Enc}{\ensuremath{\mathsf{Enc}}\xspace}
\newcommand{\Dec}{\ensuremath{\mathsf{Dec}}\xspace}
\newcommand{\WKD}{\ensuremath{\mathsf{WKD}}\xspace}
\newcommand{\ts}{\ensuremath{\mathsf{ts}}\xspace}

\newcommand{\ct}{\mathsf{ct}\xspace}
\newcommand{\ketbra}[2]{\left|#1\right\rangle\!\!\left\langle #2\right|}

\newtheorem{construction}{Construction}

\newcommand{\abs}[1]{\left\vert {#1} \right\vert}
\newcommand{\norm}[1]{\left\| {#1} \right\|}

\newcommand{\QECM}{\ensuremath{\mathsf{QECM}}\xspace}

\newcommand{\IND}{\ensuremath{\mathsf{IND}}\xspace}
\newcommand{\QPT}{\ensuremath{\mathsf{QPT}}\xspace}
\newcommand{\algo}{\mathcal}

\newcommand{\Sim}{\ensuremath{\mathsf{Sim}}\xspace}
\newcommand{\QCED}{\ensuremath{\mathsf{QCED}}\xspace}
\newcommand{\QCEDCC}{\ensuremath{\mathsf{QCED_{CC}}}\xspace}
\newcommand{\VDBQC}{\ensuremath{\mathsf{VDBQC}}\xspace}
\newcommand{\VDBQCCC}{\ensuremath{\mathsf{VDBQC_{CV}}}\xspace}
\newcommand{\Evaluate}{\ensuremath{\mathsf{Evaluate}}\xspace}
\newcommand{\StatePrep}{\ensuremath{\mathsf{StatePrep}}\xspace}
\newcommand{\PPT}{\ensuremath{\mathsf{PPT}}\xspace}
\newcommand{\pk}{\mathsf{pk}\xspace}
\newcommand{\TVD}{\mathsf{TVD}\xspace}
\newcommand{\sk}{\mathsf{sk}\xspace}

\newcommand{\sth}{\setft{~s.t.~}}
\newcommand{\tand}{\;\textnormal{~and~}\;}

\let\eps\varepsilon
\newcommand{\capprox}{\stackrel{c}{\approx}}

\makeatletter
\renewcommand{\paragraph}{%
  \@startsection{paragraph}{4}%
  {\z@}{2.25ex \@plus 1ex \@minus .2ex}{-1em}%
  {\normalfont\normalsize\bfseries}%
}
\makeatother
\crefname{figure}{Protocol}{Protocols}
\Crefname{figure}{Protocol}{Protocols}

\DeclareSymbolFont{greekletters}{OML}{ntxmi}{m}{it}
\DeclareMathSymbol{\alpha}{\mathord}{greekletters}{"0B}
\DeclareMathSymbol{\beta}{\mathord}{greekletters}{"0C}
\DeclareMathSymbol{\gamma}{\mathord}{greekletters}{"0D}
\DeclareMathSymbol{\delta}{\mathord}{greekletters}{"0E}
\DeclareMathSymbol{\epsilon}{\mathord}{greekletters}{"0F}
\DeclareMathSymbol{\zeta}{\mathord}{greekletters}{"10}
\DeclareMathSymbol{\eta}{\mathord}{greekletters}{"11}
\DeclareMathSymbol{\theta}{\mathord}{greekletters}{"12}
\DeclareMathSymbol{\iota}{\mathord}{greekletters}{"13}
\DeclareMathSymbol{\kappa}{\mathord}{greekletters}{"14}
\DeclareMathSymbol{\lambda}{\mathord}{greekletters}{"15}
\DeclareMathSymbol{\mu}{\mathord}{greekletters}{"16}
\DeclareMathSymbol{\nu}{\mathord}{greekletters}{"17}
\DeclareMathSymbol{\xi}{\mathord}{greekletters}{"18}
\DeclareMathSymbol{\pi}{\mathord}{greekletters}{"19}
\DeclareMathSymbol{\rho}{\mathord}{greekletters}{"1A}
\DeclareMathSymbol{\sigma}{\mathord}{greekletters}{"1B}
\DeclareMathSymbol{\tau}{\mathord}{greekletters}{"1C}
\DeclareMathSymbol{\upsilon}{\mathord}{greekletters}{"1D}
\DeclareMathSymbol{\phi}{\mathord}{greekletters}{"1E}
\DeclareMathSymbol{\chi}{\mathord}{greekletters}{"1F}
\DeclareMathSymbol{\psi}{\mathord}{greekletters}{"20}
\DeclareMathSymbol{\omega}{\mathord}{greekletters}{"21} 
\DeclareMathSymbol{\vartheta}{\mathord}{greekletters}{"23}
\DeclareMathSymbol{\varpi}{\mathord}{greekletters}{"24}
\DeclareMathSymbol{\varrho}{\mathord}{greekletters}{"25}
\DeclareMathSymbol{\varsigma}{\mathord}{greekletters}{"26}
\DeclareMathSymbol{\varphi}{\mathord}{greekletters}{"27}

\begin{document}
\let\ref\cref

\sloppy

\title{Quantum cryptography with classical communication \\ 
\vspace{0.5cm}
{\Large Parallel remote state preparation for\\\vspace{-0.13cm} copy-protection, verification, and more}\vspace{0.5cm}}
\date{}

\author[1]{Alexandru Gheorghiu}
\author[2]{Tony Metger\footnote{Corresponding author: \href{mailto:tmetger@ethz.ch}{tmetger@ethz.ch}}}
\author[3]{Alexander Poremba}
\affil[1]{Institute for Theoretical Studies, ETH Z{\"u}rich, 8092 Z{\"u}rich, Switzerland}
\affil[2]{Institute for Theoretical Physics, ETH Z{\"u}rich, 8093 Z{\"u}rich, Switzerland}
\affil[3]{Computing and Mathematical Sciences, California Institute of Technology, CA 91125, United States}
\maketitle
\thispagestyle{empty}

\begin{abstract}
Quantum mechanical effects have enabled the construction of cryptographic primitives that are impossible classically. For example, quantum copy-protection allows for a program to be encoded in a quantum state in such a way that the program can be evaluated, but not copied. Many of these cryptographic primitives are two-party protocols, where one party, Bob, has full quantum computational capabilities, and the other party, Alice, is only required to send random BB84 states to Bob. In this work, we show how such protocols can \emph{generically} be converted to ones where Alice is fully classical, assuming that Bob cannot efficiently solve the LWE problem. In particular, this means that all communication between (classical) Alice and (quantum) Bob is classical, yet they can still make use of cryptographic primitives that would be impossible if both parties were classical. We apply this conversion procedure to obtain quantum cryptographic protocols with classical communication for unclonable encryption, copy-protection, computing on encrypted data, and verifiable blind delegated computation.

The key technical ingredient for our result is a protocol for \emph{classically-instructed parallel remote state preparation of BB84 states}. This is a multi-round protocol between (classical) Alice and (quantum polynomial-time) Bob that allows Alice to certify that Bob must have prepared $n$ uniformly random BB84 states (up to a change of basis on his space). Furthermore, Alice knows which specific BB84 states Bob has prepared, while Bob himself does not. Hence, the situation at the end of this protocol is (almost) equivalent to one where Alice sent $n$ random BB84 states to Bob. This allows us to replace the step of preparing and sending BB84 states in existing protocols by our remote-state preparation protocol in a generic and modular way.
\end{abstract}

\newpage 
\thispagestyle{empty}
{
\hypersetup{linkcolor=black}
\tableofcontents
}

\newpage

\section{Introduction}
A central distinction between classical and quantum information is that a classical string can always be copied, but a quantum state cannot: the \emph{no-cloning theorem} states that there cannot exist a procedure that produces the state $\rho \otimes \rho$ when given as input an arbitrary quantum state $\rho$~\cite{wootters1982single}. 
The first cryptographic protocols that made use of the no-cloning theorem were Wiesner's proposal to use quantum states as unforgeable banknotes~\cite{wiesner} and Bennett and Brassard's protocol for information-theoretically secure quantum key-distribution (the BB84 QKD protocol)~\cite{bb84}.
These protocols rely on the idea of a \emph{conjugate coding scheme}: classical information can be encoded into a quantum state in (at least) two incompatible bases, most commonly the standard basis $\{\ket{0}, \ket{1}\}$ and the Hadamard basis $\{\ket{+}, \ket{-}\}$, where $\ket{\pm} = \frac{1}{\sqrt{2}}(\ket{0} \pm \ket{1})$. These four states are commonly referred to as \emph{BB84 states}.
If we encode a bit $b \in \{0,1\}$ as either $\ket{b}$ or $\ket{(-)^b} = \frac{1}{\sqrt{2}}(\ket{0} + (-1)^b \ket{1})$, then an adversary who does not know which basis we chose for the encoding cannot create a copy of this quantum state.
Furthermore, if the adversary tries to measure the state, with probability $1/2$ they will choose the ``wrong'' measurement basis, which disturbs the state and means that the adversary's tampering can be detected.

There is an important conceptual difference between the BB84 protocol and Wiesner's quantum money scheme. The former addresses the problem of key-distribution,  which is a task that can also be achieved classically under computational assumptions using public-key cryptography~\cite{dh}.
In contrast, Wiesner's quantum money scheme achieves a functionality which is entirely impossible classically, even under computational assumptions.
Recently there has been renewed interest in this latter kind of application, i.e.~to use BB84 states to construct quantum cryptographic primitives that have no classical analogue.
Perhaps the most striking example of this is the idea of \emph{quantum copy-protection}~\cite{Aar2009}. Suppose that a vendor has created a piece of software (viewed as a function that maps some input to some output) and wants to allow a user to run it (i.e. to evaluate the function), while preventing the user from producing additional ``pirated'' copies of the original software.
Clearly, this is impossible classically: any piece of software is specified by a string of symbols, which can easily be copied.
Surprisingly, it has been shown that it is possible to encode certain narrow classes of functions in the form of a quantum state in such a way that a user can evaluate the function without being able to copy it~\cite{coladangelo2020quantum}.

Copy-protection and many related protocols require only limited quantum capabilities from one party, e.g.~the \emph{vendor} in the case of copy-protection: they only need to prepare random BB84 states and send them to the other party (e.g.~the \emph{user} in copy-protection), who has full quantum computational capabilities.
In particular, this requires a quantum channel between the two parties to send the BB84 states.
The purpose of this paper is to show that such protocols, where one party's quantum operations are limited to preparing and sending random BB84 states, can be converted into protocols where that party is \emph{fully classical}.
This \emph{dequantises} such protocols in the sense that all communication becomes classical.
To achieve this, we need to construct a protocol between a \emph{classical verifier} and a \emph{computationally bounded quantum prover} that achieves the same outcome as if the verifier had prepared and sent random BB84 states to the prover.
We call this task \emph{classically-instructed parallel remote state preparation of BB84 states}, or \emph{parallel RSP} for short.
Our protocol builds on techniques introduced in~\cite{mahadev, randomness, rsp} that allow the verifier to use post-quantum cryptography to constrain the actions of an untrusted (but computationally bounded) prover.
Proving soundness for this parallel RSP protocol is the main technical result of our work. 
We then use this result to dequantise a number of cryptographic protocols, namely unclonable quantum encryption, quantum copy-protection, quantum computing on encrypted data and blind verification of quantum computation. 

\subsection{Main results} \label{sec:main_results}
We start by first describing the soundness guarantee achieved by our parallel RSP protocol.
Intuitively, the goal of our protocol is to guarantee that the prover has prepared a quantum state of the form $H^{\theta_1} \proj{v_1} H^{\theta_1} \ot \dots \ot H^{\theta_n} \proj{v_n} H^{\theta_n}$, where $\vec{v}, \vec{\theta} \in \{0, 1\}^n$. Additionally, the prover should not have any information about $\vec v$ and $\vec \theta$ beyond what is contained in its BB84 states, while the verifier should know both $\vec v$ and $\vec \theta$.
Our protocol achieves a guarantee of this kind assuming the quantum-intractability of the \emph{Learning with Errors} (LWE) problem (with quantum advice, see \cref{rem:advice_state}) introduced by Regev~\cite{lwe}. Our main result is the following (see \ref{thm:final_rigidity} for the corresponding formal statement):
\begin{theorem}[Informal] \label{thm:inf}
There exists an interactive protocol between a classical verifier and a computationally bounded quantum prover such that the following holds assuming the quantum-intractability of LWE (with quantum advice).
Fix a number $n$ of BB84 states and let $W$ and $P$ be the verifier's and prover's systems at the end of the protocol, respectively.
Then there exists an isometry $V: P \to QP'$ (for $\H_Q \cong (\C^2)^{\ot n}$ and $P'$ arbitrary) and an additional (subnormalised) state $\alpha_{P'}$ such that for any basis choice $\vec \theta \in \bits^n$, the protocol's final state $\sigma_{WP}$ conditioned on the prover being accepted satisfies
\begin{align*}
p_{\textnormal{success}} V \sigma_{WP} V^\dagger \capprox_{1/\poly(n)} \frac{1}{2^n} \sum_{\vec v \in \bits^n} \proj{v}_W \otimes \left( H^{\theta_1} \proj{v_1} H^{\theta_1} \ot \dots \ot H^{\theta_n} \proj{v_n} H^{\theta_n}  \right) \ot \alpha_{P'} \,.
\end{align*}
Here, $p_{\textnormal{success}}$ is the prover's success probability in the protocol and $\capprox_{1/\poly(n)}$ denotes computational indistinguishability up to inverse polynomial error.
\end{theorem}

We make two remarks regarding this security guarantee.
Firstly, the theorem makes a statement about the joint state of the verifier's system $W$ and the prover's system $P$ after applying an isometry $V$ that \emph{only acts on the prover's space}.
This additional isometry is unavoidable: it represents the prover's freedom to use any basis of its choice on its space.
Hence, we cannot guarantee that the prover prepares BB84 states (in the standard basis), only that it prepares BB84 states up to a change of basis.
However, crucially this change of basis is \emph{independent of which BB84 state was supposed to be prepared}, i.e., $V$ is independent of $\vec v$ and $\vec \theta$.
Put differently, the theorem guarantees that the prover prepares one of $4^n$ possible states whose relation to each other is the same as the relation between the $4^n$ BB84 states.
This does not affect the utility of the prover's state for applications.
In fact, this freedom also exists if the verifier sent $n$ BB84 states to the prover via a quantum channel: the prover could apply an isometry $V$ to these states immediately upon receipt, but the security of any application using the BB84 states is not impacted by this.

Secondly, the theorem holds \emph{for any} basis choice $\vec \theta$, but \emph{on average} over the values $\vec v$. In other words, in the protocol, the verifier gets to choose the bases at will, but the values will be uniformly random and cannot be chosen by the verifier.
Furthermore, the only dependence on $\vec v$ and $\vec \theta$ in the prover's state is via the BB84 states.
This means that the protocol forces the prover to prepare these states ``blindly'', i.e., the prover does not know which BB84 states were actually prepared. In contrast, the verifier does know, because they chose $\vec \theta$ and are in possession of the system $W$, which contains information about $\vec v$.
This asymmetry of knowledge about the prover's state is the same as what is achieved by preparing and sending BB84 states through a quantum channel and is crucial for applications.

Having introduced our parallel RSP theorem, we can turn to its cryptographic applications.
We consider various cryptographic primitives that have previously been defined and constructed in a setting where one party sends random BB84 states to the other.
For each primitive, we give a formal definition of the ``classical-client version'' and show that this definition can be satisfied using our parallel RSP protocol as a building block. 
Since our parallel RSP protocols relies on the LWE assumption, so do the dequantised protocols we present here.
Furthermore, \ref{thm:inf} only guarantees the preparation of BB84 states up to an inverse polynomial error, so as a result, the dequantised protocols only have inverse polynomial security (see \ref{sec:discussion} for a discussion of this point).
Some of these primitives have previously been dequantised using an application-specific approach (and similarly relying on computational assumptions) \cite{rsp, qfactory, radian2022semi, hiroka2021quantum, kitagawa2021secure}; in contrast, our approach is \emph{generic} and simply uses RSP to replace the sending of BB84 states.
We give a short overview of the different applications and refer to \ref{part:applications} for details.

\paragraph{Unclonable quantum encryption.} As a first application of our parallel RSP protocol, we consider the notion of \emph{unclonable quantum encryption}. This cryptographic functionality was coined by Gottesman~\cite{gottesman2004uncloneable} and then formalised by Broadbent and Lord~\cite{BroadbentL19}. In a private-key unclonable quantum encryption scheme, a classical message is encrypted into a quantum state (the \emph{quantum ciphertext}) with the following property: given only a single quantum ciphertext, it is impossible to create two states that can later both be decrypted with access to the private key.
We consider an unclonable conjugate coding \emph{hybrid encryption} scheme which is inspired by the work of Broadbent and Lord: a plaintext $\vec m \in \bit^n$ is encrypted with a randomly chosen secret key $k=(\vec s,\vec \theta) \rand \bit^n \times \bit^n$ and randomness $\vec v \rand \bit^n$ into the quantum ciphertext given by $\Enc_k(\vec m) =  \bigotimes_{i=1}^n H^{\theta_i} \proj{v_i} H^{\theta_i} \ot \proj{\vec v \oplus \vec s \oplus \vec m}\,.$
To decrypt using the secret key $k = (\vec s,\vec \theta)$, one applies $H^{\theta_1} \otimes \dots \otimes H^{\theta_n}$ to the first half of the ciphertext, measures in the computational basis with outcome $\vec x$, and then uncomputes the one-time pad in the second half using $\vec x$ and $\vec s$.
The fact that this scheme is unclonable is a consequence of the \emph{monogamy of entanglement}~\cite{BroadbentL19,Tomamichel2013}.

To dequantise this protocol, we consider a scenario in which a classical client $\algo C$ wishes to delegate an unclonable ciphertext to a quantum receiver $\algo R$. As a first step, $\algo C$ and $\algo R$ run our parallel RSP protocol to delegate a collection of random BB84 states of the form $H^{\theta_1} \ket{v_1} \otimes \dots \otimes H^{\theta_n}\ket{v_n}$, where $\vec v,\vec \theta \in \bit^n$ are random strings known only to $\algo C$. Then, $\algo C$ can choose $\vec s \in \bits^n$ and output the string $\vec v \oplus \vec s \oplus \vec m$ and set $\vec k = (\vec s,\vec \theta)$ as the secret key.
With this choice of key, the delegated parallel BB84 states are exactly the ciphertext $\Enc_k(\vec m)$. Because the final output state of the protocol is computationally indistinguishable from a tensor product of BB84
states (known to the client), we can follow a similar proof as in~\cite{BroadbentL19} to obtain a classical-client unclonable encryption scheme with inverse-polynomial security.

\paragraph{Quantum copy-protection.} In quantum copy-protection (QCP), a vendor wishes to encode a program into a quantum state in a way that enables a recipient to run the program, but not to create functionally equivalent ``pirated'' copies. The notion of QCP was introduced by Aaronson~\cite{Aar2009}, who gave the first construction for unlearnable and efficiently computable functions in a strong quantum oracle model, which has since been improved to only requiring classical oracles~\cite{aaronson2020new}.
Recent work~\cite{coladangelo2020quantum} has also provided the first construction of QCP for compute-and-compare programs in the quantum random oracle model (QROM) as well as a scheme for multi-bit point functions in the QROM based on unclonable encryption with wrong-key detection (WKD) -- a property which enables the decryption procedure to recognise incorrect keys.

Our QCP scheme for multi-bit point functions combines our unclonable hybrid encryption scheme with the generic WKD transformation in the QROM proposed by Coladangelo et al.~\cite{coladangelo2020quantum}. The basic idea behind our QCP scheme is as follows. To encode a point function $P_{\vec y,\vec m}$ (which is defined as returning $\vec m$ on input $\vec y$ and $0^n$, otherwise) we simply output $\Enc_{\vec y}(\vec m)$ together with $h(\vec y)$, where $h$ is a suitable hash function which we model as a truly random function (in the QROM).
To evaluate the program on an input $\vec x \in \bit^{2n}$, we first check whether $\vec x$ hashes to $h(\vec y)$ under $h$. If true, we decrypt as in the aforementioned hybrid encryption scheme and recover $\vec m$. Otherwise, we output $0^n$.

We then show how to obtain a QCP scheme with a classical client through the use of our parallel RSP protocol for preparing random BB84 states, similar to our aforementioned classical-client unclonable encryption scheme. Our scheme enables a classical client to delegate a correct copy-protected program from the class of multi-bit point functions consisting of uniformly random marked inputs $\vec y$ and output strings $\vec m$ with inverse-polynomial security.

\paragraph{Quantum computing on encrypted data.}
Suppose a client wishes to perform some quantum computation, represented as the action of a quantum circuit $C$ on an input state $\ket{x}$, with $x \in \{0, 1\}^n$. 
For simplicity, we will assume the desired output is classical and corresponds to a computational basis measurement of $C \ket{x}$. 
The client only has limited quantum capability and therefore wishes to delegate the computation to a quantum server while ensuring the privacy of the input $\ket{x}$ and the output resulting from the measurement of $C \ket{x}$. Essentially, the client would like to send the server an encryption of the input and, after performing an interactive protocol, obtain an encryption of the output (which the client can decrypt, but the server cannot)\footnote{This also allows the client to hide the computation itself from the server by suitably encoding it as part of the input $x$ and taking $C$ to be a \emph{universal circuit}. When the primary goal of the protocol is to hide the computation, it is referred to as a \emph{blind quantum computing protocol}~\cite{arrighi2006blind, broadbent2009universal}.}.
This primitive is called quantum computing on encrypted data (QCED).

Many protocols for QCED with differing quantum requirements on the client have been developed (see~\cite{fitzsimons2017private} for a survey). Here we will focus on the protocol of Broadbent~\cite{broadbent2015delegating} which achieves QCED with a client that is only required to prepare BB84 states and send them to the server. This makes the protocol well-suited for dequantisation via our parallel RSP protocol.
Before explaining this dequantisation, we (informally) define what a QCED protocol \textit{with a classical client} should achieve. As before, the client's input is the string $x \in \{0, 1\}^n$ and the goal is to obtain the outcome of measuring $C \ket{x}$ in the computational basis. In contrast to before, this must be achieved using only \emph{classical} interaction with the quantum server. The requirement that the client's input must stay private is captured by the condition that after interacting with the client, it must be computationally intractable for the server to decide which one of two distinct inputs the client used.

Our QCED protocol with a classical client works as follows. The client first performs the parallel RSP protocol with the server, resulting in the preparation of BB84 states (or the client aborting). Provided the protocol succeeded, the client proceeds to run Broadbent's protocol~\cite{broadbent2015delegating} as if the server had received those BB84 states via a quantum channel. The security proof is straightforward. First, we know that after performing RSP the server's state is computationally indistinguishable from a tensor product of BB84 states (known to the client). Furthermore, the interaction in~\cite{broadbent2015delegating} preserves this computational indistinguishability. Hence, the server's state at the end of the protocol is indistinguishable from the state the server would have obtained by executing the protocol with random BB84 states and the security of our protocol follows from~\cite{broadbent2015delegating}.

\paragraph{Verifiable delegated blind quantum computation.} 

The final application we consider is verifiable delegated blind quantum computation (VDBQC). VDBQC is an interactive protocol between two parties, in this case denoted as the verifier and the prover. The verifier delegates a computation to the prover and, in addition to ensuring input-output privacy as in QCED, the protocol also ensures that the probability for the verifier to accept an incorrect output is small. In other words, if the prover deviates from the protocol and does not perform the verifier's instructed computation, the verifier should be able to detect this and abort with high probability.
As with QCED, a number of such protocols have been developed and we refer the reader to~\cite{gheorghiu2019verification} for a survey.

Here we focus on a protocol by Morimae~\cite{morimae2018blind}. This protocol achieves verifiability by combining a protocol for blind quantum computation (or QCED) with the \emph{history state construction}, which is a special encoding of a quantum circuit into a quantum state~\cite{kitaev2002classical, kempe2006complexity}. 
In Morimae's protocol, for a given circuit $C$ the verifier uses a QCED protocol to delegate to the prover the preparation of two such history states (one for $C$ and one for the complement of $C$, where the output qubit is negated). The verifier then requests these states from the prover and proceeds to measure them in the computational or Hadamard basis. This allows the verifier to determine the output of the computation. The history state construction guarantees that malicious behavior on the prover's part would be detected by the verifier's measurement. Additionally, the use of a QCED protocol ensures that the prover is ``blind'', i.e.~does not know which computation the verifier delegated.

To dequantise this protocol, we use our QCED protocol with a classical client to delegate the preparation of the two history states to the quantum prover. 
We then replace the verifier's measurements on this state by a \textit{measurement protocol} due to Mahadev~\cite{mahadev}, which allows the classical verifier to delegate these measurements to the prover in a way that forces the prover to report the correct outcomes. We thus obtain a VDBQC protocol with a classical verifier. Crucially, through the use of the classical client QCED protocol and Mahadev's measurement protocol, the prover is ``computationally blind'', i.e.~unable to distinguish which computation the verifier has performed. In contrast, Mahadev's verification protocol~\cite{mahadev} does not have this property.\footnote{In~\cite{rsp}, the authors also construct a blind verification protocol based on RSP. However, they approach the problem in a composable framework, which requires them to make an additional assumption on the prover (called the \emph{measurement buffer} in~\cite{rsp}). In contrast, our protocol requires no extra assumptions on the prover. We describe the issue with the measurement buffer assumption in more detail in~\ref{sec:prev_work}.}

\subsection{Related work} \label{sec:prev_work}
A number recent of works starting with~\cite{randomness,mahadev} have developed techniques that allow a classical verifier to use post-quantum cryptography to force an untrusted (but computationally bounded) quantum prover to behave in a certain way.
Here, we briefly describe these works and explain their relation to our parallel RSP protocol.

In a breakthrough result~\cite{mahadev}, Mahadev introduced a protocol that allows a classical verifier to delegate a quantum computation to a quantum computer and be able to verify the correctness of the result.
The key ingredient for this protocol is a \emph{measurement protocol}, which allows the verifier to securely delegate single-qubit measurements in the standard or Hadamard basis to a quantum prover, assuming that the prover cannot break the LWE assumption.
This can then be applied to so-called \emph{prepare-and-measure protocols}: if one has a protocol that involves a quantum prover preparing and sending a quantum state to the verifier and the verifier performing single-qubit measurements on this state, one can use Mahadev's measurement protocol to delegate these quantum measurements to the prover itself.
This yields a protocol in which the prover only sends classical measurement outcomes to the verifier, hence making the verifier classical.

This measurement protocol is in many ways similar to what we seek to do in this paper: it removes the need for quantum communication between a fully quantum prover and a verifier with very limited quantum capabilities (only measuring single qubits in the computational or Hadamard basis).
The difference to our work is that we are concerned with \emph{prepare-and-send protocols}, in which the verifier sends random BB84 states to the prover instead of receiving them.

It turns out that replacing the quantum communication of prepare-and-send protocols requires significantly stronger control over the untrusted prover.
At a high level, the reason is the following: for Mahadev's measurement protocol, it  suffices to show that \emph{there exists} a quantum state that is consistent with the distribution of  measurement outcomes reported by the prover, in the sense that the measurement outcomes for different bases could have been obtained by measurements on (copies of) the same state.
In contrast, if we want to replace the step of the verifier sending a physical quantum state to the prover, we need to show that the prover has \emph{actually constructed} a certain quantum state, not just that such a quantum state exists mathematically.\footnote{In fact, in~\cite{vidick2020classical} it was shown that Mahadev's measurement protocol does ensure that the prover \emph{knows} (in the sense of a proof of knowledge) the state it is measuring, not just that it exists mathematically. The notion of ``knowing'' a quantum state is quite subtle to define and we forego a detailed description here, but point out that this is weaker than showing that the prover actually constructed the state and (to the best of our knowledge) not sufficient to use Mahadev's protocol for prepare-and-measure scenarios.}
We give a more detailed description of what it means to ``actually construct'' a quantum state in \ref{sec:proof_sketch}.

The first classical protocol that provably forced a quantum prover to prepare a certain quantum state was the single-qubit RSP protocol of~\cite{rsp} (see also~\cite{qfactory} for a related result).
This protocol essentially achieves our informal theorem as stated above for a single qubit, i.e.~$n=1$.\footnote{The protocol in~\cite{rsp}  allows for the qubit to be prepared in one of 10 possible states} which includes the 4 BB84 states. Here, we only focus on the 4 BB84 states as this is the case we will deal with in our parallel RSP protocol. At first sight, it might seem as though a simple hybrid argument, which replaces each BB84 qubit with a (sequential) instance of~\cite{rsp}, suffices to achieve the multi-qubit task. However, the single-qubit RSP protocol of~\cite{rsp} only ensures that each BB84 qubit can be individually replaced by an RSP protocol up to a \emph{global} isometry. Because the prover's state can be entangled in arbitrary ways between intermediate applications of the protocol, it is difficult to justify that all of the individual replacements together form an actual $n$-qubit BB84 state; as we explain below, the fact that the protocol from~\cite{rsp} is composable does not remedy this situation, either. While some prior work~\cite{https://doi.org/10.48550/arxiv.1604.01586} showed that composable \emph{single-qubit} RSP suffices in the context of quantum verification, one would have to show a similar result for each application of interest. Our parallel RSP protocol, in contrast, can be used in a plug-and-play manner for many cryptographic protocols and applications. In addition, our protocol has fewer rounds than a sequential repetition of~\cite{rsp} and also immediately yields a \emph{proof of quantum space} (a certificate that the prover has a certain number of qubits).
We give a brief outline of~\cite{rsp} and its soundness proof in \ref{sec:proof_sketch}.

The main difficulty in going from~\cite{rsp} to our parallel RSP result is enforcing a tensor product structure on the prover's space: we would like to show that, if we execute multiple instances of a single-qubit RSP protocol in parallel, a successful prover must treat each of these copies independently.
Mathematically, this means that we need to be able to split the prover's a priori uncharacterised Hilbert space into a \emph{tensor product}, where each tensor factor is supposed to correspond to one instance of the RSP protocol.
This is a more demanding version of the classic question of parallel repetition: there, one is interested in showing that any prover's winning probability in the protocol decays in essentially the same way as it would for a prover who executes the instances independently. 
In contrast, we need to show that the prover really does execute the different instances independently in a physically meaningful sense.
We call this stronger requirement \emph{parallel rigidity}.

In~\cite{rsp}, the authors show that their protocol has composable security. This may suggest that one can obtain a parallel rigidity statement simply by composing the protocol with itself in sequence or in parallel. However, this is not the case because the composable security statement in~\cite{rsp} requires an additional assumption called a \emph{measurement buffer}, which effectively acts as a trusted intermediary between the verifier and the prover. A sequential or parallel composition of the protocol in~\cite{rsp} would utilise a different measurement buffer for each instance, thereby forcing the prover to treat the different instances in a (largely) independent way. In particular, this means that one already assumes a tensor product structure with $n$ separate qubits in the prover's space, whereas in our work enforcing this tensor product structure is the key technical challenge. 
For cryptographic applications, we do not want to place any such assumption on the prover and instead allow the prover to perform arbitrary global operations involving all instances. This is what our parallel RSP protocol achieves.
Furthermore, as shown in~\cite{badertscher2020security}, achieving a composable single-qubit RSP without the measurement buffer is impossible.
This means that one cannot hope to achieve parallel RSP by showing a stronger composable version of single-qubit RSP; instead, it is necessary to directly analyse parallel executions of the protocol, as we do in this paper.

The question of parallel rigidity has been studied extensively in the literature on quantum self-testing~\cite{coudron2016parallel, coladangelo2016parallel, linearity, natarajan2018low}, where one considers a setting of two non-communicating provers.
Unfortunately, those techniques are not immediately transferable to the setting we consider here, namely a single computationally bounded prover. 

Some progress towards the question of parallel rigidity for single computationally bounded provers was made in~\cite{self-testing}, which gives a protocol that allows a classical verifier to certify that a quantum prover must have prepared and measured a Bell state, i.e.~an entangled 2-qubit quantum state.
This has since been applied to device-independent quantum key distribution~\cite{metger2021device} and oblivious transfer~\cite{broadbent2021device}, and been extended to work for magic states~\cite{mizutani2021computational}.
The protocol from~\cite{self-testing} uses a 2-fold parallel repetition of~\cite{rsp} (with additional steps to allow for the certification of an entangled state, not just product states).
As part of their soundness proof,~\cite{self-testing} do show a kind of parallel rigidity result for 2 instances of the RSP protocol.
However, their method does not generalise to an $n$-fold parallel repetition without an exponential decay in parameters.
Hence, for our $n$-fold parallel rigidity proof, new techniques are needed.
We give a more detailed comparison between our new parallel rigidity proof and the method in~\cite{self-testing} at the end of \ref{sec:proof_sketch}.
We note that in independent concurrent work,~\cite{parallel_epr} also gave an 
$n$-fold parallel rigidity proof in the computational setting, but the class of states they deal with is different from random BB84 pairs and they do not consider the dequantisation of cryptographic protocols.

In addition to this line of work focused on rigidity statements, application-specific dequantisations were already considered for private-key quantum money~\cite{radian2022semi}, certifiable deletion of quantum encryption~\cite{hiroka2021quantum} and secure software leasing~\cite{kitagawa2021secure}. In all these cases the authors derived the desired security statement from properties of trapdoor claw-free functions, a cryptographic primitive which is also the basis of our RSP protocol. 
While this is less generic and modular than our approach and requires a new analysis for each application, it does have the advantage that one can obtain \emph{negligible} security, whereas with RSP we obtain inverse polynomial security. We comment more on the possibility of negligible security from RSP-like primitives in~\ref{sec:discussion}.

\subsection{Soundness proof for parallel RSP protocol} \label{sec:proof_sketch}
We give a short sketch of the soundness proof for our parallel RSP protocol.
The full protocol is described as \ref{fig:protocol_multi_round} (though our discussion in the introduction is restricted to \ref{fig:protocol_test}) and its soundness proof is given in \ref{sec:soundness}.
We briefly explain the difference between \ref{fig:protocol_test} and \ref{fig:protocol_multi_round}: \ref{fig:protocol_test} is a protocol to \emph{test} the prover, i.e.~in this protocol the prover is asked to prepare \emph{and measure} a quantum state, and  the verifier runs checks on the prover's answer.
The soundness statement for this protocol is a self-testing statement in the sense of~\cite{self-testing}, which characterises which states and measurements the prover used in the protocol.
Although we do not spell this out in \ref{sec:soundness}, it is easy to obtain an explicit self-testing statement from our proof.
In contrast, \ref{fig:protocol_multi_round} is a protocol for remote state preparation, so the prover is supposed to prepare, but not yet measure, a particular quantum state.
Instead, this quantum state will be used for other applications.
This means that we do not want to make a statement about how the prover measured its state, but rather what state remains in its quantum memory.
The soundness of \ref{fig:protocol_multi_round} follows from that of \ref{fig:protocol_test} via a statistical argument.
In the following, we focus on \ref{fig:protocol_test}.
We do not explain the protocol and the cryptographic primitives underlying it in detail; instead, we give a very high-level description of the relevant part of the soundness proof of the RSP protocol from~\cite{rsp} and then explain our method for proving a parallel rigidity statement based on that result.

The main cryptographic primitive underlying the RSP protocol is a so-called extended noisy trapdoor claw-free function (ENTCF) family, which can be constructed assuming the quantum hardness of LWE~\cite{lwe,mahadev}.
An ENTCF family is a family of functions indexed by a set of keys $\kg \cup \kf$.
$\kg$ and $\kf$ are disjoint sets of keys with the property that given a $k \in \kg \cup \kf$, it is computationally intractable to determine which set this key belongs to.
See~\cite[Section 4]{mahadev} for further details on ENTCF families.

In the RSP protocol from~\cite{rsp}, for a given \emph{basis choice} $\theta \in \bits$ (where ``0'' corresponds to the computational and ``1'' to the Hadamard basis), the verifier samples a key $k \in \mK_{\theta}$, alongside some trapdoor information $t$.
The verifier sends $k$ to the prover and keeps $t$ private.
The verifier and prover then interact classically; for us, the main point of interest is the last round of the protocol, i.e.~the last message from the verifier to the prover and back.
Let us denote the protocol's transcript up to the last round by $\ts$.
Before the last round, the remaining quantum state of an \emph{honest} prover is the single-qubit state $H^\theta \proj{v} H^\theta$ for $v \in \bits$.
From the transcript and the trapdoor information, the verifier can compute $v$; in contrast, the prover, who does not know the trapdoor, cannot efficiently compute $\theta$ or $v$.
In the last round, the verifier sends $\theta$ to the prover,  who returns $v' \in \bits$; the verifier then checks whether $v' = v$.
The honest prover would generate $v'$ by measuring its remaining qubit $H^\theta \proj{v} H^\theta$ in the basis $\theta$ and therefore always pass the verifier's check.

We can model this last round of the protocol (with a potentially dishonest prover) as follows:
at the start, the prover has a state $\sigma^{(\theta, v)}$, which it produced as a result of the previous rounds of the protocol.
For an honest prover, $\sigma^{(\theta, v)} = H^\theta \proj{v} H^\theta$.
Of course, this state can depend on all of $\ts$, but we only make the dependence on $\theta$ and $v$ explicit.
Upon receiving $\theta \in \bits$ the prover measures a binary observable $Z$ (if $\theta = 0$) or $X$ (if $\theta = 1$) and returns the outcome $v'$.
An honest prover would simply use the Pauli observables $Z = \sigma_Z$ and $X = \sigma_X$.
The key step in the proof of~\cite{rsp} is to show that, due to the properties of ENTCF families, for any (potentially dishonest) prover that is accepted with high probability, the observables $X$ and $Z$ must anti-commute when acting on the prover's state.
Then, \ref{thm:inf} (for $n = 1$) follows from known results~\cite{scarani-singlet, linearity, gowers_hatami}.

For our parallel RSP protocol we run $n$ independent copies of the protocol from~\cite{rsp} in parallel, except that the basis choice $\theta_i$ is the same for each copy.\footnote{The advantage of this is that a prover that succeeds with high probability on average over $\theta$ must also succeed with high probability for each $\theta$ individually. If we were to sample $\theta$ independently for each of the parallel copies we could not conclude that a prover succeeds with high probability for any particular choice of $\theta_1, \dots, \theta_n$ as there are exponentially many such choices.}
The prover's state before the last round of each copy of the RSP protocol is now denoted by $\sigma^{(\theta, \vec v)}$, where $\vec v \in \bits^n$ can be calculated by the verifier from the transcript $\ts$ by repeating the same calculation as above for each parallel copy.
Generalising from the single-qubit case, given $\theta\in\bits$ the prover performs a measurement to generate $\vec v \in \bits^n$, which we can describe by binary observables $Z_i, X_i$ (for $\theta = 0,1$ respectively) that correspond to the observable used to produce the $i$-th entry of $\vec v$.
(For an honest prover, $\sigma^{(\theta, \vec v)} = H^\theta \proj{v_1} H^\theta \ot \dots \ot H^\theta \proj{v_n} H^\theta$ and $Z_i$ is a Pauli-$Z$ measurement on the $i$-th qubit.)

The main challenge in the proof is to establish that the prover must treat all of the parallel copies of the RSP protocol independently, i.e.~to show that its (a priori uncharacterised) Hilbert space can be partitioned into $n$ identical subspaces, one for each copy of the protocol.
At first sight, it might look as though for this it suffices to show that $X_i$ and $Z_j$ (approximately) commute for all $i \neq j$.
However, this is not the case because any such commutation statement can only be shown in a special \emph{state-dependent distance}~\cite{vidick_thesis}, which does not allow us to combine individual commutation statements into the \emph{global} statement that the Hilbert space factorises into $n$ subspaces.
Instead, we need to consider the family $\{Z(\vec a) X(\vec b)\}_{\vec a, \vec b \in \bits^n}$ of $4^n$ binary observables, where $Z(\vec a) = Z_1^{a_1} \cdots Z_n^{a_n}$.
We then have to show that $\{Z(\vec a) X(\vec b)\}$ form an approximate representation of the Pauli group~\cite{gowers_hatami, thomas_fsmp}.\footnote{When we say ``Pauli group'' we always mean the Pauli group modulo complex conjugation, which is also sometimes called the Heisenberg-Weyl group.}
This means that when acting on the prover's (unknown) state $\sigma^{(\theta)}$ (where $\sigma^{(\theta)}$ is like $\sigma^{(\theta, \vec v)}$, but averaged over all $\vec v$), the operators $\{Z(\vec a) X(\vec b)\}$ behave essentially like Pauli operators. 
Formally, this means showing that on average over $\vec a, \vec b \in \bits^n$, 
\begin{align}
\tr{Z(\vec a) X(\vec b) Z(\vec a) X(\vec b) \sigma^{(\theta)}} \approx (-1)^{\vec a \cdot \vec b} \,. \label{eqn:pauli_intro}
\end{align}
This is the appropriate generalisation of the statement that $Z$ and $X$ anti-commute in the single-qubit case.
It is easy to check that \ref{eqn:pauli_intro} holds when $Z_i$ and $X_i$ are the Pauli observables.

Our proof of \ref{eqn:pauli_intro} has five main steps, which we briefly sketch here with references to the corresponding parts of the formal proof. 
\begin{itemize}[label=\arabic*., wide, labelindent=0pt,itemsep=0pt]
\item[\textbf{(1)}] Instead of working with the observables $X_i$, we define ``inefficient observables'' $\tilde X_i = (-1)^{v_i} X_i$, where $v_i$ is the $i$-th bit of the verifier's string $\vec v$ (\ref{def:observables}). $\tilde X_i$ is not an observable that an efficient prover can implement because it depends on $v_i$, which requires the trapdoor information to be computed efficiently.
Intuitively, while $X_i$ describes the prover's answer, $\tilde X_i$ describes whether that answer is accepted by the verifier.
This has the advantage that the state $\sigma^{(\theta=1)}$ (averaged over $\vec v$) of a successful prover is an approximate $+1$-eigenstate of $\tilde X_i$, but not of $X_i$ (\ref{lem:observables_succ_prob}).
\item[\textbf{(2)}] We extend the family of states $\{\sigma^{(\theta)}\}_{\theta \in \bits}$ to a larger family of ``counterfactual states'' $\{\sigma^{(\vec \theta)}\}_{\vec \theta \in \bits^n}$, which are defined as the states the prover would have prepared if the verifier had sent keys $k_i \in \mK_{\theta_i}$.
In \ref{fig:protocol_test} the basis choice is the same for all $i$, i.e.~$\vec \theta = \vec 0$ or $\theta=\vec 1$, so for other choices of $\vec \theta$ these states are never actually prepared.
However, they are still well-defined because for any prover in the actual protocol, we can fix that prover's operations (as a quantum circuit acting on a given input) and then consider what state those operations would produce if given keys with an arbitrary basis choice $\vec \theta$.
The reason these counterfactual states are useful is that we can show that, as a consequence of the properties of ENTCF families, the states $\{\sigma^{(\vec \theta)}\}_{\vec \theta}$ are computationally indistinguishable (\ref{sec:theta_ext}).
\item[\textbf{(3)}] We now want to show various commutation and anti-commutation relations for the observables $Z(\vec a)$ and $\tilde X(\vec b)$ (\ref{lem:anticomm_single_obs} and \ref{lem:anticomm_many_z}).
For example, we want to show that $Z_i$ and $\tilde X_i$ anti-commute, but $Z_i$ and $\tilde X_j$ commute (for $i \neq j$).
To show these relations, we make use of the counterfactual states $\sigma^{(\vec \theta)}$ in the following way:
for any particular relation, we can pick a $\vec \theta$ that makes showing this relation especially convenient.
For example, to show that $Z_i$ and $\tilde X_j$ commute, we would choose a $\vec \theta$ with $\theta_i = 0$ and $\theta_j = 1$ since the verifier can check the outcomes of ``$Z$-type observables'' for $\theta = 0$ and ``$X$-type observables'' for $\theta = 1$. Using the properties of ENTCF families, we can argue that the prover's measurements on these counterfactual states still yield outcomes that would pass the verifier's checks for each choice of $\theta_i$.
Based on this, we can show the desired relations for a ``convenient'' choice of counterfactual state $\sigma^{(\vec \theta)}$.
Then, we can relate these statements back to the prover's actual states $\sigma^{(\theta)}$ using the computational indistinguishability of $\{\sigma^{(\vec \theta)}\}$.
This is somewhat delicate because $\tilde X_i$ are inefficient, see \ref{lem:partial_lifting} for the precise statement.
\item[\textbf{(4)}] We can combine the various commutation and anti-commutation statements from the previous step to show that the observables $\{Z(\vec a) \tilde X(\vec b)\}$ behave like Pauli observables on $\sigma^{(\theta=1)}$, i.e.~we show \ref{eqn:pauli_intro} but with $\tilde X$ instead of $X$ (\ref{lem:approx_rep_tilde}).
This step relies on the fact that $\sigma^{(\theta=1)}$ is an approximate $+1$-eigenstate of $\tilde X(\vec b)$ for all $\vec b$; see the proof of \ref{lem:approx_rep_tilde} for details on how this simplifies the analysis.
\item[\textbf{(5)}] Since we now know that $\{Z(\vec a) \tilde X(\vec b)\}$ behave essentially like Pauli observables, we can define an explicit isometry $\tilde V$ which can be shown to map $\{Z(\vec a) \tilde X(\vec b)\}$ to the corresponding Pauli observables (\ref{lem:operator_rounding_tilde}).
This means that we have good control over these inefficient observables, and we know how the inefficient and efficient observables are related.
We can use this to define a modified isometry $V$ that maps the efficient observables $\{Z(\vec a) X(\vec b)\}$ to the corresponding Pauli observables (\ref{lem:operator_rounding_no_tilde}).
This is a stronger version of \ref{eqn:pauli_intro} and, combined again with the verifier's checks in the protocol and properties of ENTCF families, can be used to show that the prover must have prepared BB84 states (\ref{sec:bb84}).
\end{itemize}

We briefly comment on the relation between our soundness proof and that in \cite{self-testing}. 
At a high level, the soundness proof in \cite{self-testing} also shows a kind of ``parallel rigidity'' of two executions of a remote state preparation protocol.
However, their proof proceeds quite differently from ours: they first show that observables ``on the first qubit'' anti-commute, which allows them to make a partial statement about the prover's state. This in turn can be used to extend the statement about the prover's observables to two-qubit observables, which is finally used to prove a statement about the prover's two-qubit state.
This qubit-by-qubit approach is extremely costly in terms of parameters due to switching back and forth between making partial statements about the observables and state, and cannot reasonably be extended to $n$ qubits.
In contrast, we can make a global statement about the prover's $4^n$ possible observables without first characterising parts of the prover's state.
This allows us to prove a parallel rigidity statement for $n$ qubits without an exponential degradation of parameters.

\subsection{Discussion} \label{sec:discussion}

We have shown how a classical verifier can certify a tensor product of BB84 states in the memory of a quantum prover, assuming the quantum-intractability of the LWE problem. Importantly, the prover does not know which BB84 states it has prepared, whereas the verifier does. Hence, the result at the end of the protocol is as if the verifier had sent random BB84 states to the prover.
This allows us to dequantise a number of quantum cryptographic primitives, yielding a generic and modular way of translating these protocols to a setting where only classical communication is used.
We have demonstrated the versatility of this approach by applying it to unclonable encryption, quantum copy-protection, computing on encrypted data, and blind verification.
Naturally, we expect that other primitives that rely on BB84 states can also be dequantised using our approach.
Examples of this include quantum encryption with certified deletion~\cite{Broadbent_2020,Poremba22} and private key quantum money~\cite{wiesner,radian2022semi}.
We leave these and other applications to future work.

Apart from applying our technique to dequantise additional cryptographic primitives, our work raises a number of further open problems. 
Firstly, while our RSP primitive is based on the hardness of LWE, we can ask whether it is possible to achieve this functionality from weaker computational assumptions. For instance, would it be possible to perform an RSP-like protocol assuming only the existence of quantum-secure one-way functions? This is of particular interest because recent results have shown that secure two-party computation can be achieved from one-way functions and quantum communication~\cite{bartusek2021one, grilo2021oblivious}. These results are based on the fact that an oblivious-transfer protocol can be implemented from one-way functions and quantum communication that consists of BB84 states. However, an RSP primitive like ours would allow one to generically dequantise that quantum communication. Hence if RSP (with sufficiently strong parameters) can be obtained from quantum-secure one-way functions, then secure two-party computation can also be obtained from those functions, together with classical communication. In light of earlier work~\cite{impagliazzo1989limits, rudich1991use} we conjecture that this is impossible.
Formalising this intuition could lead to a better understanding of the minimum assumptions required for performing RSP-like protocols.

Secondly, a more technical open problem concerns the parameters of our rigidity theorem,~\ref{thm:inf}. As stated above, provided the prover accepts, the state the verifier certifies is $1/\poly(n)$-close to a tensor product of $n$ BB84 states (up to an isometry). The $1/\poly(n)$ closeness means that the soundness error of our dequantised protocols also scales as $1/\poly(n)$. It would be desirable to achieve \emph{negligible} soundness error, particularly when considering composable instances of these protocols.
This is not possible with the approach taken in this paper as the statistical argument in \cref{thm:final_rigidity} will necessarily introduce $1/\poly(n)$ factors.
However, it might be possible to circumvent an explicit RSP statement: the advantage of the RSP statement in our paper is that one can use it to dequantise existing protocols easily, but these existing protocols typically only use BB84 states because of their no-cloning properties.
Therefore, instead of using an RSP protocol to prepare those states, one could instead try to show a ``post-quantum cryptographic no-cloning property'' directly that could plausibly be used to dequantise these protocols while preserving negligible soundness.

Finally, we mention that our derivation of the parameters in the rigidity theorem (\ref{thm:final_rigidity}) is likely not optimal and could be optimised to improve the efficiency of our protocol. The situation here is similar to that of parallel self-testing in the multi-prover setting, with the first works having round complexity that scaled as a high-degree polynomial~\cite{reichardt2013classical} and more recent works achieving quasilinear scaling~\cite{linearity, coladangelo2019verifier}. It would be interesting to see whether ideas from these newer works are also applicable in the setting of parallel remote state preparation.

\paragraph{Acknowledgements.}
We thank Honghao Fu, Thomas Vidick, and Daochen Wang for helpful discussions, and Jeffrey Champion and John Wright for allowing us to use the results in \ref{sec:ineff_pauli}, which are based on unpublished joint work by them and the second author. We also thank Matty Hoban for pointing out a typo in an earlier draft. AG is supported by Dr.~Max R\"ossler, the Walter Haefner Foundation and the ETH Z\"urich Foundation.
TM acknowledges support from the QuantERA project ``eDict'' and the Air Force Office of Scientific Research (AFOSR) Grant No. FA9550-19-1-0202. AP
is partially supported by AFOSR YIP award number FA9550-16-1-0495 and the Institute for Quantum Information and Matter (an NSF Physics Frontiers Center; NSF Grant PHY-1733907).

During the preparation of this manuscript, we became aware of related independent work by Fu, Wang, and Zhao, who also show a parallel rigidity statement in the setting of a single computationally bounded prover. More specifically, they are able to self-test (in the sense of~\cite{self-testing}), i.e., certify the preparation and measurement of, two kinds of states: (i) $n$ parallel BB84 states subject to the restriction that only one of the qubits may be in the Hadamard basis and the remaining $n-1$ are in the computational basis, and (ii) parallel EPR pairs. They also present a dimension testing protocol based on their self-test. We refer to their paper~\cite{parallel_epr} for more details and thank them for their cooperation in publishing our respective results at the same time.

\newpage
\part{Classically-instructed parallel remote state preparation of BB84 states}
\section{Preliminaries}

We will follow the notation of \cite{self-testing}, and also make frequent use of results from the preliminaries of that paper.
We include a copy of the most frequently used notations and lemmas from that paper here for the reader's convenience.

\subsection{Notation} \label{sec:notation}
A function $n: \N \to \R_+$ is called negligible if $\lim_{\lambda \to \infty} n(\lambda) p(\lambda) = 0$ for any polynomial $p$. We use $\negl(\lambda)$ to denote an arbitrary negligible function. We use the notation $x \rand \algo S$ to denote that $x$ is being sampled from the set $\algo S$ uniformly at random. We denote the concatenation of two strings $x$ and $y$ as $x||y$.

We use $\H$ to denote an arbitrary finite-dimensional Hilbert space, and use indices to differentiate between distinct spaces. The map $\rm{Tr} : \mL(\H) \to \C$ denotes the trace, and ${\rm Tr}_B: \mL(\H_A \ot \H_B) \to \mL(\H_A)$ is the partial trace over subsystem $B$. $\pos(\H)$ denotes the set of positive semidefinite operators on $\H$, and $\mD(\H) = \{A \in \pos(\H)\,|\; \tr{A} = 1\}$ is the set of density matrices on $\H$. 

The single qubit Pauli operators are $\sigma_X = \left( \begin{smallmatrix} 0 & 1 \\ 1 & 0 \end{smallmatrix}  \right)$ and $\sigma_Z = \left( \begin{smallmatrix} 1 & 0 \\ 0 & -1 \end{smallmatrix}  \right)$. The Hadamard basis states are written as $\ket{(-)^b} = \frac{1}{\sqrt{2}} ( \ket{0} + (-1)^b \ket{1})$.

An observable on $\H$ is a Hermitian linear operator on $\H$. A binary observable is an observable that only has eigenvalues $\in \{-1, 1\}$. For a binary observable $O$ and $b \in \bits$, we denote by $O^{(b)}$ the projector onto the $(-1)^b$-eigenspace of $O$.
For any procedure which takes a quantum state as input and produces a bit (or more generally an integer) as output, e.g., by measuring the input state, we denote the probability distribution over outputs $b$ on input state $\psi$ by $\pr{b \, | \, \psi}$.

Our main protocol \ref{fig:protocol_test} will be (almost) a parallel repetition of a sub-protocol.
We make use of vector notation to denote tuples of items corresponding to the different copies of the sub-protocol.
For example, if each of the $n$ parallel sub-protocols requires a key $k_i$, we denote $\vec k = (k_1, \dots, k_n)$.
A function that takes as input a single value can be extended to input vectors in the obvious way: for example, if $f$ takes as input a single key $k$, then we write $f(\vec k)$ for the vector $(f(k_1), \dots, f(k_n))$.
We will also use $\vec 0$ and $\vec 1$ for the bitstrings consisting only of 0 and 1, respectively (and whose length will be clear from the context), and $\vec 1^i \in \bits^n$ for the bitstring whose $i$-th bit is 1 and whose remaining bits are 0.

\subsection{Extended trapdoor claw-free functions}

Our remote state preparation protocol is based on a cryptographic primitive called extended noisy trapdoor claw free function families (ENTCF families), which are defined in~\cite[Section 4]{mahadev} and can be constructed from the Learning with Errors assumption~\cite{lwe, randomness}
We use the same notation as in~\cite[Section 4]{mahadev}, with the exception that we write $\kg$ instead of $\mK_{\mG}$ and $\kf$ instead of $\mK_{\mF}$.
In addition, we also define the following functions for convenience:

\begin{definition}[Decoding maps, {\cite[Definition 2.1]{self-testing}}] \label{def:decoding_maps}
~
\begin{enumerate}
\item For a key $k \in \kg \cup \kf$, an image $y \in \mY$, a bit $b \in \bits$, and a preimage $x \in \mX$, we define $\Chk(k, y, b, x)$ to return 1 if $y \in \supp(f_{k,b}(x))$, and 0 otherwise. (This definition is as in~\cite[Definition 4.1 and 4.2]{mahadev}).
\item For a key $k \in \kg$ and a $y \in \mY$, we define $\hat{b}(k, y)$ by the condition $y \in \cup_x \supp(f_{k, \hat{b}(k, y)}(x))$. (This is well-defined because $f_{k, 0}$ and $f_{k, 1}$ form an injective pair.)
\item For a key $k \in \kg \cup \kf$ and a $y \in \mY$, we define $\hat{x}_b(k, y)$ by the condition $y \in \supp(f_{k, b}(\hat{x}_b(k, y)))$, and $\hat{x}_b(k, y) = \bot$ if $y \notin \cup_x \supp (f_{k, b}(x))$. For $k \in \kg$, we also use the shorthand  $\hat{x}(k, y)=\hat{x}_{\hat{b}(k, y)}(k, y)$.
\item For a key $k \in \kf$, a $y \in \mY$, and a $d \in \bits^w$, we define $\hat{u}(k, y, d)$ by the condition $d \cdot \left( \hat{x}_0(k, y) \oplus \hat{x}_1(k, y) \right) = \hat{u}(k, y, d)$.
\end{enumerate}
The above decoding maps applied to vector inputs are understood to act in an element-wise fashion. 
For example, for $\vec k \in \kf^{\times n}, \vec y \in \mY^{\times n}$, and $\vec d \in \bits^{w \times n}$, we denote by $\hat u(\vec k, \vec y, \vec d) \in \bits^n$ the string defined by $\left( \hat u(\vec k, \vec y, \vec d) \right)_i = \hat u(k_i, y_i, d_i)$.
\end{definition}

\subsection{Computational efficiency}

Our proof will make a computational assumption, meaning that we assume that a certain problem cannot be solved by an \emph{efficient} adversary.
To formally define what ``efficient'' means, one has to introduce a security parameter $\lambda$, which roughly speaking specifies the instance size of a problem.
We usually leave this security parameter implicit, i.e.~when we say a procedure is efficient, we mean that there is a family of procedures, one for each $\lambda$, whose complexity has the required scaling in $\lambda$.
Suppressing the explicit $\lambda$-dependence, we can give the following definition.

\begin{definition}[$\QPT$ procedure]
A procedure (with potentially quantum inputs and outputs and an implicit dependence on a security parameter $\lambda$) is called \emph{quantum polynomial-time} ($\QPT$) if there exists a (classical) polynomial-time Turing machine $M$ that, on input $1^\lambda$, outputs a description of a circuit (with a fixed gate set) that implements the procedure.
We often use the term ``efficient'' to mean ``implementable by a $\QPT$ procedure''.
\end{definition}


\subsection{Distance measures}

\begin{definition}[Norms]
Let $A \in \mL(\H)$ with singular values $\lambda_1, \dots, \lambda_n \geq 0$.
Then, the trace norm is defined as 
\begin{align*}
\norm{A}_1 = \sum_{i} \lambda_i \,,
\end{align*}
and the spectral norm is 
\begin{align*}
\norm{A}_\infty = \max \{\lambda_i\} \,.
\end{align*}
\end{definition}

\begin{definition}[Total variation distance]
Let $P, Q$ be probability distributions over a finite alphabet $\mX$.
Then, the total variation distance between $P$ and $Q$ is given by 
\begin{align*}
\TVD[P, Q] = \frac{1}{2} \sum_{x \in \mX} |P(x) - Q(x)| \,.
\end{align*}
\end{definition}

\begin{definition}[Approximate equality, {\cite[Definition~2.8 and Definition~2.14]{self-testing}}] \label{def:approx_dist}
We overload the symbol ``$\approx$'' in the following ways (leaving the dependence on the security parameter implicit in the quantities on the left):
\begin{enumerate}
\item {\bf Complex numbers:} For $a, b \in \C$ we define: 
\begin{equation}
a \approx_\eps b \iff \abs{a - b} = O(\eps) + \negl(\lambda) \,.
\end{equation}
\item {\bf Operators:} For $A, B \in \mL(\H)$, we define: 
\begin{equation}
A \approx_\eps B \iff \norm{A - B}_1^2 = O(\eps) + \negl(\lambda) \,.
\end{equation}
(We will most frequently use this for (possibly subnormalised) quantum states $A, B \in \pos(\H)$.)
\item {\bf Operators on a state:} For $A, B \in \mL(\H)$ and $\psi \in \pos(\H)$, we define: 
\begin{equation}
A \approx_{\eps, \psi} B \iff \tr{(A-B)^\dagger (A-B) \psi} = O(\eps) + \negl(\lambda) \,.
\end{equation}
\item {\bf Computationally indistinguishable states:} For two (families of not necessarily normalised) states $\psi, \psi' \in \pos(\H)$ 
which are computationally indistinguishable up to $\delta$ (i.e., no efficient distinguisher has advantage exceeding $\delta$ in distinguishing $\psi$ from $\psi'$\footnote{A distinguisher $\mA$ is a CPTP map from the input state to a classical single-qubit state (i.e.~a distribution over $\bits$). The distinguishability is the trace distance between $\mA(\psi)$ and $\mA(\psi')$. Note that this definition does not require $\psi$ and $\psi'$ to be normalised: for non-normalised states, the action of $\mA$ is still well-defined and will produce a non-normalised distribution over $\bits$ as output.}), we write:
\begin{equation}
\psi \capprox_{\delta} \psi' \,.
\end{equation}
\end{enumerate}
If we write $\approx_{0}$, we mean that the quantities are negligibly close. All asymptotic statements are understood to be in the limits $\eps \to 0$ and $\lambda \to \infty$.
It is easy to check that all the distance measures underlying the $\approx$-notation are indeed metrics and in particular satisfy the triangle inequality (see \cite[Section 2.4]{self-testing} for details).
\end{definition}

\begin{lemma}[{\cite[Lemma~2.16]{self-testing}}] \label{lem:state_dep_distance_expanded}
Let $\H_1, \H_2$ be Hilbert spaces with $\dim(\H_1) \leq \dim(\H_2)$ and $V: \H_1 \to \H_2$ an isometry. Let $A$ and $B$ be unitaries on $\H_1$ and $\H_2$, respectively, $\psi_1 \in \pos(\H_1)$, $\psi_2 \in \pos(\H_2)$, and $\eps \geq 0$. Then:
\begin{align}
\tr{V^\dagger B^\dagger V A \psi_1} \approx_{\eps} \tr{\psi_1}  \implies V^\dagger B  V  \approx_{\eps, \psi_1} A \,,\\
\tr{V A^\dagger V^\dagger B \psi_2} \approx_{\eps} \tr{\psi_2} \implies V A V^\dagger  \approx_{\eps, \psi_2} B \,.
\end{align} 
We will also frequently use this lemma in the simpler case where $V = \1$.
\end{lemma}
\begin{remark}
In \cite[Lemma~2.16]{self-testing}, this lemma is only stated for the case where $A$ and $B$ are binary observables (i.e.~are Hermitian in addition to being unitary), but it is easy to check that the proof of \cite[Lemma~2.16]{self-testing} also goes through in the more general case stated here.
\end{remark}

\begin{lemma}[Replacement lemma, {\cite[Lemma~2.21]{self-testing}}] \label{lem:replace_in_trace}
Let $\psi \in \pos(\H)$, and $A, B, C \in \mL(\H)$. If $A \approx_{\eps, \psi} B$ and $\norm{C}_\infty = O(1)$, then
\begin{align}
\tr{C A \psi} \approx_{\eps^{1/2}} \tr{C B \psi} \,, \\
\tr{A C \psi} \approx_{\eps^{1/2}} \tr{B C \psi} \,.
\end{align}
\end{lemma}

\section{Protocol description and completeness}

In this section, we introduce our protocol for classically instructed parallel remote state preparation.
The formal description of the protocol is given as \ref{fig:protocol_multi_round}.
On a high-level, \ref{fig:protocol_multi_round} works as follows: the verifier first runs a number of \emph{test rounds} (\ref{fig:protocol_test}), where the prover is asked to measure its entire quantum state.
These test rounds are used by the verifier to check whether the prover behaves as intended.
Once the verifier is convinced of this, the verifier runs a preparation round (\ref{fig:protocol_prep}).
Test and preparation rounds are indistinguishable from the point of view of the prover, except that unlike in a test round, in a preparation the prover is not asked to measure its final state. 
Instead, the final state can then be used in another protocol, as we will do in \ref{part:applications}.

The main guarantee achieved by our protocol is the following: if the prover passed the test rounds with high probability, then the verifier knows that (up to a change of basis) the prover must have prepared $n$ parallel random BB84 states.
The formal soundness theorem is given in \ref{thm:final_rigidity}.
Before we turn to the soundness proof, we show completeness, i.e., that an honest prover can succeed in our protocol with overwhelming probability.

\begin{table}[p!]
\begin{longfbox}[breakable=false, padding=1em, padding-right=1.8em, padding-top=1.2em, margin-top=1em, margin-bottom=1em, background-color=gray!20]
\begin{protocol} {\bf Test round protocol.} \label{fig:protocol_test} \end{protocol}
Let $\lambda \in \N$ be the security parameter, $(\mF, \mG)$ an ENTCF family, and $n = \poly(\lambda)$ the number of BB84 states that the verifier wishes to prepare.
\begin{enumerate}[label=\arabic*.]
 \item The verifier selects a uniformly random basis $\theta \rand \bits$, where 0 corresponds to the computational and 1 to the Hadamard basis.
\item The verifier samples keys and trapdoors $(k_1, t_{k_1}; \dots k_n, t_{k_n})$ by computing $(k_i, t_{k_i}) \leftarrow \Gen_{\mK_{\theta}}(1^\lambda)$.
The verifier then sends $(k_1, \dots k_n)$ to the prover (but keeps the trapdoors $t_{k_i}$ private).
\item The verifier receives $(y_1, \dots, y_n) \in \mY^{\times n}$ from the prover.
\item The verifier selects a round type $\in \{$preimage round, Hadamard round$\}$ uniformly at random and sends the round type to the prover. 
\begin{enumerate}
\item For a \emph{preimage round}: The verifier receives $(b_1, x_1; \; \dots b_n, x_n)$ from the prover, with $b_i \in \bits$ and $x_i \in \mX$. The verifier sets $\mathtt{flag} \leftarrow \mathtt{fail_{Pre}}$ if $\Chk(k_i, y_i, b_i, x_i) = 0$.
\item For a \emph{Hadamard round}: The verifier receives $d_1, \dots d_n \in \bits^w$ from the prover (for some $w$ depending on the security parameter). The verifier sends $q = \theta$ to the prover, and receives answers $v_1, \dots v_n \in \bits$. The verifier performs the following checks:

\begin{tabular}{p{4cm} p{8cm}}
Case & Verifier's check \\
\hline
$q = \theta = 0$ & Set $\mathtt{flag} \leftarrow \mathtt{fail_{Had}}$ if $\hat{b}(k_i, y_i) \neq v_i$ for some $i$. \\
$q = \theta = 1$ & Set $\mathtt{flag} \leftarrow \mathtt{fail_{Had}}$ if $\hat{u}(k_i, y_i, d_i) \neq v_i$. 
\end{tabular}
\end{enumerate}
\end{enumerate}
\emph{Note.} We denote the ``question'' separately by $q$ (even though here we always have $q = \theta$) because when the variant of this protocol in \ref{fig:protocol_prep} is used in the context of another cyrptographic task, the verifier can also send questions $q$ which are different from $\theta$.
\end{longfbox}
\end{table}

\begin{table}[p!]
\begin{longfbox}[breakable=false, padding=1em, padding-right=1.8em, padding-top=1.2em, margin-top=1em, margin-bottom=1em, background-color=gray!20]
\begin{protocol} {\bf Preparation round protocol.} \label{fig:protocol_prep}\end{protocol}
Let $\lambda \in \N$ be the security parameter, $(\mF, \mG)$ an ENTCF family, and $n = \poly(\lambda)$ the number of BB84 states that the verifier wishes to prepare.
\begin{enumerate}[label=\arabic*.]
\item The verifier selects bases $\vec \theta \rand \bits^n$, where 0 corresponds to the computational and 1 to the Hadamard basis.
\item The verifier samples keys and trapdoors $(k_1, t_{k_1}; \dots k_n, t_{k_n})$ by computing $(k_i, t_{k_i}) \leftarrow \Gen_{\mK_{\theta_i}}(1^\lambda)$.
The verifier then sends $(k_1, \dots k_n)$ to the prover (but keeps the trapdoors $t_{k_i}$ private).
\item The verifier receives $y_1, \dots y_n \in \mY$ from the prover.
\item The verifier sends ``Hadamard round'' to the prover as the round type.
\item The verifier receives $d_1, \dots d_n \in \bits^w$ from the prover (for some $w$ depending on the security parameter). 
The verifier computes a string $\vec v$ according to 
\begin{equation*}
v_i = 
\begin{cases}
\hat b(k_i, y_i) & \text{if } \theta_i = 0 \,, \\
\hat u(k_i, y_i, d_i) & \text {if } \theta_i = 1 \,.
\end{cases}
\end{equation*}
\end{enumerate}
\end{longfbox}
\end{table}

\begin{table}[p!]
\begin{longfbox}[breakable=false, padding=1em, padding-right=1.8em, padding-top=1.2em, margin-top=1em, margin-bottom=1em, background-color=gray!20]
\begin{protocol} {\bf Multi-round protocol for preparation of BB84 states.}\label{fig:protocol_multi_round} \end{protocol}
Let $\lambda \in \N$ be the security parameter, $(\mF, \mG)$ an ENTCF family, $n=\poly(\lambda)$ the number of BB84 states that the verifier wishes to prepare, $N = M^2$ the maximum number of test rounds (for $M \in \N)$, and $\delta$ an error tolerance parameter. For $j \in [M]$ we denote by $B_j = \{(j-1)M+1, \dots, jM\}$ the $j$-th ``block'' of $M$ rounds.
\begin{enumerate}[label=\arabic*.]
\item The verifier (privately) samples $S \rand \{0, \dots, M-1\}$ (the number of $M$-round blocks of test rounds that will be performed).
\item The verifier performs $S M$ executions of \ref{fig:protocol_test} with the prover. The verifier aborts if \emph{for any} $j \in [S]$, the fraction of rounds in $B_j$ for which $\mathtt{flag} = \mathtt{fail_{Pre}}$ or $\mathtt{flag} = \mathtt{fail_{Had}}$ exceeds $\delta$.
\item The verifier (privately) samples $R \rand [M]$ and executes \ref{fig:protocol_test} with the prover $R-1$ times. Then, the verifier executes \ref{fig:protocol_prep} with the prover and records the basis choice $\vec \theta$ and the string $\vec v$ from that execution.
\end{enumerate}
\end{longfbox}
\end{table}

\begin{proposition}\label{prop:honest-prover-correctness}
There exists an efficient quantum prover that is accepted in \ref{fig:protocol_multi_round} with probability negligibly close to 1 in the security parameters $\lambda$ (for parameter choices $n$ at most polynomial in $\lambda$ and $\delta$ at least inverse polynomial in $\lambda$).
Furthermore, the final state of such a prover at the end of \ref{fig:protocol_multi_round} is 
\begin{align}
\bigotimes_{i \in [n]} H^{\theta_i} \proj{v_i} H^{\theta_i} \,, \label{eqn:honest_prover_final_state}
\end{align}
where $\vec v$ and $\vec \theta$ are the strings recorded by the verifier in \ref{fig:protocol_multi_round}.
\end{proposition}
\begin{proof}
We describe the honest behaviour.
An honest prover behaves the same in each of the test rounds in step 2.~of \ref{fig:protocol_multi_round}.
Hence, to show that an honest prover succeeds in \ref{fig:protocol_multi_round} with probability negligibly close to 1, it suffices to describe an honest strategy for \ref{fig:protocol_test} that succeeds with probability negligibly close to 1.\footnote{In fact, it would suffice to have an honest strategy for \ref{fig:protocol_test} that succeeds with probability $1 - \delta/2$. Then, as long as we choose $N = \Omega(\lambda^c)$ for some $c > 0$, a standard application of Hoeffding's inequality shows success in \ref{fig:protocol_multi_round} with probability negligibly close to 1.}
In \ref{fig:protocol_test}, the prover receives $n$ keys $k_1, \dots, k_n$ and returns answers for each key $k_i$ individually. 
Since the verifier's checks are independent for each $i$, we only need to describe an honest procedure for one key $k_i$ that succeeds in the verifier's checks for that $i$ with probability negligibly close to~1.

The honest behaviour for a single key $k_i$ is essentially the same as in \cite[Lemma 3.4]{rsp} and \cite[Proposition 3.1]{self-testing}.
For the sake of completeness, we spell out the details, but note that most of the proof is simply copied from \cite[Proposition 3.1]{self-testing}.

Given a key $k_i \in \mK_{\theta}$, the prover prepares the state 
\begin{align*}
    \frac{1}{\sqrt{2 \cdot \abs{\mX}}} \sum_{b \in \bits} \sum_{x \in \mX, \, y \in \mY} \sqrt{f_{k_i, b}(x)(y)} \ket{b} \ket{x} \ket{y} \,.
\end{align*}

Preparing this state can be efficiently done (up to negligible error) using the $\Samp$ procedure from the definition of ENTCF families (\cite[Definition 3.1]{randomness} and \cite[Definition 4.2]{mahadev}). The prover then measures the ``image register'' (i.e., the register that store $y$) to obtain an image $y_i \in \mY$ and sends this back to the verifier. The post-measurement for each $i$ is 
\begin{align}
\begin{cases}
\ket{\hat{b}(k_i, y_i)} \ket{\hat{x}(k_i, y_i)} & \text{if $k_i \in \kg$}\,, \\
\frac{1}{\sqrt{2}} \left( \ket{0} \ket{\hat{x}_0(k_i, y_i)} + \ket{1} \ket{\hat{x}_1(k_i, y_i)}  \right) & \text{if $k_i \in \kf$}\,.
\end{cases}
\label{eqn:hones_post_y_states}
\end{align}

If the verifier selects a preimage round, the prover measures both registers in the computational basis and returns the result. From the states in \ref{eqn:hones_post_y_states} it is clear that the prover succeeds with probability negligibly close to 1 in the preimage round.

If the verifier selects a Hadamard round, the prover measures the ``$x$-register'' in the Hadamard basis to obtain $d_i$ and returns this to the verifier. 
We introduce the shorthand $b_i = \hat{b}(k_i, y_i)$ and $u_i = \hat u(k_i, y_i, d_i)$.
At this point, the prover's state for each $i$ is (up to a global phase)
\begin{align}
\begin{cases}
\ket{b_i}  & \text{if $k_i \in \kg$} \,, \\
\ket{(-)^{u_i}}  & \text{if $k_i \in \kf$} \,.
\end{cases} \label{eqn:final_bit_states}
\end{align}
The prover now receives a question $q = \theta$ and measures the remaining qubit in the computational basis if $q = 0$ and in the Hadamard basis if $q_i = 1$.
Then it is clear from the expression for the prover's remaining qubit in \ref{eqn:final_bit_states} that the prover will pass the verifier's check.

Having described the honest behaviour for a test round, we now turn our attention to showing \ref{eqn:honest_prover_final_state}.
In a preparation round, the prover receives keys $k_1, \dots, k_n$ with $k_i \in \mK_{\theta_i}$.
For each key $k_i$, an honest prover performs the same steps as described above for a test round.\footnote{In fact, due to the injective invariance property of ENTCF functions, the prover cannot distinguish keys in $\kg$ from those in $\kf$. Hence, any (potentially dishonest) prover does not know whether it is in a test or preparation round, so it must perform the same actions in both. This will play a central role in our soundness proof in \ref{sec:soundness}.}
Then, \ref{eqn:honest_prover_final_state} follows immediately from \ref{eqn:final_bit_states}.
\end{proof}

\section{Protocol soundness} \label{sec:soundness}

The purpose of this section is to prove that a prover that succeeds in \ref{fig:protocol_multi_round} must have prepared a product of BB84 states without knowing which BB84 states it has prepared, which is formally stated as \ref{thm:final_rigidity}.
This means that the situation at the end of \ref{fig:protocol_multi_round} is essentially identical to one in which the verifier has sent random BB84 states to the prover via a quantum channel.
We will make use of this in \ref{part:applications}.

The structure of the proof is as follows: a prover in \ref{fig:protocol_multi_round} is subjected to multiple test rounds, i.e., executions of \ref{fig:protocol_test}, by the verifier.
We can model the behaviour of an arbitrary prover in a single test round as a \emph{device}, defined in \ref{sec:devices}.
The bulk of the soundness proof is concerned with showing that any such device that has a high success probability in \ref{fig:protocol_test} must have prepared (and subsequently measured) a product of BB84 states (\ref{lem:bb84_same_alpha}).
For this, we first show that certain ``inefficient observables'' associated with a device approximately satisfy the same commutation and anti-commutation relations as Pauli observables (\ref{lem:approx_rep_tilde}).
This allows us to conclude that these inefficient observables can be mapped to Pauli observables by an isometry (\ref{lem:operator_rounding_tilde}).\footnote{This can be viewed as a specific instantiation of a general stability theorem for approximate group representations due to Gowers and Hatami~\cite{gowers_hatami}.}
Finally, since we now know a lot about the structure of these inefficient observables, we can  relate them to the device's actual (efficient) observables in a way that allows us to conclude that the efficient observables can be mapped to Pauli operators, too, albeit by a slightly different isometry (\ref{lem:operator_rounding_no_tilde}).

\subsection{Modelling devices in \ref{fig:protocol_test}} \label{sec:devices}

\begin{definition}[Devices] \label{def:devices}
A device $D = (S, \Pi, M, P)$ is specified by the following:
\begin{enumerate}
\item A set $S = \{ \psi^{(\vec \theta)} \}_{\vec \theta \in \bits^n}$ of states $\psi^{(\vec \theta)} \in \mD(\H_D \ot \H_Y)$, where $\dim(\H_Y) = |\mY|^{n}$ and the states are classical on $\H_Y$:
\begin{equation}
\psi^{(\vec \theta)} = \sum_{\vec y \in \mY^n} \psi^{(\vec \theta)}_{\vec y} \ot \proj{\vec y}_Y \,.
\end{equation}
In the context of \ref{fig:protocol_test}, $\psi^{(\vec \theta)}$ is the prover's state after returning $\vec y$ for the case where the verifier makes basis choices $\vec \theta$.\footnote{In \ref{fig:protocol_test}, the only two basis choices are $\vec \theta = \vec 0$ and $\vec \theta = \vec 1$. However, $\psi^{(\vec \theta)}$ is still well-defined as the state that the prover (who is defined in terms of the quantum circuits he runs on a given input) would prepare if given keys of basis choice $\vec \theta$, even though this never occurs in \ref{fig:protocol_test}. The  reason for including  all basis choices $\vec \theta$ is that in \ref{fig:protocol_prep}, the verifier makes a uniformly random basis choice $\vec \theta$, prompting the prover to use the state $\psi^{(\vec \theta)}$.} Each $\psi^{(\vec \theta)}$ also implicitly depends on the specific keys chosen by the verifier (not just the basis choice $\vec \theta$); all the statements we make hold on average over key choices (for a fixed basis choice $\vec \theta$).
\item A projective measurement $\Pi$ on $\H_D \ot \H_Y$:
\begin{equation}
\Pi = \left\{ \Pi^{(\vec b, \vec x)} = \sum_{\vec y} \Pi^{(\vec b, \vec x)}_{\vec y} \ot \proj{\vec y}_{Y} \right\}_{\vec b \in \bits^n; \; \vec x \in \mX^n} \,.
\end{equation}
This is the measurement used by the prover to compute his answer $(\vec b, \vec x)$ in the preimage challenge.
\item A projective measurement $M$ on $\H_D \ot \H_Y$:
\begin{equation}
M = \left\{ M^{(\vec d)} = \sum_{\vec y} M^{(\vec d)}_{\vec y} \ot \proj{\vec y}_{Y} \right\}_{\vec d \in \bits^{w\times n}} \,.
\end{equation}
This is the measurement used by the prover to compute his answer $\vec d$ in the Hadamard challenge.
We use an additional Hilbert spaces $\H_R$ to record the outcomes of measuring $M$ and write the post-measurement state after applying $M$ to $\psi^{(\vec \theta)}$ as 
\begin{equation}
\sigma^{(\vec \theta)} \deq \sum_{\vec y, \vec d} M^{(\vec d)}_{\vec y} \psi^{(\vec \theta)}_{\vec y} M^{(\vec d)}_{\vec y} \ot \proj{\vec y, \vec d}_{YR} \,. \label{eqn:def_sigma}
\end{equation}
\item A set $P = \{P_q\}$, where for each $q \in \bits$, $P_{q}$ is a projective measurement on $\H_D \ot \H_Y \ot \H_R$:
\begin{equation}
P_{q} = \left\{P^{(\vec v)}_{q} = \sum_{\vec y, \vec d} P^{(\vec v)}_{q, \vec y, \vec d} \ot \proj{\vec y, \vec d}_{YR} \right\}_{\vec v \in \bits^n}\,.
\end{equation}
In the context of \ref{fig:protocol_test}, given question $q$, the prover will measure $\{P_q^{(\vec v)}\}$ and return the outcome $\vec v$ as his answer.
\end{enumerate}
\end{definition}

\begin{definition}[Efficient devices] \label{def:efficient_dev}
A device is called \emph{efficient} if the states $\psi^{(\vec \theta)}$ can be prepared efficiently and the measurements $\Pi$, $M$, and $P_{q}$ can be performed efficiently.
\end{definition}

\begin{remark} \label{rem:advice_state}
Throughout this paper, whenever we say that a procedure is ``efficient'' (e.g.~in \ref{def:efficient_dev}, or in proofs when constructing adversaries that break a property of ENTCF families) we mean a quantum polynomial-time procedure \emph{with advice}, i.e.~in addition to the input, the procedure also has access to a quantum state (the ``advice state'') which only depends on the input length, but not the input itself. 
This advice state need not be efficiently preparable itself.
The additional advice state plays no role in the single-round soundness analysis (i.e.~the entire soundness proof up to \ref{thm:final_rigidity}) and hence we do not write it explicitly.
The reason the advice state is required at all is that if we want to model a prover in \ref{fig:protocol_multi_round} as a sequence of (single-round) devices of the kind in \ref{def:devices}, we need to account for the fact that the prover in round $i+1$ has access to the final state after round $i$, which we can view as an advice state for the device modelling the $(i+1)$-th round of the prover.
We note that the same convention is (implicitly) used in \cite{randomness, rsp}, and that as a result these works and our paper require that the LWE problem is hard to solve quantumly with quantum advice (see \cite[Definition 2.5]{randomness} for the precise statement).
\end{remark}

\begin{definition}[Observables] \label{def:observables}
For a device $D = (S, \Pi, M, P)$ with projective measurements as in \ref{def:devices}, we define the following binary observables:
\begin{align}
Z_i &= \sum_{\vec v} (-1)^{v_i} P_0^{(\vec v)} \,,\\
X_i &= \sum_{\vec v} (-1)^{v_i} P_1^{(\vec v)} \,,\\
\tilde X_i &= \sum_{\vec v, \vec y, \vec d} (-1)^{v_i + \hat u(k_i, y_i, d_i)}P^{(\vec v)}_{1, \vec y, \vec d} \ot \proj{\vec y, \vec d}_{YR} \,.
\end{align}
We further use the following notation for products of observables: for $\vec a \in \bits^n$, we define
\begin{equation}
Z(\vec a) \deq Z_1^{a_1} \dots Z_n^{a_n} = \sum_{\vec v} (-1)^{\vec a \cdot \vec v} P_0^{(\vec v)} \,,
\end{equation}
and likewise for $X(\vec a)$ and $\tilde X(\vec a)$. 
Expanding $X_i = \sum_{\vec y, \vec d} X_{i, \vec y, \vec d} \ot \proj{\vec y, \vec d}$ (and analogously for the other observables), it directly follows that
\begin{align}
\tilde X_{i, \vec y, \vec d} = (-1)^{\hat u(k_i, y_i, d_i)} X_{i, \vec y, \vec d} \quad\tand\quad
\tilde X(\vec a)_{\vec y, \vec d} = (-1)^{\vec a \cdot \hat u(\vec k, \vec y, \vec d)} X(\vec a)_{\vec y, \vec d} \,.
\end{align}
\end{definition}

\begin{remark}
It is easy to check that the operators defined above are indeed binary observables using that $\{P^{(\vec v)}_{q, \vec y, \vec d}\}_{\vec v}$ are projective measurements.
\end{remark}

\begin{definition}[Partial post-measurement states] \label{def:partial_postmeas_states}
For $k \in \kg \cup \kf$ and $v \in\bits$ define the set $V_{k, v} \subseteq \mY \times \bits^w$ by the following condition: 
\begin{equation}
(y, d) \in V_{k, v} \iff \begin{cases}
\hat b(k, y) = v  & \textnormal{if } k \in \kg \,, \\
\hat u(k, y, d) = v & \textnormal{if } k \in \kf \,.
\end{cases}
\end{equation}
Then for $\vec k, \vec \theta, \vec v$ define 
\begin{equation}
\sigma^{(\vec \theta, \vec v)} = \sum_{y_1, d_1 \in V_{k_1, v_1}} \cdots \sum_{y_n, d_n \in V_{k_n, v_n}} \sigma^{(\vec \theta)}_{\vec y, \vec d} \ot \proj{\vec y, \vec d} \,.
\end{equation}
Further for $\vec a \in \bits^n$ we define 
\begin{equation}
\sigma^{(\vec \theta, v, \vec a)} = \sum_{\vec v : \, \vec v \cdot \vec a = v} \sigma^{(\vec \theta, \vec v)} \,.
\end{equation}
\end{definition}

To get an intuition for this definition, consider the case $\theta = \vec 0$.
Then for any $\vec a \in \bits^n$, $\sigma^{(\vec 0, v, \vec a)}$ is that part of the state $\sigma^{(\vec 0)}$ for which the honest device would receive outcome $v$ when measuring the observable $Z(\vec a)$.

\begin{definition}[Success probabilities] \label{def:succ_prob}
For any device $D = (S, \Pi, M, P)$ we define $\gamma_P(D)$ as the device's failure probability in a preimage round, and $\gamma_H(D)$ as the failure probability in a Hadamard round in \ref{fig:protocol_test}: 
\begin{align}
\gamma_P(D) &= \pr{\mathtt{flag} = \mathtt{fail_{Pre}}  \,|\;\text{round type = preimage round}} \,,\\
\gamma_H(D) &= \pr{\mathtt{flag} = \mathtt{fail_{Had}} \,|\; \text{round type = Hadamard round}} \,.
\end{align}
\end{definition}

The following lemma tells us that for any device that has a non-negligible failure probability in the preimage test, there is another device that is ``close'' to the original one in the sense that its measurements are the same as for the original device and its states only differ by $O(\gamma_P(D))$, but that has a negligible failure probability in the preimage test.
A device with a negligible failure probability in the preimage test is called \emph{perfect}.
This means that for the soundness proof, it suffices to only consider perfect devices:
for any arbitrary device, we can first switch to the corresponding perfect device at the cost of incurring an approximation error of $O(\gamma_P(D))$, and then apply our soundness proof to the perfect device.
We omit the proof of \ref{lem:reduction_to_perfect} as it is essentially identical to that of \cite[Lemma 4.13]{self-testing}.\footnote{To avoid confusion, we point out that \cite[Lemma 4.13]{self-testing} phrases the approximation in terms of trace distance, but \ref{lem:reduction_to_perfect} uses the $\approx$-notation for states, which corresponds to the square of the trace distance (see \ref{def:approx_dist}). This is why we have $\gamma_P(D)$ in \ref{lem:reduction_to_perfect}, not $\gamma_P(D)^{1/2}$ as in \cite[Lemma 4.13]{self-testing}.}

\begin{definition}[Perfect device]
We call a device $D$ \emph{perfect} if $\gamma_P(D) = \negl(\lambda)$.
\end{definition}

\begin{lemma} \label{lem:reduction_to_perfect}
Let $D = (S, \Pi, M, P)$ be an efficient device with $\gamma_P(D) < 1$, where $S = \left\{ \psi^{(\vec \theta)} \right\}$. Then there exists an efficient \emph{perfect} device $D' = (S', \Pi, M, P)$, which uses the same measurements $\Pi, M, P$ and whose states $S' = \left\{ \psi'^{(\vec \theta)} \right\}$ satisfy for any $\vec \theta \in \bits^n$:
\begin{equation}
\psi'^{(\vec \theta)} \approx_{\gamma_P(D)} \psi^{(\vec \theta)} \,.
\end{equation}
\end{lemma}

The following lemma shows what outcomes a successful device must produce when measuring the observables from \ref{def:observables} on the partial post-measurement states from \ref{def:partial_postmeas_states}.

\begin{lemma} \label{lem:observables_succ_prob}
For any efficient device $D = (S, \Pi, M, P)$, string $\vec a \in \bits^n$, and bit $v \in \bits$:
\begin{align}
\tr{Z(\vec a)^{(v)} \sigma^{(\vec 0, v, \vec a)}} &\approx_{\gamma_H(D)} \tr{\sigma^{(\vec 0, v, \vec a)}} \,, \label{eqn:z_vec_succprob}\\
\tr{X(\vec a)^{(v)} \sigma^{(\vec 1, v, \vec a)}} &\approx_{\gamma_H(D)} \tr{\sigma^{(\vec 1, v, \vec a)}} \,, \label{eqn:x_vec_succprob}\\
\tr{\tilde X(\vec a) \sigma^{(\vec 1)}} &\approx_{\gamma_H(D)} 1 \label{eqn:xtilde_vec_succprob}\,.
\end{align}
\end{lemma}

\begin{proof}
We first prove \ref{eqn:z_vec_succprob}.
Since the case $q = \theta = 0$ occurs with probability 1/2 in the protocol in \ref{fig:protocol_test}, the device's failure probability in this case can be at most $2 \gamma_H(D)$.
Furthermore, since the device only succeeds if $v_i = \hat b(k_i, y_i)$ for all $i \in [n]$ in the protocol, it is in particular the case that for any $\vec a \in \bits^n$, $\vec a \cdot \hat b(\vec k, \vec y) = \vec a \cdot \vec v$ with probability at least $1 - 2 \gamma_H(D)$.
Now comparing the definition of $\sigma^{(\vec 0, v, \vec a)}$ with the verifier's checks in the protocol, this means that 
\begin{align}
\sum_{v' \in \bits} \tr{Z(\vec a)^{(v')} \sigma^{(\vec 0, v', \vec a)}} \geq 1 - 2 \gamma_H(D) \,. \label{eqn:z_sum_succprob}
\end{align}
Since $Z(\vec a)^{(v)}$ is a projector and $\sigma^{(\vec 0)} = \sigma^{(\vec 0, v, \vec a)} + \sigma^{(\vec 0, v \oplus 1, \vec a)}$ is a normalised state, 
\begin{align}
\tr{Z(\vec a)^{(v \oplus 1)} \sigma^{(\vec 0, v \oplus 1, \vec a)}} \leq \tr{\sigma^{(\vec 0, v \oplus 1, \vec a)}} = 1 - \tr{\sigma^{(\vec 0, v, \vec a)}} \,.
\end{align}
Inserting this into \ref{eqn:z_sum_succprob}, we obtain 
\begin{align}
1 - 2 \gamma_H(D) \leq \tr{Z(\vec a)^{(v)} \sigma^{(\vec 0, v, \vec a)}} + 1 - \tr{\sigma^{(\vec 0, v, \vec a)}} \,,
\end{align}
which implies 
\begin{align}
\tr{Z(\vec a)^{(v)} \sigma^{(\vec 0, v, \vec a)}} \geq \tr{\sigma^{(\vec 0, v, \vec a)}} - 2 \gamma_H(D) \,.
\end{align}
For the inequality in the other direction, we note that since $Z(\vec a)^{(v)}$ is a projector, we immediately have $\tr{Z(\vec a)^{(v)} \sigma^{(\vec 0, v, \vec a)}} \leq \tr{\sigma^{(\vec 0, v, \vec a)}}$, finishing the proof of \ref{eqn:z_vec_succprob}.

The proof of \ref{eqn:x_vec_succprob} is completely analogous.
Finally, we need to show \ref{eqn:xtilde_vec_succprob}.
For this, observe that by definition $\tilde X(\vec a)$ returns outcome 0 if the device's answers satisfy $\vec a \cdot \vec v = \vec a \cdot \hat u(\vec k, \vec y, \vec d)$.
Since this is implied by the checks performed by the verifier in the protocol, we have 
\begin{align}
\tr{\tilde X(\vec a)^{(0)} \sigma^{(\vec 1)}} \approx_{\gamma_H(D)} 1 \,.
\end{align}
Since $\tilde X(\vec a)^{(0)} + \tilde X(\vec a)^{(1)} = \1$, this implies $\tr{\tilde X(\vec a)^{(1)} \sigma^{(\vec 1)}} \approx_{\gamma_H(D)} 0$. The lemma follows because $\tilde X(\vec a) = \tilde X(\vec a)^{(0)} - \tilde X(\vec a)^{(1)}$.
\end{proof}

\subsection{Extending \ref{lem:observables_succ_prob} to different basis choices} \label{sec:theta_ext}

In \cite{mahadev}, it was shown that keys sampled from $\kg$ and $\kf$ (according to the corresponding sampling procedures, see \cite[Definition 4.3]{mahadev} for details) are computationally indistinguishable (up to negligible advantage).
This property is called \emph{injective invariance}.
In this section, we will use this property to extend the statement of \ref{lem:observables_succ_prob} to basis choices $\vec \theta$ other than $\vec \theta = \vec 0$ and $\vec \theta = \vec 1$.
For this, we need a slight strengthening of injective invariance to multiple keys, which follows almost immediately from the original injective invariance property.
\begin{lemma}\label{lem:injective_invariance}
For any two $\vec \theta, \vec \theta' \in \bits^n$, no efficient distinguisher has non-negligible advantage in the following distinguishing game:
the challenger samples a bit $b \rand \bits$. The challenger then samples keys $\vec k$ with $k_i \in \mK_{\theta_i}$ if $b = 0$ and $k_i \in \mK_{\theta'_i}$ if $b = 1$, and sends $\vec k$ to the distinguisher. The distinguisher provides a guess for the bit $b$.
\end{lemma}
\begin{proof}
First suppose that $\vec \theta$ and $\vec \theta'$ only differ in one position $\theta_i \neq \theta'_i$.
For convenience, we can assume that $\theta_i = 0$ and $\theta'_i = 1$ (and swap the labelling of $\vec \theta$ and $\vec \theta'$ if this is not the case).
Consider an efficient distinguisher $\mD$ for the distinguishing game in the lemma.
From such a distinguisher $\mD$ we can construct a distinguisher $\mD'$ for the original injective invariance game: $\mD'$ is given a key $k_i \in \kg \cup \kf$, samples $n-1$ further keys $k_j \in \mK_{\theta_j}$ for $j \in [n] \setminus \{i\}$, and runs $\mD$ on the keys $\vec k = (k_1, \dots, k_n)$.
If $\mD$ guesses $b = 0$, then $\mD'$ guesses $k \in \kg$, otherwise $\mD'$ guesses $k \in \kf$.
From the construction it is clear that $\mD'$ guesses correctly if and only if $\mD$ guesses correctly.
However, from the injective invariance property of the ENTCF family, no such $\mD'$ can succeed with non-negligible advantage. Hence, $\mD$ also cannot succeed with non-negligible advantage in the distinguishing game from the lemma.

We can extend this proof to any $\vec \theta, \vec \theta' \in \bits^n$ by a simple hybrid argument: for any $\vec \theta, \vec \theta' \in \bits^n$, we can find $\vec \theta^1, \dots, \vec \theta^n$ with $\vec \theta^1 = \vec \theta$ and $\vec \theta^n = \vec \theta'$ such that $\vec \theta^i$ and $\vec \theta^{i+1}$ differ on at most one position.
Applying the above argument in each step of the hybrid argument, the advantage in distinguishing $\vec \theta$ from $\vec \theta'$ can be at most $n \cdot \negl(\lambda)$, which is still negligible in $\lambda$ since $n = \poly(\lambda)$.
\end{proof}

\begin{lemma} \label{lem:states_indist}
Consider an efficient device $D = (S, \Pi, M, P)$. 
Then, for any $\vec \theta, \vec \theta'$: 
\begin{align}
\sigma^{(\vec \theta)} \capprox_0 \sigma^{(\vec \theta')} \,.
\end{align}
\end{lemma}

\begin{proof}
Since the device $D$ (and hence the preparation of the states $\sigma^{(\vec \theta)}$) is efficient, this follows immediately from \ref{lem:injective_invariance}.
\end{proof}

The fact that the states $\sigma^{(\vec \theta)}$ are computationally indistinguishable allows us to extend relations between operators that hold on a particular $\sigma^{(\vec \theta)}$ to other states $\sigma^{(\vec \theta')}$.
This will be formalised by \ref{lem:partial_lifting}.
The reason that the lemma does not just deal with efficient observables, but also observables which are efficient only with access to some of the trapdoors, is that the observables $\tilde X(\vec a)$ (as defined in \ref{def:observables}) are only efficient if one has access to all trapdoors $t_i$ for $i \sth a_i = 1$.
Hence, the following lemma will also be applicable to operator relations involving $\tilde X(\vec a)$ (provided the constraint on $\vec \theta$ and $\vec \theta'$ stated in the lemma is satisfied).
\begin{lemma} \label{lem:partial_lifting}
Consider an efficient device $D = (S, \Pi, M, P)$ and a subset of indices $\mI \subset [n]$. Suppose that for any keys $\vec k = (k_1, \dots, k_n)$ with corresponding trapdoors $\vec t = (t_1, \dots, t_n)$, $A$ and $B$ are observables which can be efficiently implemented with access to the trapdoors $\{t_i\}_{i \in \mI}$. 
Then for any $\vec \theta, \vec \theta' \in \bits^n$ such that $\theta_i = \theta'_i$ for all $i \in \mI$, 
\begin{align}
A \approx_{\eps, \sigma^{(\vec \theta)}} B \iff A \approx_{\eps, \sigma^{(\vec \theta')}} B \,, \label{eqn:partial_lifting_result}
\end{align}
and 
\begin{align}
\tr{A \sigma^{(\vec \theta)}} \approx_0 \tr{A \sigma^{(\vec \theta')}} \,.  \label{eqn:partial_lifting_in_trace}
\end{align}
\end{lemma}
\begin{proof}
We first prove \ref{eqn:partial_lifting_result}.
Fix $\vec \theta, \vec \theta' \in \bits^n$.
By \ref{def:approx_dist}, it suffices to show that 
\begin{align}
\tr{(A - B)^\dagger (A-B) \sigma^{(\vec \theta)}} \approx_{0} \tr{(A - B)^\dagger (A-B) \sigma^{(\vec \theta')}} \,.
\end{align}
Suppose for the sake of contradiction that there exists a positive non-negligible function $\mu(\lambda)$ such that 
\begin{align}
\tr{(A - B)^\dagger (A-B) \sigma^{(\vec \theta)}} - \tr{(A - B)^\dagger (A-B) \sigma^{(\vec \theta')}} > \mu(\lambda) \,.
\end{align}
(For the other case where this difference is smaller than a negative non-negligible function the proof is analogous.)
We want to show that this allows us to break the extended injective invariance property of \ref{lem:injective_invariance}.
For this, consider the injective invariance game of \ref{lem:injective_invariance} on the indices $[n] \setminus \mI$: a challenger chooses $b \rand \bits$ samples a set of keys $\{k_i\}_{i \in [n] \setminus \mI}$, where $k_i \in \mK_{\theta_i}$ if $b = 0$ and $k_i \in \mK_{\theta'_i}$ if $b = 1$.
The keys $\{k_i\}_{i \in [n] \setminus \mI}$ are sent to a distinguisher $\mD$, who is asked to provide a guess for $b$.

The distinguisher $\mD$ proceeds as follows: $\mD$ samples keys $k_j \in \mK_{\theta_j}$ for $j \in \mI$ along with the corresponding trapdoors $t_j$.
$\mD$ then runs the device's operations on the keys $\vec k = (k_1, \dots, k_n)$ and will obtain a state $\sigma$, with $\sigma = \sigma^{(\vec \theta)}$ if $b = 0$ and $\sigma = \sigma^{(\vec \theta')}$ if $b=1$.

By \cite[Lemma 2.6]{self-testing}, there exists an efficient procedure that given the state $\sigma$ and the trapdoors $\{t_j\}_{j \in \mI}$ (so that for this procedure, $A$ and $B$ are efficient binary observables) produces a bit $b'$ with 
\begin{align}
\pr{b'=0 \, | \; \sigma} = \frac{1}{4} \tr{(A - B)^\dagger (A-B) \sigma} \,.
\end{align}
$\mD$ runs this procedure on the state $\sigma$ it produced using the device's operations and outputs $b'$ as its guess for $b$.

It is clear that $\mD$ is efficient. Its distinguishing advantage is given by 
\begin{align}
& \pr{\text{$\mD$ guesses 0} \, | \; b = 0} - \pr{\text{$\mD$ guesses 0} \, | \; b = 1} \\
&= \pr{b' = 0 \, | \; b = 0} - \pr{b' = 0 \, | \; b = 1} \\
&= \pr{b'=0 \, | \; \sigma^{(\vec \theta)}} - \pr{b'=0 \, | \; \sigma^{(\vec \theta')}} \\
&= \frac{1}{4} \tr{(A - B)^\dagger (A-B) \sigma^{(\vec \theta)}} - \frac{1}{4} \tr{(A - B)^\dagger (A-B) \sigma^{(\vec \theta')}} \\
& > \frac{1}{4} \mu(\lambda) \,,
\end{align}
which is non-negligible by assumption, yielding the desired contradiction.

For \ref{eqn:partial_lifting_in_trace}, the proof is analogous, except that instead of using the procedure from \cite[Lemma 2.6]{self-testing} to estimate $\tr{(A - B)^\dagger (A-B) \sigma}$, in this case the distinguisher can simply measure the observable $A$ on the state it has produced by running the device's operations.
\end{proof}

We can apply this lemma to extend \ref{eqn:xtilde_vec_succprob} to states other than $\sigma^{(\vec 1)}$:
\begin{lemma} \label{lem:xtilde_succprob_allstates}
For any efficient device $D = (S, \Pi, M, P)$ and strings $\vec a, \vec \theta \in \bits^n$ such that $a_i = 1 \implies \theta_i = 1$:
\begin{align}
\tr{\tilde X(\vec a) \sigma^{(\vec \theta)}} &\approx_{\gamma_H(D)} 1 \,.
\end{align}
\end{lemma}

\begin{proof}
We define $\mI = \{i \in [n] \,|\; a_i = 1\}$.
The observable $\tilde X(a)$ is efficient if one has access to the trapdoors $\{t_i\}_{i \in \mI}$.
Since further $\theta_i = 1$ for all $i \in \mI$ by assumption, the lemma follows from \ref{eqn:xtilde_vec_succprob} by means of \ref{lem:partial_lifting}.
\end{proof}

Using \ref{lem:state_dep_distance_expanded}, we can immediately obtain the following corollary:
\begin{corollary} \label{lem:xtilde_is_one}
For any efficient device $D = (S, \Pi, M, P)$ and strings $\vec a, \vec \theta \in \bits^n$ such that $a_i = 1 \implies \theta_i = 1$:
\begin{align}
\tilde X(\vec a) \approx_{\gamma_H(D), \sigma^{(\vec \theta)}} \1 \,.
\end{align}
\end{corollary}

We still need to extend \ref{eqn:z_vec_succprob} and \ref{eqn:x_vec_succprob} to different choices of $\theta$.
For this, we need to strengthen the statement of \ref{lem:states_indist} and show that certain partial post-measurement states are also indistinguishable.

\begin{lemma} \label{lem:partial_postmeas_indist}
For an efficient device $D = (S, \Pi, M, P)$ and strings $\vec b, \vec \theta^0, \vec \theta^1 \in \bits^n$ such that $b_i = 1 \implies \theta_i^0 = \theta_i^1$ the following holds for any $v \in \bits$:
\begin{align}
\sigma^{(\vec \theta^0, v, \vec b)} \capprox_{0} \sigma^{(\vec \theta^1, v, \vec b)} \,. \label{eqn:diff_sigmas_indist}
\end{align}
\end{lemma}

\begin{proof}
Fix $v$ and $\vec b$.
As a first step, we note that it suffices to show the lemma for the case where $\vec \theta^0$ and $\vec \theta^1$ only differ on one bit by the same argument as in the proof of \ref{lem:injective_invariance}.

Hence consider $\vec \theta^0, \vec \theta^1$ that only differ on one bit $i$ s.t.~$b_i = 0$.
Without loss of generality, we can assume $\theta_i^0 = 0$ and $\theta_i^1 = 1$.
Further consider an efficient distinguisher $\mD$ that, given $\sigma^{(\vec \theta^c, v, \vec b)}$, guesses $c \in \bits$.
Let $\Delta^\mD$ be the distinguishing advantage of $\mD$.

Using $\mD$, we construct the following adversary $\mA$ for the injective invariance game.
$\mA$ is given a uniformly random key $k \in \kg \cup \kf$ and asked to guess $\theta^* \in \bits$ such that $k \in \mK_{\theta^*}$.
$\mA$ sets $k_i \deq k$ and samples keys $k_j \in \mK_{\theta_j}$ for $j \neq i$ (along with the corresponding trapdoors $t_j$). 
Then, $\mA$ performs the device's operations with these keys, obtaining $\vec y$ and $\vec d$ in the process.
Since $\mA$ knows $t_j$ for all $j \in [n]$ for which $b_j = 1$, it can efficiently compute 
\begin{align}
v' = \sum_{j \sth \theta_j = 0} b_j \, \hat b(k_j, y_j) + \sum_{j \sth \theta_j = 1} b_j \, \hat u(k_j, y_j, d_j) \mod 2\,.
\end{align}
If $v' \neq v$, $\mA$ submits a uniformly random guess for $\theta^*$.
If $v' = v$, $\mA$ runs the distinguisher $\mD$ on the post-measurement state.
$\mD$ outputs $c' \in \bits$, and $\mA$ guesses $\theta^*=c'$.

Let $\Delta^\mA$ be the distinguishing advantage of $\mA$.
Since $\mD$ is efficient, so is $\mA$, so from the injective invariance property of ENTCF families we get that $\Delta^\mA \leq \negl(\lambda)$, so 
\begin{align}
\pr{\mA \textnormal{ correct}} \leq \frac{1}{2} + \negl(\lambda)\,. \label{eqn:inj_inv_A}
\end{align}
On the other hand, 
\begin{align}
\pr{\mA \textnormal{ correct}} 
&= \pr{v = v'} \pr{\mA \textnormal{ correct} \, | \, v = v'} +  \pr{v \neq v'} \pr{\mA \textnormal{ correct} \, | \, v \neq v'} \,.
\end{align}
In the case $v \neq v'$, $\mA$ guesses uniformly at random , so $\pr{\mA \textnormal{ correct} \, | \, v \neq v'} = 1/2$.
In the case $v = v'$, by construction $\mA$ guesses correctly if and only if $\mD$ guesses correctly, so $\pr{\mA \textnormal{ correct} \, | \, v = v'} = \frac{1}{2} + \Delta^\mD$.
Inserting this, we get 
\begin{align}
\pr{\mA \textnormal{ correct}} 
&= \pr{v = v'} \left( \frac{1}{2} + \Delta^\mD \right) +  (1 - \pr{v = v'}) \frac{1}{2} \\
&= \frac{1}{2} + \pr{v = v'}  \Delta^\mD \,.
\end{align}
Now notice that 
\begin{align}
\pr{v = v'} 
&= \frac{1}{2} \tr{\sigma^{(\vec \theta^0, v, \vec b)}} + \frac{1}{2} \tr{\sigma^{(\vec \theta^1, v, \vec b)}} \\
&\geq \frac{1}{2} \norm{\sigma^{(\vec \theta^0, v, \vec b)} - \sigma^{(\vec \theta^1, v, \vec b)}}_1 \\
&\geq \Delta^\mD \,,
\end{align}
since the distinguishing advantage of any distinguisher is upper-bounded by the trace distance.
Combining this with \ref{eqn:inj_inv_A}, we get 
\begin{align}
\frac{1}{2} + \negl(\lambda) \geq \frac{1}{2} + \left( \Delta^\mD  \right)^2 \,.
\end{align}
Since the square root of a negligible function is again negligible, $\Delta^\mD = \negl(\lambda)$ as desired.
\end{proof}

Since $Z(\vec a)$ and $X(\vec b)$ are efficient observables, \ref{lem:partial_postmeas_indist} immediately implies that we can extend \ref{eqn:z_vec_succprob} and \ref{eqn:x_vec_succprob} to different choices of $\vec \theta$:

\begin{corollary} \label{lem:observables_succ_prob_mixed_theta}
Consider an efficient device $D = (S, \Pi, M, P)$ and a bit $v \in \bits$.
\begin{enumerate}
\item For any $\vec \theta, \vec a \in \bits^n$ such that $a_i = 1 \implies \theta_i = 0$:
\begin{align}
\tr{Z(\vec a)^{(v)} \sigma^{(\vec \theta, v, \vec a)}} &\approx_{\gamma_H(D)} \tr{\sigma^{(\vec \theta, v, \vec a)}} \,. \label{eqn:z_vec_succprob_mixed_theta}
\end{align}
\item For any $\vec \theta, \vec a \in \bits^n$ such that $a_i = 1 \implies \theta_i = 1$:
\begin{align}
\tr{X(\vec a)^{(v)} \sigma^{(\vec \theta, v, \vec a)}} &\approx_{\gamma_H(D)} \tr{\sigma^{(\vec \theta, v, \vec a)}} \,. \label{eqn:x_vec_succprob_mixed_theta}
\end{align}
\end{enumerate}
\end{corollary}

Using \cite[Lemma 2.19]{self-testing}, this can also be restated as follows:

\begin{corollary} \label{lem:observables_one_mixed_theta}
Consider an efficient device $D = (S, \Pi, M, P)$ and a bit $v \in \bits$.
\begin{enumerate}
\item For any $\vec \theta, \vec a \in \bits^n$ such that $a_i = 1 \implies \theta_i = 0$:
\begin{align}
Z(\vec a) \approx_{\gamma_H(D), \sigma^{(\vec \theta, v, \vec a)}}  (-1)^v \1 \,.  
\end{align}
\item For any $\vec \theta, \vec a \in \bits^n$ such that $a_i = 1 \implies \theta_i = 1$:
\begin{align}
X(\vec a) \approx_{\gamma_H(D), \sigma^{(\vec \theta, v, \vec a)}} (-1)^{v} \1 \,. 
\end{align}
\end{enumerate}
\end{corollary}

\subsection{Pauli group relations for inefficient observables\footnote{This section is based on unpublished work by Jeffrey Champion, John Wright, and the second author. We thank Jeffrey and John for allowing us to use these results here.}} \label{sec:ineff_pauli}
The goal of this section is to prove that the observables $\{\tilde X(\vec a)\}_{\vec a \in \bits^n}$ and $\{ Z(\vec b)\}_{\vec b \in \bits^n}$ approximately satisfy the relations of the $n$-qubit Pauli observables.
More formally, we will show the following proposition:

\begin{proposition} \label{lem:approx_rep_tilde}
For any efficient perfect device $D = (S, \Pi, M, P)$ and any $\vec a, \vec b \in \bits^n$:
\begin{equation}
\tilde X(\vec a) Z(\vec b) \approx_{n \gamma_H(D)^{1/4}, \sigma^{(\vec 1)}} (-1)^{\vec a \cdot \vec b} Z(\vec b) \tilde X(\vec a) \,.
\end{equation}
\end{proposition}

The proof of this can be found at the end of this section. 
We first need to show a number of preparatory results.

\begin{lemma} \label{lem:z_pi_equiv}
We define
\begin{equation}
\Pi_i^{(b)} =  \sum_{b_1, \dots, b_{i-1}, b_{i+1}, \dots, b_n} \sum_{\vec x} \Pi^{(\vec b, \vec x)} \,,
\end{equation}
where $\vec b = (b_1, \dots, b_{i-1}, b, b_{i+1}, \dots, b_n)$.
For any efficient perfect device $D = (S, \Pi, M, P)$, the following holds for any $i$:
\begin{align}
\sum_{b, \vec d} \norm{
\left( M^{(\vec d)} \Pi^{(b)}_{i} - Z_{i, \vec d}^{(b)} M^{(\vec d)} \right) \left( \psi^{(\vec 1^i)} \right)^{1/2}}^2_2 \approx_{\gamma_H(D)} 0 \,, \label{eqn:zi_pi_equiv_statement}
\end{align}
where $\norm{\cdot}_2$ denotes the Schatten 2-norm (also called Hilbert-Schmidt norm) and $\vec 1^i \in \bits^n$ is the bitstring whose $i$-th bit is 1 and whose remaining bits are 0 (as defined in \ref{sec:notation}).
\end{lemma}
\begin{proof}
The proof is analogous to \cite[Lemma 4.19]{self-testing}, but we spell out the details for completeness.
Writing out the definition of the 2-norm and multiplying out terms, we find that the l.h.s.~of \ref{eqn:zi_pi_equiv_statement} equals
\begin{multline}
\sum_{b, \vec d} \tr{M^{(\vec d)}\Pi^{(b)}_{i} \psi^{(\vec 1^i)} \Pi^{(b)}_{i} M^{(\vec d)}}
+ \sum_{b, \vec d} \tr{Z_{i, \vec d}^{(b)} M^{(\vec d)} \psi^{(\vec 1^i)} M^{(\vec d)} Z_{i, \vec d}^{(b)}} \\
- \sum_{b, \vec d} \tr{M^{(\vec d)} Z_{i, \vec d}^{(b)} M^{(\vec d)} \left( \Pi^{(b)}_{i} \psi^{(\vec 1^i)} + \psi^{(\vec 1^i)} \Pi^{(b)}_{i} \right)} \label{eqn:z_pi_expanded}
\end{multline}
Since $\{M^{(\vec d)}\}_{\vec d}$, $\{\Pi^{(b)}_{i}\}_{b}$, and $\{Z_{i, \vec d}^{(b)}\}_{b}$ form projective measurements, the first two terms equal 1.
For the third term, we note that since $\Pi^{(0)}_{i} + \Pi^{(1)}_{i} = \1$,
\begin{align}
\Pi^{(b)}_{i} \psi^{(\vec 1^i)} + \psi^{(\vec 1^i)} \Pi^{(b)}_{i} = 2 \, \Pi^{(b)}_{i} \psi^{(\vec 1^i)} \Pi^{(b)}_{i}  + \Pi^{(0)}_{i} \psi^{(\vec 1^i)} \Pi^{(1)}_{i} + \Pi^{(1)}_{i} \psi^{(\vec 1^i)} \Pi^{(0)}_{i} \,.
\end{align}
Note that since $\{Z_{i, \vec d}^{(b)}\}_{b}$ and $\{M^{(\vec d)}\}_{\vec d}$ are projective  measurements, we have
\begin{align*}
\sum_{b, \vec d} \tr{M^{(\vec d)} Z_{i, \vec d}^{(b)} M^{(\vec d)} \left( \Pi^{(0)}_{i} \psi^{(\vec 1^i)} \Pi^{(1)}_{i} + \Pi^{(1)}_{i} \psi^{(\vec 1^i)} \Pi^{(0)}_{i} \right)} = \tr{\Pi^{(0)}_{i} \psi^{(\vec 1^i)} \Pi^{(1)}_{i} + \Pi^{(1)}_{i} \psi^{(\vec 1^i)} \Pi^{(0)}_{i}} = 0 \,,
\end{align*}
where the last equality holds because $\{\Pi^{(b)}_{i}\}_{b}$ are orthogonal projectors.
Therefore, the third term in \ref{eqn:z_pi_expanded} equals
\begin{align}
2 \, \sum_{b, \vec d} \tr{Z_{i, \vec d}^{(b)} M^{(\vec d)} \Pi^{(b)}_{i} \psi^{(\vec 1^i)} \Pi^{(b)}_{i} M^{(\vec d)}} \,.
\end{align}

We now want to replace $\psi^{(\vec 1^i)}$ by $\psi^{(\vec 0)}$ in the above expression.
For this, observe that given any state $\rho$, we can efficiently estimate $\sum_{b, \vec d} \tr{Z_{i, \vec d}^{(b)} M^{(\vec d)} \Pi^{(b)}_{i} \psi^{(\vec 1^i)} \Pi^{(b)}_{i} M^{(\vec d)}}$ by measuring $\{\Pi^{(b)}_{i}\}_{b}$, $\{M^{(\vec d)}\}_{\vec d}$, and $\{Z_{i, \vec d}^{(b)}\}_{b}$ in sequence and checking whether the $\Pi_i$- and $Z_i$-measurements yielded the same result.
Therefore, since $\psi^{(\vec 1^i)} \capprox_0 \psi^{(\vec 0)}$ by an argument analogous to \ref{lem:states_indist}, we have
\begin{align}
\sum_{b, \vec d} \tr{Z_{i, \vec d}^{(b)} M^{(\vec d)} \Pi^{(b)}_{i} \psi^{(\vec 1^i)} \Pi^{(b)}_{i} M^{(\vec d)}} \approx_0 \sum_{b, \vec d} \tr{Z_{i, \vec d}^{(b)} M^{(\vec d)} \Pi^{(b)}_{i} \psi^{(\vec 0)} \Pi^{(b)}_{i} M^{(\vec d)}} \,.
\end{align}
This expression equals 1 if the $\Pi_i$- and $Z_i$-measurements yield the same result.
Since in $\psi^{(\vec 0)}$ the basis choice for the $i$-th key is $\theta_i = 0$ and the device $D$ is perfect, for any image $y_i$ the $\Pi_i$-measurement will output the bit $\hat b (k_i, y_i)$ with probability negligibly close to 1.
Hence, by \ref{def:partial_postmeas_states}, 
\begin{align}
\sum_{\vec d} M^{(\vec d)} \Pi^{(b)}_{i} \psi^{(\vec 0)} \Pi^{(b)}_{i} M^{(\vec d)} \ot \proj{\vec d} \approx_0 \sigma^{(\vec 0, b, \vec 1^i)} \,.
\end{align}
With this and noting that $Z(\vec 1^i) = Z_i$, we get from \ref{lem:observables_succ_prob} that
\begin{align}
\sum_{b, \vec d} \tr{Z_{i, \vec d}^{(b)} M^{(\vec d)} \Pi^{(b)}_{i} \psi^{(\vec 0)} \Pi^{(b)}_{i} M^{(\vec d)}} 
&\approx_{0} \sum_{b} \tr{Z_{i}^{(b)} \sigma^{(\vec 0, b, \vec 1^i)} }\\
&\approx_{\gamma_H(D)} \sum_{b} \tr{ \sigma^{(\vec 0, b, \vec 1^i)} } \\
&= \tr{ \sigma^{(\vec 0)} } \\
&= 1 \,,
\end{align}
concluding the proof.
\end{proof}

\begin{lemma} \label{lem:anticomm_single_obs}
For any efficient perfect device $D = (S, \Pi, M, P)$, the following holds for any $i \in [n]$:
\begin{align}
\tr{Z_i \tilde X_i Z_i \sigma^{(\vec 1^i)}} \approx_{\gamma_H(D)^{1/2}} -1 \,.
\end{align}
\end{lemma}

\begin{proof}
Using that $Z_i$ is a binary observable, it is easy to check that 
\begin{equation}
Z_i \tilde X_i Z_i = \left( 2 \sum_{b} Z_i^{(b)} \tilde X_i Z_i^{(b)}  \right) - \tilde X_i \,.
\end{equation}
Since 
\begin{align}
\tr{\tilde X_i \sigma^{(\vec 1^i)}} \approx_{\gamma_H(D)} 1
\end{align}
by \ref{lem:xtilde_succprob_allstates}, to show the lemma it suffices to show that 
\begin{equation}
\sum_b \tr{Z_i^{(b)} \tilde X_i Z_i^{(b)} \sigma^{(\vec 1^i)}} \approx_{\gamma_H(D)^{1/2}} 0 \,.
\end{equation}
Inserting the definition of $\sigma^{(\vec 1^i)}$ from \ref{eqn:def_sigma}, we get 
\begin{align}
\sum_b \tr{Z_i^{(b)} \tilde X_i Z_i^{(b)} \sigma^{(\vec 1^i)}} = \sum_{b, \vec d} \tr{ \tilde X_{i, \vec d} Z_{i, \vec d}^{(b)} M^{(\vec d)} \psi^{(\vec 1^i)} M^{(\vec d)} Z_{i, \vec d}^{(b)}}
\end{align}
Using \ref{lem:z_pi_equiv} and the Cauchy-Schwarz inequality, we can show that 
\begin{align}
\sum_{b, \vec d} \tr{ \tilde X_{i, \vec d} Z_{i, \vec d}^{(b)} M^{(\vec d)} \psi^{(\vec 1^i)} M^{(\vec d)} Z_{i, \vec d}^{(b)}} 
\approx_{\gamma_H(D)^{1/2}} \sum_{b, \vec d} \tr{ \tilde X_{i, \vec d} M^{(\vec d)} \Pi_i^{(b)} \psi^{(\vec 1^i)} \Pi_i^{(b)} M^{(\vec d)} } \,. \label{eqn:after_cs_application}
\end{align}
This is a standard step in state-dependent distance calculations, but we spell out the details for completeness in \ref{lem:cs_replacement_step}.

To prove the lemma, it now suffices to show that the r.h.s.~of \ref{eqn:after_cs_application} is negligible.
For the sake of contradiction, suppose that there exists a non-negligible function $\mu(\lambda)$ such that 
\begin{align}
\sum_{b, \vec d} \tr{ \tilde X_{i, \vec d} M^{(\vec d)} \Pi_i^{(b)} \psi^{(\vec 1^i)} \Pi_i^{(b)} M^{(\vec d)} } > \mu(\lambda) \,. \label{eqn:ahcb_seq_meas_contradiction}
\end{align}
From this, we can construct the following adversary $\mA$ for the adaptive hardcore bit property:
$\mA$ receives a key $k_i \in \kf$ and samples keys $k_1, \dots, k_{i-1}, k_{i+1}, k_n \in \kg$. 
It then performs the device's operations with these keys to obtain $\psi^{(\vec 1^i)}$.
$\mA$ now measures the projectors $\{\Pi_i^{(b, x)}\}_{b, x}$ and records the result $(b, x)$, where 
\begin{align}
\Pi_i^{(b, x)} = \sum_{b_1, \dots, b_{i-1}, b_{i+1}, \dots, b_n} \sum_{x_1, \dots, x_{i-1}, x_{i+1}, \dots, x_n} \Pi^{(\vec b, \vec x)} \,.
\end{align}
$\mA$ then measures $\{M^{(\vec d)}\}_{\vec d}$ receiving an outcome $\vec d$, followed by $\{X_{i, \vec d}^{(c)}\}_{c}$, receiving an outcome $c$.
$\mA$ submits $(b, x, d_i, c)$ as a guess for the adaptive hardcore bit property.

First note that since the device operations are efficient, $\mA$ is efficient, too.
To show that $\mA$ breaks the adaptive hardcore bit property, we need to show that on average over $\vec k$, $\vec y$, and $\vec d$, $\hat u(k_i, y_i, d_i) = c$ with non-negligible advantage.

For this, we observe that since the device $D$ is perfect, if the adversary measures $(b, x)$ as described above, conditioned on receiving the bit $b$, the preimage $x$ it obtains will be $x_b(k_i, y_i)$ with probability negligibly close to 1.
As a consequence, by  the gentle measurement lemma~\cite{winter1999} the post-measurement state after measuring $\{\Pi_i^{(b, x)}\}_{b, x}$ is negligibly close to the post-measurement state after measuring $\{\Pi_i^{(b)}\}_{b}$ (i.e.~it does not matter for the post-measurement state whether $\mA$ also measures the preimage $x$ or not): 
\begin{align}
\Pi_i^{(b, x)} \psi^{(\vec 1^i)} \Pi_i^{(b, x)} \approx_0 \Pi_i^{(b)} \psi^{(\vec 1^i)} \Pi_i^{(b)} \,.
\end{align}

In the next step, the adversary measures $\{M^{(\vec d)}\}_{\vec d}$, followed by $\{X_{i, \vec d'}^{(c)}\}_{c}$.
The probability that $\mA$ guesses correctly, i.e.~that $c = \hat u(k_i, y_i, d_i)$, is then given by (up to negligible error from to the approximation $\sum_{b, x} \Pi_i^{(b, x)} \psi^{(\vec 1^i)} \Pi_i^{(b, x)}  \approx_0 \sum_{b} \Pi_i^{(b)} \psi^{(\vec 1^i)} \Pi_i^{(b)}$)
\begin{align}
\pr{\text{$\mA$ guesses correctly}} \approx_0 \sum_{b, \vec y, \vec d} \tr{ X_{i, \vec y, \vec d}^{(\hat u(k_i, y_i, d_i))} M^{(\vec d)}_{\vec y} \Pi^{(b)}_{i, \vec y} \psi^{(\vec 1^i)}_{\vec y} \Pi^{(b)}_{i, \vec y} M^{(\vec d)}_{\vec y} } \,.
\end{align}
Since $X_{i, \vec y, \vec d}$ is a binary observable, we can express this as
\begin{align}
\pr{\text{$\mA$ guesses correctly}} 
&= \sum_{b, \vec y, \vec d} \tr{ \frac{1}{2}\left( \1 + (-1)^{\hat u(k_i, y_i, d_i)} X_{i, \vec y, \vec d} \right)  M^{(\vec d)}_{\vec y} \Pi^{(b)}_{i, \vec y} \psi^{(\vec 1^i)}_{\vec y} \Pi^{(b)}_{i, \vec y} M^{(\vec d)}_{\vec y} } \\
\intertext{Using that $M^{(\vec d)}_{\vec y}$ and $\Pi^{(b)}_{i, \vec y}$ form projective measurements:}
&= \frac{1}{2} + \sum_{b, \vec y, \vec d} \tr{ (-1)^{\hat u(k_i, y_i, d_i)} X_{i, \vec y, \vec d}  M^{(\vec d)}_{\vec y} \Pi^{(b)}_{i, \vec y} \psi^{(\vec 1^i)}_{\vec y} \Pi^{(b)}_{i, \vec y} M^{(\vec d)}_{\vec y} } \\
\intertext{By \ref{def:observables}, $(-1)^{\hat u(k_i, y_i, d_i)} X_{i, \vec y, \vec d} = \tilde X_{i, \vec y, \vec d}$:}
&= \frac{1}{2} + \sum_{b, \vec y, \vec d} \tr{ \tilde X_{i, \vec y, \vec d}  M^{(\vec d)}_{\vec y} \Pi^{(b)}_{i, \vec y} \psi^{(\vec 1^i)}_{\vec y} \Pi^{(b)}_{i, \vec y} M^{(\vec d)}_{\vec y} } \\
&= \frac{1}{2} + \sum_{b, \vec d} \tr{ \tilde X_{i, \vec d}  M^{(\vec d)} \Pi^{(b)}_{i} \psi^{(\vec 1^i)} \Pi^{(b)}_{i} M^{(\vec d)} } \,.
\end{align}
Hence, \ref{eqn:ahcb_seq_meas_contradiction} would imply that the adversary has non-negligible advantage in guessing the adaptive hardcore bit, a contradiction.
\end{proof}

\begin{lemma} \label{lem:cs_replacement_step}
With the setup and notation of \ref{lem:anticomm_single_obs}, the following holds:
\begin{align}
\sum_{b, \vec d} \tr{ \tilde X_{i, \vec d} Z_{i, \vec d}^{(b)} M^{(\vec d)} \psi^{(\vec 1^i)} M^{(\vec d)} Z_{i, \vec d}^{(b)}} 
\approx_{\gamma_H(D)^{1/2}} \sum_{b, \vec d} \tr{ \tilde X_{i, \vec d} M^{(\vec d)} \Pi_i^{(b)} \psi^{(\vec 1^i)} \Pi_i^{(b)} M^{(\vec d)} } \,.
\end{align}
\end{lemma}
\begin{proof}
To show the lemma, there are two instances where we need to replace terms of the form $M^{(\vec d)} Z_{i, \vec d}$ by $\Pi_i^{(b)} M^{(\vec d)}$.
As a first step, we will show that
\begin{align}
\sum_{b, \vec d} \tr{ \tilde X_{i, \vec d} Z_{i, \vec d} M^{(\vec d)} \psi^{(\vec 1^i)} \left( M^{(\vec d)} Z_{i, \vec d} - \Pi_i^{(b)} M^{(\vec d)} \right)} \approx_{\gamma_H(D)^{1/2}} 0 \,, \label{eqn:cs_replacement_goal}
\end{align}
with $\Pi_i^{(b)}$ defined as in \ref{lem:z_pi_equiv}. 
For this, we apply the triangle inequality and write the l.h.s.~of \ref{eqn:cs_replacement_goal} as a Hilbert-Schmidt inner product. Then we can apply the Cauchy-Schwarz inequality twice: 
\begin{align}
&\left|\sum_{b, \vec d} \tr{ \tilde X_{i, \vec d} Z_{i, \vec d} M^{(\vec d)} \psi^{(\vec 1^i)} \left( M^{(\vec d)} Z_{i, \vec d} - \Pi_i^{(b)} M^{(\vec d)} \right)}\right| \\
&\leq \sum_{b, \vec d} \left| \left\langle  \left( Z_{i, \vec d} M^{(\vec d)} - M^{(\vec d)} \Pi_i^{(b)} \right) \left( \psi^{(\vec 1^i)} \right)^{1/2} , \tilde X_{i, \vec d} Z_{i, \vec d} M^{(\vec d)} \left( \psi^{(\vec 1^i)} \right)^{1/2} \right\rangle \right| \\
& \leq \sum_{b, \vec d} \norm{\tilde X_{i, \vec d} Z_{i, \vec d} M^{(\vec d)} \left( \psi^{(\vec 1^i)} \right)^{1/2}}_2 \norm{\left( Z_{i, \vec d} M^{(\vec d)} - M^{(\vec d)} \Pi_i^{(b)} \right) \left( \psi^{(\vec 1^i)} \right)^{1/2}}_2 \\
&\leq \left( \sum_{b, \vec d} \norm{\tilde X_{i, \vec d} Z_{i, \vec d} M^{(\vec d)} \left( \psi^{(\vec 1^i)} \right)^{1/2}}_2^2  \right)^{1/2} \left( \sum_{b, \vec d} \norm{\left( Z_{i, \vec d} M^{(\vec d)} - M^{(\vec d)} \Pi_i^{(b)} \right) \left( \psi^{(\vec 1^i)} \right)^{1/2}}_2^2  \right)^{1/2} \\
\intertext{Writing out the definition of the norm, it is easy to check that the first factor equals 1, and the second factor is bounded by \ref{lem:z_pi_equiv}:}
& \leq O(\gamma_H(D)^{1/2}) \,.
\end{align}
Repeating the above steps for the other $Z_{i, \vec d}^{(b)} M^{(\vec d)}$-term, we obtain the desired statement.
\end{proof}

\begin{lemma} \label{lem:anticomm_many_z}
For any efficient perfect device $D = (S, \Pi, M, P)$, any index $i \in [n]$, and any string $\vec b \in \bits^n$, the following holds:
\begin{align}
\tilde X_i Z(\vec b) \approx_{\gamma_H(D)^{1/2}, \sigma^{(\vec 1)}} (-1)^{b_i} Z(\vec b) \tilde X_{i} \,.
\end{align}
\end{lemma}
\begin{proof}
Since $Z(\vec b)$ is efficient, $\tilde X_i$ is efficient with access to the trapdoor $t_i$ for the $i$-th key, and $\vec 1$ and $\vec 1^i$ agree on their $i$-th bit, by \ref{lem:partial_lifting} it suffices to show that 
\begin{align}
\tilde X_i Z(\vec b) \approx_{\gamma_H(D)^{1/2}, \sigma^{(\vec 1^i)}} (-1)^{b_i} Z(\vec b) \tilde X_{i} \,.
\end{align}
By \ref{lem:state_dep_distance_expanded}, this is implied by
\begin{align}
\tr{\tilde X_i Z(\vec b) \tilde X_i Z(\vec b) \sigma^{(\vec 1^i)}} \approx_{\gamma_H(D)^{1/2}} (-1)^{b_i} \,.
\end{align}
We can use \ref{lem:xtilde_is_one} (noting that the $i$-th bit of $\vec 1^i$ is 1) together with \ref{lem:replace_in_trace} to replace the leftmost $\tilde X_i$-operator by $\1$:
\begin{align}
\tr{\tilde X_i Z(\vec b) X_i Z(\vec b) \sigma^{(\vec 1^i)}} \approx_{\gamma_H(D)^{1/2}} \tr{Z(\vec b) \tilde X_i Z(\vec b) \sigma^{(\vec 1^i)}} \,.
\end{align}
For the rest of the proof, we distinguish two cases:
\paragraph{Case $b_i = 0$.}
In this case, we split $\sigma^{(\vec 1^i)} = \sum_{v \in \bits} \sigma^{(\vec 1^i, v, \vec b)}$:
\begin{align}
\tr{Z(\vec b) \tilde X_i Z(\vec b) \sigma^{(\vec 1^i)}} 
&= \sum_{v \in \bits} \tr{Z(\vec b) \tilde X_i Z(\vec b) \sigma^{(\vec 1^i, v, \vec b)}} \\
\intertext{Since $b_i = 0$ by assumption, the condition $b_i = 1 \implies \theta_i = 0$ (with $\vec \theta = \vec 1^i$) in \ref{lem:observables_one_mixed_theta} is satisfied. We can therefore use that lemma together with \ref{lem:replace_in_trace} on each term in the sum to get:}
&\approx_{\gamma_H(D)^{1/2}} \sum_{v \in \bits} (-1)^v \,\tr{Z(\vec b) \tilde X_i \sigma^{(\vec 1^i, v, \vec b)}} \\
&\approx_{\gamma_H(D)^{1/2}} \sum_{v \in \bits} \tr{\tilde X_i \sigma^{(\vec 1^i, v, \vec b)}} \\
&= \tr{\tilde X_i \sigma^{(\vec 1^i)}} \\
\intertext{Noting that $\tilde X_i = \tilde X(\vec 1^i)$, we can use \ref{lem:xtilde_is_one} and \ref{lem:replace_in_trace} to conclude:}
&\approx_{\gamma_H(D)^{1/2}} \tr{ \sigma^{(\vec 1^i)}} \\
&= 1 \,.
\end{align}

\paragraph{Case $b_i = 1$.}
We define $\vec b'$ by $b'_i = 0$ and $b'_j = b_j$ for $j \neq i$.
Then, by \ref{lem:observables_one_mixed_theta} we have that $Z(\vec b') \approx_{\sigma^{(\vec 1^i, v, \vec b')}} (-1)^v \1$.
Furthermore, by definition of the $Z$-observables, $Z(\vec b) = Z_i Z(\vec b') = Z(\vec b') Z_i$.
We can use this with \ref{lem:replace_in_trace} similarly to the previous case to obtain 
\begin{align}
\tr{Z(\vec b) \tilde X_i Z(\vec b) \sigma^{(\vec 1^i)}} 
&= \sum_{v \in \bits} \tr{Z(\vec b') Z_i \tilde X_i Z_i Z(\vec b') \sigma^{(\vec 1^i, v, \vec b)}} \\
&\approx_{\gamma_H(D)^{1/2}} \sum_{v \in \bits} (-1)^v \tr{Z(\vec b') Z_i \tilde X_i Z_i \sigma^{(\vec 1^i, v, \vec b)}} \\
&\approx_{\gamma_H(D)^{1/2}} \sum_{v \in \bits} \tr{Z_i \tilde X_i Z_i \sigma^{(\vec 1^i, v, \vec b)}} \\
&= \tr{Z_i \tilde X_i Z_i \sigma^{(\vec 1^i)}} \\
&\approx_{\gamma_H(D)^{1/2}} -1 \,,
\end{align}
where we used \ref{lem:anticomm_single_obs} in the last line.
\end{proof}

\begin{proof}[Proof of \ref{lem:approx_rep_tilde}]
Recall that our goal is to prove that 
\begin{equation}
\tilde X(\vec a) Z(\vec b) \approx_{n \gamma_H(D)^{1/4}, \sigma^{(\vec 1)}} (-1)^{\vec a \cdot \vec b} Z(\vec b) \tilde X(\vec a) \,.
\end{equation}
By \ref{lem:state_dep_distance_expanded}, it suffices to show that
\begin{align}
\tr{Z(\vec b) \tilde X(\vec a) Z(\vec b) \tilde X(\vec a) \sigma^{(\vec 1)}} \approx_{n \gamma_H(D)^{1/4}} (-1)^{\vec a \cdot \vec b} \,.
\end{align}
As a first step, we apply \ref{lem:xtilde_is_one} to obtain 
\begin{align}
\tr{Z(\vec b) \tilde X(\vec a) Z(\vec b) \tilde X(\vec a) \sigma^{(\vec 1)}} \approx_{\gamma_H(D)^{1/2}} \tr{Z(\vec b) \tilde X(\vec a) Z(\vec b) \sigma^{(\vec 1)}} \,. 
\end{align}
We can write $\tilde X(\vec a) = \prod_{i =1}^n \tilde X_i^{a_i}$ (where $\tilde X_i^{a_i} = \1$ if $a_i = 0$):
\begin{align}
\tr{Z(\vec b) \tilde X(\vec a) Z(\vec b) \sigma^{(\vec 1)}} 
&= \tr{Z(\vec b) \left( \prod_{i =1}^{n-1} \tilde X_i^{a_i} \right) X_n^{a_n} Z(\vec b) \sigma^{(\vec 1)}} 
\intertext{Applying \ref{lem:replace_in_trace} with \ref{lem:anticomm_many_z}  (if $a_n = 1$; if $a_n = 0$, the same step follows trivially and with equality):}
&\approx_{\gamma_H(D)^{1/4}} (-1)^{a_n b_n} \tr{Z(\vec b) \left( \prod_{i =1}^{n-1} \tilde X_i^{a_i} \right) Z(\vec b) X_n^{a_n} \sigma^{(\vec 1)}} 
\intertext{Applying \ref{lem:replace_in_trace} with \ref{lem:xtilde_is_one} (if $a_n = 1$; if $a_n = 0$, the same step follows trivially and with equality):}
&\approx_{\gamma_H(D)^{1/2}} (-1)^{a_n b_n} \tr{Z(\vec b) \left( \prod_{i =1}^{n-1} \tilde X_i^{a_i} \right) Z(\vec b) \sigma^{(\vec 1)}} \,. 
\end{align}
Repeating the above steps for each of the remaining $\tilde X_i^{a_i}$, we find that 
\begin{align}
\tr{Z(\vec b) \tilde X(\vec a) Z(\vec b) \sigma^{(\vec 1)}} \approx_{n \gamma_H(D)^{1/4}} \left( \prod_{i = 1}^n (-1)^{a_i b_i} \right) \tr{Z(\vec b) Z(\vec b) \sigma^{(\vec 1)}} = (-1)^{\vec a \cdot \vec b}
\end{align}
as desired.
The factor of $n$ in the approximation arises because we incur an approximation error of $O(\gamma_H(D)^{1/4})$ for dealing with each of the $n$ operators $\tilde X_i^{a_i}$, so by the triangle inequality the total approximation error is $O(n \gamma_H(D)^{1/4})$.
\end{proof}

\subsection{Switching to efficient observables}

\begin{definition}[Rounding isometries] \label{def:iso}
For a device $D$ with associated Hilbert space $\H_D$ and $\vec y \in \mY^{\times n}$, $d \in \bits^{w \times n}$, we define the isometry $\tilde V_{y, d}: \H_D \to \H_D \ot \H_A \ot \H_Q$ by the following action on an arbitrary state $\ket{\varphi}_D$: 
\begin{align}
\tilde V_{\vec y, \vec d} \ket{\varphi}_D = \E_{\vec a, \vec b \in \bits^n} \left( \left( \tilde X(\vec a)_{\vec y, \vec d} Z(\vec b)_{\vec y, \vec d} \right)_{D} \ot \left( \sigma_X(\vec a) \sigma_Z(\vec b) \right)_A \right) \ket{\varphi}_D \ot \left( \ket{\phi^{(0,0)}}^{\ot n} \right)_{AQ} \,,
\end{align}
where $\ket{\phi^{(0,0)}} = \frac{\ket{00} + \ket{11}}{\sqrt{2}}$ denotes an EPR pair, and $\left( \ket{\phi^{(0,0)}}^{\ot n} \right)_{AQ}$ is distributed between $A$ and $Q$ such that every EPR pair has one qubit in either system.
We can combine the different $V_{y,d}$ into one isometry 
\begin{align}
\tilde V = \sum_{\vec y, \vec d} \tilde V_{\vec y, \vec d} \ot \proj{\vec y, \vec d} : \H_{D} \ot \H_{Y} \ot \H_R \to \H_{D} \ot \H_{Y} \ot \H_R \ot \H_A \ot \H_Q \,.
\end{align}
We similarly define 
\begin{align}
V_{\vec y, \vec d} \ket{\varphi}_D = \E_{\vec a, \vec b \in \bits^n} \left( \left( X(\vec a)_{\vec y, \vec d} Z(\vec b)_{\vec y, \vec d} \right)_{D} \ot \left( \sigma_X(\vec a) \sigma_Z(\vec b) \right)_A \right) \ket{\varphi}_D \ot \left( \ket{\phi^{(0,0)}}^{\ot n} \right)_{AQ}
\end{align}
and 
\begin{align}
V = \sum_{\vec y, \vec d} V_{\vec y,\vec d} \ot \proj{\vec y, \vec d} \,.
\end{align}
\end{definition}

\begin{remark}
The fact that $\tilde V$ and $V$ as defined in \ref{def:iso} are indeed isometries, i.e. that $\tilde V^\dagger \tilde V = \1$ and $V^\dagger V = \1$ follows straightforwardly from the fact that 
\begin{align}
\left( \bra{\phi^{(0,0)}}^{\ot n} \right)_{AQ} \left( \sigma_X(\vec a') \sigma_Z(\vec b') \right)_A^\dagger \left( \sigma_X(\vec a) \sigma_Z(\vec b) \right)_A \left( \ket{\phi^{(0,0)}}^{\ot n} \right)_{AQ} = \delta_{\vec a, \vec a'} \delta_{\vec b, \vec b'} \,,
\end{align}
and $\tilde X(\vec a)^2 = X(\vec a)^2 = Z(\vec b)^2 = \1$.
\end{remark}

\begin{remark}
The isometry $\tilde V$ is inefficient since $\tilde X$ is an inefficient observable and is furthermore only defined for the basis choice $\vec \theta = \vec 1$ since $\tilde X$ depends on the function $\hat u$ (see \ref{def:observables}).
In contrast, since $X$ and $Z$ are both efficient observables, $V$ is an efficient isometry and is well-defined for any basis choice.
\end{remark}

The following lemma relates $\tilde V$ and $V$.
\begin{lemma} \label{lem:relation_V_tilde_notilde}
For any keys $\vec k \in \kf^n$: 
\begin{align}
V_{\vec y, \vec d} = \sigma_Z(\hat u(\vec k, \vec y, \vec d))_{A} \ot \sigma_Z(\hat u(\vec k, \vec y, \vec d))_{Q} \tilde V_{\vec y, \vec d} \,.
\end{align}
\end{lemma}

\begin{proof}
For any state $\ket{\varphi}_D$, we have: 
\begin{align*}
& \sigma_Z(\hat u(\vec k, \vec y, \vec d))_{A} \ot \sigma_Z(\hat u(\vec k, \vec y, \vec d))_{Q} \tilde V_{\vec y, \vec d} \ket{\varphi}_D \\
&= \E_{a, b \in \bits^n} \left( \tilde X(\vec a)_{\vec y, \vec d} Z(\vec b)_{\vec y, \vec d} \right)_{D} \ket{\varphi}_D \ot \left[ \left( \sigma_Z(\hat u(\vec k, \vec y, \vec d)) \sigma_X(\vec a) \sigma_Z(\vec b) \right)_A \ot \sigma_Z(\hat u(\vec k, \vec y, \vec d))_{Q} \left( \ket{\phi^{(0,0)}}^{\ot n} \right)_{AQ}  \right] \\
\intertext{Repeatedly using that $\left( \sigma_Z  \right)_{A} \ket{\phi^{(0,0)}}_{AQ} = \left( \sigma_Z  \right)_{Q} \ket{\phi^{(0,0)}}_{AQ}$:}
&= \E_{a, b \in \bits^n} \left( \tilde X(\vec a)_{\vec y, \vec d} Z(\vec b)_{\vec y, \vec d} \right)_{D} \ket{\varphi}_D \ot \left[ \left( \sigma_Z(\hat u(\vec k, \vec y, \vec d)) \sigma_X(\vec a) \sigma_Z(\vec b) \sigma_Z(\hat u(\vec k, \vec y, \vec d)) \right)_A \left( \ket{\phi^{(0,0)}}^{\ot n} \right)_{AQ}  \right] \\
\intertext{Since $\sigma_Z(\hat u(\vec k, \vec y, \vec d)) \sigma_X(\vec a) \sigma_Z(\vec b) \sigma_Z(\hat u(\vec k, \vec y, \vec d)) = (-1)^{a \cdot \hat u(\vec k, \vec y, \vec d)} \sigma_X(\vec a) \sigma_Z(\vec b)$:}
&= \E_{a, b \in \bits^n} \left( (-1)^{a \cdot u(\vec k, \vec y, \vec d)} \tilde X(\vec a)_{\vec y, \vec d} Z(\vec b)_{\vec y, \vec d} \right)_{D} \ket{\varphi}_D \ot \left[ \left( \sigma_X(\vec a) \sigma_Z(\vec b) \right)_A \left( \ket{\phi^{(0,0)}}^{\ot n} \right)_{AQ}  \right] \\
\intertext{Recalling from \ref{def:observables} that $(-1)^{\vec a \cdot \hat u(\vec k, \vec y, \vec d)} \tilde X(\vec a)_{\vec y, \vec d} = X(\vec a)_{\vec y, \vec d}$:}
&= \E_{a, b \in \bits^n} \left( X(\vec a)_{\vec y, \vec d} Z(\vec b)_{\vec y, \vec d} \right)_{D} \ket{\varphi}_D \ot \left[ \left( \sigma_X(\vec a) \sigma_Z(\vec b) \right)_A \left( \ket{\phi^{(0,0)}}^{\ot n} \right)_{AQ}  \right] \\
&= V \ket{\varphi}_D \,.
\end{align*}
\end{proof}

\begin{lemma} \label{lem:operator_rounding_tilde}
For an efficient perfect device $D = (S, \Pi, M, P)$ and any $\vec a, \vec b \in \bits^n$ we have
\begin{align}
\tr{ \tilde V^\dagger  \left(\sigma_X(\vec a) \sigma_Z(\vec b) \right)_Q^\dagger \tilde V \tilde X(\vec a)_{DYR} Z(\vec b)_{DYR} \sigma^{(\vec 1)}_{DYR} } \approx_{n^{1/2} \gamma_H(D)^{1/8}} 1 \,.
\end{align}
\end{lemma}

\begin{proof}
Inserting the definition of $\tilde V$:
\begin{align*}
& \tr{ \tilde V^\dagger  \left(\sigma_X(\vec a) \sigma_Z(\vec b) \right)_Q^\dagger \tilde V \tilde X(\vec a)_{DYR} Z(\vec b)_{DYR} \sigma^{(\vec 1)}_{DYR} } \\
&= \E_{\vec a', \vec b'} \tr{ \tilde V^\dagger  \left( \tilde X(\vec a') Z(\vec b') \tilde X(\vec a) Z(\vec b) \sigma^{(\vec 1)} \right)_{DYR} \ot \left( \sigma_X(\vec a') \sigma_Z(\vec b') \right)_A \ot \left( \sigma_X(\vec a) \sigma_Z(\vec b)\right)^\dagger_Q \left( \ket{\phi^{(0,0)}}^{\ot n} \right)_{AQ}}
\intertext{Using $\left( \sigma_X(\vec a) \sigma_Z(\vec b)\right)^\dagger_Q \left( \ket{\phi^{(0,0)}}^{\ot n} \right)_{AQ} = \left( \sigma_X(\vec a) \sigma_Z(\vec b)\right)_A \left( \ket{\phi^{(0,0)}}^{\ot n} \right)_{AQ}$ and the Pauli group relation $\sigma_X(\vec a') \sigma_Z(\vec b') \sigma_X(\vec a) \sigma_Z(\vec b) = (-1)^{\vec a \cdot \vec b'} \sigma_X(\vec a + \vec a') \sigma_Z(\vec b + \vec b')$:}
&= \E_{\vec a', \vec b'} (-1)^{\vec a \cdot \vec b'} \tr{ \tilde V^\dagger  \left( \tilde X(\vec a') Z(\vec b') \tilde X(\vec a) Z(\vec b) \sigma^{(\vec 1)} \right)_{DYR} \ot \left( \sigma_X(\vec a + \vec a') \sigma_Z(\vec b + \vec b') \right)_A \left( \ket{\phi^{(0,0)}}^{\ot n} \right)_{AQ}}
\intertext{Using \ref{lem:approx_rep_tilde} with \ref{lem:replace_in_trace} to exchange the order of $\tilde X(\vec a)$ and $Z(\vec b)$ (which act directly on the state in the above expression), and then combining $Z(\vec b) Z(\vec b') = Z(\vec b + \vec b')$:}
&\approx_{n^{1/2} \gamma_H(D)^{1/8}} \E_{\vec a', \vec b'} (-1)^{\vec a \cdot \vec b' + \vec a \cdot \vec b} \tr{ \tilde V^\dagger  \left( \tilde X(\vec a') Z(\vec b + \vec b') \tilde X(\vec a) \sigma^{(\vec 1)} \right)_{DYR} \ot \left( \sigma_X(\vec a + \vec a') \sigma_Z(\vec b + \vec b') \right)_A \left( \ket{\phi^{(0,0)}}^{\ot n} \right)_{AQ}}
\intertext{We can now again use \ref{lem:approx_rep_tilde} with \ref{lem:replace_in_trace} to exchange the order of $Z(\vec b + \vec b')$ and $\tilde X(\vec a)$ (and note that this will cancel the $(-1)^{\vec a \cdot \vec b' + \vec a \cdot \vec b}$ pre-factor in the above expression), and combine $X(\vec a) X(\vec a') = X(\vec a + \vec a')$:}
&\approx_{n^{1/2} \gamma_H(D)^{1/8}} \E_{\vec a', \vec b'} \tr{ \tilde V^\dagger  \left( \tilde X(\vec a + \vec a') Z(\vec b + \vec b') \sigma^{(\vec 1)} \right)_{DYR} \ot \left( \sigma_X(\vec a + \vec a') \sigma_Z(\vec b + \vec b') \right)_A \left( \ket{\phi^{(0,0)}}^{\ot n} \right)_{AQ}}
\intertext{If we now shift the indices in the expectation $\vec a' \mapsto \vec a' - \vec a$ and $\vec b' \mapsto \vec b' - \vec b$, then this simplifies to:}
&= \tr{\tilde V^\dagger \tilde V \sigma^{(\vec 1)}} \\
&= 1 \,.
\end{align*}
\end{proof}

We can now combine \ref{lem:relation_V_tilde_notilde} and \ref{lem:operator_rounding_tilde} to show that the isometry $V$ maps the observables $X(\vec a) Z(\vec b)$ to the corresponding Pauli observables.
\begin{proposition} \label{lem:operator_rounding_no_tilde}
For an efficient perfect device $D = (S, \Pi, M, P)$ and any $\vec a, \vec b \in \bits^n$ we have
\begin{align}
V X(\vec a) Z(\vec b) V^\dagger &\approx_{n^{1/2} \gamma_H(D)^{1/8}, V \sigma^{(\vec 1)} V^\dagger } \left(  \sigma_X(\vec a) \sigma_Z(\vec b) \right)_Q \ot \1_{YRDA} \,.
\end{align}
\end{proposition}

\begin{proof}
By \ref{lem:state_dep_distance_expanded}, it suffices to show that 
\begin{align}
\tr{\left(  \sigma_X(\vec a) \sigma_Z(\vec b) \right)_Q^\dagger  V X(\vec a) Z(\vec b) V^\dagger V \sigma^{(\vec 1)} V^\dagger} \approx_{n^{1/2} \gamma_H(D)^{1/8}} 1 \,.
\end{align}
For this, we perform the following calculation.
Using $V^\dagger V = \1$, cyclicity of the trace, and tracing over the registers $Y$ and $R$:
\begin{align*}
& \tr{\left(  \sigma_X(\vec a) \sigma_Z(\vec b) \right)_Q^\dagger  V X(\vec a) Z(\vec b) V^\dagger V \sigma^{(\vec 1)} V^\dagger} \\
&= \sum_{\vec y, \vec d} \tr{ V_{\vec y, \vec d}^\dagger \left(\sigma_X(\vec a) \sigma_Z(\vec b) \right)_Q^\dagger V_{\vec y, \vec d} X(\vec a)_{\vec y, \vec d} Z(\vec b)_{\vec y, \vec d} \sigma^{(\vec 1)}_{\vec y, \vec d} } \\
\intertext{Using \ref{lem:relation_V_tilde_notilde}:}
&= \sum_{\vec y, \vec d} \tr{\tilde V_{\vec y, \vec d}^\dagger \sigma_Z(\hat u(\vec k, \vec y, \vec d))_{A} \ot \sigma_Z(\hat u(\vec k, \vec y, \vec d))_{Q} \left(\sigma_X(\vec a) \sigma_Z(\vec b) \right)_Q^\dagger V_{\vec y, \vec d} X(\vec a)_{\vec y, \vec d} Z(\vec b)_{\vec y, \vec d} \sigma^{(\vec 1)}_{\vec y, \vec d} } \\
\intertext{Exchanging the order of $\sigma_Z(\hat u(\vec k, \vec y, \vec d))_{A} \ot \sigma_Z(\hat u(\vec k, \vec y, \vec d))_{Q}$ and $\left(\sigma_X(\vec a) \sigma_Z(\vec b) \right)_Q^\dagger$ (which produces a factor of $(-1)^{a \cdot \hat u(\vec k, \vec y, \vec d)}$):}
&= \sum_{\vec y, \vec d} (-1)^{\vec a \cdot \hat u(\vec k, \vec y, \vec d)} \tr{\tilde V_{\vec y, \vec d}^\dagger \left(\sigma_X(\vec a) \sigma_Z(\vec b) \right)_Q^\dagger \left[ \sigma_Z(\hat u(\vec k, \vec y, \vec d))_{A} \ot \sigma_Z(\hat u(\vec k, \vec y, \vec d))_{Q} V_{\vec y, \vec d}  \right]  X(\vec a)_{\vec y, \vec d}  Z(\vec b)_{\vec y, \vec d} \sigma^{(\vec 1)}_{\vec y, \vec d} } \\
\intertext{By \ref{lem:relation_V_tilde_notilde}, the expression in square brackets is simply $\tilde V_{\vec y, \vec d}$. Additionally recalling from \ref{def:observables} that $(-1)^{\vec a \cdot \hat u(\vec k, \vec y, \vec d)} X(\vec a)_{\vec y, \vec d} = \tilde X(\vec a)_{\vec y, \vec d}$:}
&= \sum_{\vec y, \vec d} \tr{\tilde V_{\vec y, \vec d}^\dagger \left(\sigma_X(\vec a) \sigma_Z(\vec b) \right)_Q^\dagger \tilde V_{\vec y, \vec d} \tilde X(\vec a)_{\vec y, \vec d} Z(\vec b)_{\vec y, \vec d} \sigma^{(\vec 1)}_{\vec y, \vec d} } \\
\intertext{Finally, we can re-introduce the systems $Y$ and $R$ and use \ref{lem:operator_rounding_tilde} to obtain:}
&= \tr{\tilde V^\dagger \left(\sigma_X(\vec a) \sigma_Z(\vec b) \right)_Q^\dagger \tilde V \tilde X(\vec a) Z(\vec b) \sigma^{(\vec 1)}} \\
&\approx_{n^{1/2} \gamma_H(D)^{1/8}} 1 \,. 
\end{align*}
\end{proof}

\subsection{Preparing BB84 states} \label{sec:bb84}

\begin{lemma} \label{lem:pauli_projectors_one}
For an efficient perfect  device $D = (S, \Pi, M, P)$ and any $\vec \theta \in \bits^n$:
\begin{enumerate}
\item If $\theta_i = 0$, then
\begin{align}
\sum_{\vec v} \proj{\vec v} \ot \left( \sigma_{Z,i}^{(v_i)} \right)_Q \approx_{n^{1/4} \gamma_H(D)^{1/16}, \sum_{\vec v'} \proj{\vec v'} \ot V \sigma^{(\vec \theta, \vec v')} V^\dagger}  \1 \,.
\end{align}
\item If $\theta_i = 1$, then
\begin{align}
\sum_{\vec v} \proj{\vec v} \ot \left( \sigma_{X,i}^{(v_i)} \right)_Q \approx_{n^{1/4} \gamma_H(D)^{1/16}, \sum_{\vec v'} \proj{\vec v'} \ot V \sigma^{(\vec \theta, \vec v')} V^\dagger} \1 \,.
\end{align}
\end{enumerate}
\end{lemma}

\begin{proof}
We first prove the first statement.
It is easy to check that $\sum_{\vec v} \proj{\vec v} \ot \left( \sigma_{Z,i}^{(v_i)} \right)_Q$ is a projector, so we can expand the definition of the state-dependent distance and compute:
\begin{align*}
&\tr{ \left( \sum_{\vec v} \proj{\vec v} \ot \left( \sigma_{Z,i}^{(v_i)} \right)_Q - \1 \right)^\dagger \left( \sum_{\vec v} \proj{\vec v} \ot \left( \sigma_{Z,i}^{(v_i)} \right)_Q - \1 \right) \sum_{\vec v'} \proj{\vec v'} \ot V \sigma^{(\vec \theta, \vec v')} V^\dagger} \\
&= \tr{ \left( \1 - \sum_{\vec v} \proj{\vec v} \ot \left( \sigma_{Z,i}^{(v_i)} \right)_Q \right) \sum_{\vec v'} \proj{\vec v'} \ot V \sigma^{(\vec \theta, \vec v')} V^\dagger} \\
&= 1 - \sum_{\vec v} \tr{ \left( \proj{\vec v} \ot \left( \sigma_{Z,i}^{(v_i)} \right)_Q \right) \sum_{\vec v'} \proj{\vec v'} \ot V \sigma^{(\vec \theta, \vec v')} V^\dagger} \\
&= 1 - \sum_{\vec v} \tr{ \left( \sigma_{Z,i}^{(v_i)} \right)_Q V \sigma^{(\vec \theta, \vec v)} V^\dagger } \\
&= 1 - \sum_{v_i \in \bits} \; \tr{ \left( \sigma_{Z,i}^{(v_i)} \right)_Q V \sigma^{(\vec \theta, v_i, \vec 1^i)} V^\dagger} \,,
\end{align*}
where for the last line we used that $\sum_{v_1, \dots, v_{i-1}, v_{i+1}, \dots, v_n \in \bits} \sigma^{(\vec \theta, \vec v)} = \sigma^{(\vec \theta, v_i, \vec 1^i)}$.
Therefore, to show the first part of the lemma, we need to show that 
\begin{align}
\sum_{v_i \in \bits} \; \tr{ \left( \sigma_{Z,i}^{(v_i)} \right)_Q V \sigma^{(\vec \theta, v_i, \vec 1^i)} V^\dagger} \approx_{n^{1/4} \gamma_H(D)^{1/16}} 1\,.
\end{align}
For this, recall from \ref{lem:operator_rounding_no_tilde} that we have 
\begin{align}
V Z_i V^\dagger &\approx_{n^{1/2} \gamma_H(D)^{1/8}, V \sigma^{(\vec 1)} V^\dagger } \left( \sigma_{Z,i} \right)_Q \ot \1_{YRDA} \,.
\end{align}
Since $V$ and $Z_i$ are efficient, by \ref{lem:partial_lifting} this implies that for any $\vec \theta$,
\begin{align}
V Z_i V^\dagger &\approx_{n^{1/2} \gamma_H(D)^{1/8}, V \sigma^{(\vec \theta)} V^\dagger } \left( \sigma_{Z,i} \right)_Q \ot \1_{YRDA} \,.
\end{align}
(Compared to \ref{lem:partial_lifting}, here we have an additional isometry $V$ applied to the state in the approximation. However, because $V$ is efficient, it is easy to see that the proof of \ref{lem:partial_lifting} still goes through.)
Using \cite[Lemma 2.18(ii)]{self-testing} and \cite[Lemma 2.24]{self-testing}, we get for any $v_i \in \bits$:
\begin{align}
V Z_i^{(v_i)} V^\dagger &\approx_{n^{1/2} \gamma_H(D)^{1/8}, V \sigma^{(\vec \theta, v_i, \vec 1^i)} V^\dagger } \left( \sigma_{Z,i}^{(v_i)} \right)_Q \ot \1_{YRDA} \,.
\end{align}
Using the replacement lemma (\ref{lem:replace_in_trace}), we obtain 
\begin{align}
\sum_{v_i \in \bits} \tr{ \left( \sigma_{Z,i}^{(v_i)} \right)_Q V \sigma^{(\vec \theta, v_i, \vec 1^i)} V^\dagger} 
&\approx_{n^{1/4} \gamma_H(D)^{1/16}} \sum_{v_i \in \bits} \tr{ V Z_i^{(v_i)} V^\dagger V \sigma^{(\vec \theta, v_i, \vec 1^i)} V^\dagger} \\
& = \sum_{v_i \in \bits} \tr{ Z_i^{(v_i)} \sigma^{(\vec \theta, v_i, \vec 1^i)} } \\
& \approx_{\gamma_H(D)} 1 \,,
\end{align}
where the last line follows from \ref{eqn:z_vec_succprob_mixed_theta} because $\theta_i = 0$. This finishes the proof of the first statement.

For the second statement, we can perform the same calculation, but use \ref{eqn:x_vec_succprob_mixed_theta} instead of \ref{eqn:z_vec_succprob_mixed_theta}.
\end{proof}

We are now in a position to show that on average over $\vec v$ and under the isometry $V$, an efficient perfect device must have prepared a product of BB84 states tensored with an additional state $\alpha^{(\vec \theta, \vec v)}$.
\begin{lemma} \label{lem:bb84_mixed_alpha}
For an efficient perfect device $D = (S, \Pi, M, P)$ and any $\vec \theta \in \bits^n$, there exists a set of subnormalised states $\{\alpha^{(\vec \theta, \vec v)}\}_{\vec v \in \bits^n}$ such that
\begin{align}
\sum_{\vec v \in \bits^n} \proj{\vec v} \ot V \sigma^{(\vec \theta, \vec v)} V^\dagger &\approx_{n^{5/4} \gamma_H(D)^{1/16}} \sum_{\vec v \in \bits^n} \proj{\vec v} \ot \left( \bigotimes_{i} H^{\theta_i} \proj{v_i} H^{\theta_i}  \right)_Q \ot \alpha^{(\vec \theta, \vec v)} \,.
\end{align}
\end{lemma}

\begin{proof}
We define the shorthand 
\begin{align*}
M(\theta) = \begin{cases}
Z & \text{if $\theta = 0$\,,} \\
X & \text{if $\theta = 1$\,.} \\
\end{cases}
\end{align*}
We can then apply \ref{lem:pauli_projectors_one} and \cite[Lemma 2.22]{self-testing} twice to get 
\begin{align*}
&\sum_{\vec v \in \bits^n} \proj{\vec v} \ot V \sigma^{(\vec \theta, \vec v)} V^\dagger \\
&\approx_{n^{1/4} \gamma_H(D)^{1/16}} \left( \sum_{\vec v} \proj{\vec v} \ot \left( \sigma_{M(\theta_1),1}^{(v_1)} \right)_Q \right) \sum_{\vec v \in \bits^n} \proj{\vec v} \ot V \sigma^{(\vec \theta, \vec v)} V^\dagger \left( \sum_{\vec v} \proj{\vec v} \ot \left( \sigma_{M(\theta_1),1}^{(v_1)} \right)_Q \right)
\end{align*}
We can repeat this for the remaining indices $i = 2, \dots, n$.
Since we incur an approximation error of $n^{1/4} \gamma_H(D)^{1/16}$ for each of the $n$ steps, the total approximation error will be $n^{5/4} \gamma_H(D)^{1/16}$, so we have
\begin{align*}
&\sum_{\vec v \in \bits^n} \proj{\vec v} \ot V \sigma^{(\vec \theta, \vec v)} V^\dagger \\
&\begin{multlined}\approx_{n^{5/4} \gamma_H(D)^{1/16}} \left( \sum_{\vec v} \proj{\vec v} \ot \left( \sigma_{M(\theta_1),1}^{(v_1)} \right)_Q  \right) \dots \left( \sum_{\vec v} \proj{\vec v} \ot \left( \sigma_{M(\theta_n),n}^{(v_n)} \right)_Q  \right)
\sum_{\vec v \in \bits^n} \proj{\vec v} \ot V \sigma^{(\vec \theta, \vec v)} V^\dagger \\
\left( \sum_{\vec v} \proj{\vec v} \ot \left( \sigma_{M(\theta_1),1}^{(v_1)} \right)_Q  \right) \dots \left( \sum_{\vec v} \proj{\vec v} \ot \left( \sigma_{M(\theta_n),n}^{(v_n)} \right)_Q  \right)
\end{multlined} \\
&= \sum_{\vec v \in \bits^n} \proj{\vec v} \ot \left( \prod_i \sigma_{M(\theta_i),i}^{(v_i)} \right)_Q V \sigma^{(\vec \theta, \vec v)} V^\dagger \left( \prod_i \sigma_{M(\theta_i),i}^{(v_i)} \right)_Q \,.
\intertext{Now noting that $\prod_i \sigma_{M(\theta_i),i}^{(v_i)} = \bigotimes_{i} H^{\theta_i} \proj{v_i} H^{\theta_i}$, we obtain}
&= \sum_{\vec v \in \bits^n} \proj{\vec v} \ot \left( \bigotimes_{i} H^{\theta_i} \proj{v_i} H^{\theta_i} \right)_Q \ot \left( \bigotimes_{i} \bra{v_i} H^{\theta_i} \right)_Q V \sigma^{(\vec \theta, \vec v)} V^\dagger \left( \bigotimes_{i} H^{\theta_i} \ket{v_i}  \right)_Q \\
&\begin{multlined}=\sum_{\vec v \in \bits^n} \proj{\vec v} \ot \left( \bigotimes_{i} H^{\theta_i} \proj{v_i} H^{\theta_i} \right)_Q \\ \qquad \ot \setft{Tr}_Q \left[ \left( \bigotimes_{i} H^{\theta_i} \proj{v_i} H^{\theta_i} \right)_Q V \sigma^{(\vec \theta, \vec v)} V^\dagger \left( \bigotimes_{i} H^{\theta_i} \proj{v_i} H^{\theta_i} \right)_Q \right] \end{multlined}
\intertext{Analogously to how we added the factors $\prod_i \sigma_{M(\theta_i),i}^{(v_i)}$ in a previous step, we can now replace the factors $\left( \bigotimes_{i} H^{\theta_i} \proj{v_i} H^{\theta_i} \right)_Q$ inside the partial trace by identity, resulting in}
&\approx_{n^{1/4} \gamma_H(D)^{1/16}} \sum_{\vec v \in \bits^n} \proj{\vec v} \ot \left( \bigotimes_{i} H^{\theta_i} \proj{v_i} H^{\theta_i} \right)_Q \ot \setft{Tr}_Q \left[ V \sigma^{(\vec \theta, \vec v)} V^\dagger \right] \,.
\end{align*}
We then obtain the desired statement by defining 
\begin{align}
\alpha^{(\vec \theta, \vec v)} \deq \setft{Tr}_Q \left[ V \sigma^{(\vec \theta, \vec v)} V^\dagger  \right] \,.
\end{align}
\end{proof}

The statement of \ref{lem:bb84_mixed_alpha} says that up to an isometry, the device's state has to be (information-theoretically) close to a product of BB84 states tensored with an auxiliary state $\alpha^{(\vec \theta, \vec v)}$.
The drawback of this is that the auxiliary state depends on $\vec \theta$ and $\vec v$.
For many applications, it is crucial that the only dependence on $\vec \theta$ and $\vec v$ in the prover's state is in the form of BB84 states, so we want that the auxiliary state be independent of $\vec \theta$ and $\vec v$.
This holds in the sense that a computationally bounded device cannot learn anything about $\vec \theta$ and $\vec v$ from $\alpha^{(\vec \theta, \vec v)}$, as formalised in the following proposition.
\begin{proposition} \label{lem:bb84_same_alpha}
For an efficient perfect device $D = (S, \Pi, M, P)$, there exists an efficiently preparable state $\alpha'$ such that for any $\vec \theta \in \bits^n$:
\begin{align}
\sum_{\vec v \in \bits^n} \proj{\vec v} \ot V \sigma^{(\vec \theta, \vec v)} V^\dagger &\capprox_{n^{5/8} \gamma_H(D)^{1/32}} \frac{1}{2^n} \sum_{\vec v \in \bits^n} \proj{\vec v} \ot \left( \bigotimes_{i} H^{\theta_i} \proj{v_i} H^{\theta_i}  \right)_Q   \ot \alpha' \,. \label{eqn:bb84_same_alpha}
\end{align}
\end{proposition}

\begin{proof}
Fix $\vec \theta$. 
For any $\vec \theta'$, we define
\begin{align}
\alpha^{(\vec \theta')} = \sum_{\vec v} \alpha^{(\vec \theta', \vec v)} \,,
\end{align}
where $\alpha^{(\vec \theta', \vec v)}$ are the states in \ref{lem:bb84_mixed_alpha}.
From \ref{lem:states_indist} we know that for any $\vec \theta', \vec \theta''$, $\sigma^{(\vec \theta')} \capprox_0 \sigma^{(\vec \theta'')}$. Since $V$ is an efficient isometry, it follows from this and \ref{lem:bb84_mixed_alpha} that
\begin{align}
\alpha^{(\vec \theta')} \capprox_{n^{5/4} \gamma_H(D)^{1/16}} \alpha^{(\vec \theta'')} \,.
\end{align}
We now set $\alpha^{(\vec \theta'')} = \alpha^{(\vec 0)} \eqqcolon \alpha'$ (any other choice also works).
From the proof of \ref{lem:bb84_mixed_alpha} we see that this state $\alpha'$ is efficiently preparable because the prover's operations are efficient, so the state $\alpha'$ (which is averaged over the possible answer's the prover can give, i.e.~no post-selection is needed) can be efficiently prepared by executing the prover with basis choice $\vec 0$, applying the (efficient) isometry $V$, and tracing out additional registers.
Taking $\vec \theta' = \vec \theta \oplus \vec 1$, we have 
\begin{align}
&\frac{1}{2^n} \sum_{\vec v \in \bits^n} \proj{\vec v} \ot \left( \bigotimes_{i} H^{\theta_i} \proj{v_i} H^{\theta_i}  \right)_Q   \ot \alpha^{(\vec \theta')} \\
&\capprox_{n^{5/4} \gamma_H(D)^{1/16}} \frac{1}{2^n} \sum_{\vec v \in \bits^n} \proj{\vec v} \ot \left( \bigotimes_{i} H^{\theta_i} \proj{v_i} H^{\theta_i}  \right)_Q   \ot \alpha' \,. \label{eqn:same_alpha1}
\end{align}
From \ref{lem:bb84_mixed_alpha}, we have
\begin{align}
\sum_{\vec v \in \bits^n} \proj{\vec v} \ot V \sigma^{(\vec \theta, \vec v)} V^\dagger &\approx_{n^{5/4} \gamma_H(D)^{1/16}} \sum_{\vec v \in \bits^n} \proj{\vec v} \ot \left( \bigotimes_{i} H^{\theta_i} \proj{v_i} H^{\theta_i}  \right)_Q \ot \alpha^{(\vec \theta, \vec v)} \,. \label{eqn:same_alpha2}
\end{align}
Comparing \ref{eqn:bb84_same_alpha}, \ref{eqn:same_alpha1}, and \ref{eqn:same_alpha2}, we see that to show the lemma it suffices to show 
\begin{multline}
\sum_{\vec v \in \bits^n} \proj{\vec v} \ot \left( \bigotimes_{i} H^{\theta_i} \proj{v_i} H^{\theta_i}  \right)_Q \ot \alpha^{(\vec \theta, \vec v)} 
\\\capprox_{n^{5/4} \gamma_H(D)^{1/16}} \frac{1}{2^n} \sum_{\vec v \in \bits^n} \proj{\vec v} \ot \left( \bigotimes_{i} H^{\theta_i} \proj{v_i} H^{\theta_i}  \right)_Q   \ot \alpha^{(\vec \theta')} \,. \label{eqn:same_alpha3}
\end{multline}

To show \ref{eqn:same_alpha3}, we first note that by tracing out the first system in the statement of \ref{lem:bb84_mixed_alpha}, we get that there exists an $\eps = O(n^{5/4} \gamma_H(D)^{1/16})$ such that for any $\vec \theta''$
\begin{align}
\norm{\sum_{\vec v \in \bits^n} V \sigma^{(\vec \theta'', \vec v)} V^\dagger - \sum_{\vec v \in \bits^n} \left( \bigotimes_{i} H^{\theta''_i} \proj{v_i} H^{\theta''_i}  \right)_Q \ot \alpha^{(\vec \theta'', \vec v)}}_1^2 \leq \eps \,. \label{eqn:state_explicit_eps}
\end{align}
Supposing that \ref{eqn:same_alpha3} does not hold, there exists an efficient measurement $\{\Lambda, \1 - \Lambda\}$ such that 
\begin{multline}
\setft{Tr}\Bigg[\Lambda \Bigg( \sum_{\vec v \in \bits^n} \proj{\vec v} \ot \left( \bigotimes_{i} H^{\theta_i} \proj{v_i} H^{\theta_i}  \right)_Q \ot \alpha^{(\vec \theta, \vec v)} \\- \frac{1}{2^n} \sum_{\vec v \in \bits^n} \proj{\vec v} \ot \left( \bigotimes_{i} H^{\theta_i} \proj{v_i} H^{\theta_i}  \right)_Q \ot \alpha^{(\vec \theta')} \Bigg) \Bigg]
\geq 2\, \eps^{1/2} + \mu(\lambda) \label{eqn:lambda_assumption_long}
\end{multline}
for some non-negligible function $\mu(\lambda)$.

Since both states in the above expression are classical on the first two systems, without loss of generality we can assume that $\Lambda$ is classical on the first two systems, too, i.e.~we can write (remembering that $\vec \theta$ is fixed)
\begin{align}
\Lambda = \sum_{\vec v} \proj{\vec v} \ot \left( \bigotimes_{i} H^{\theta_i} \proj{v_i} H^{\theta_i}  \right)_Q \ot \Lambda_{\vec v} \,.
\end{align}
With this, we can rewrite \ref{eqn:lambda_assumption_long} in a more compact form: 
\begin{align}
\sum_{\vec v} \tr{\Lambda_{\vec v} \left( \alpha^{(\vec \theta, \vec v)} - \frac{1}{2^n} \alpha^{(\vec \theta')} \right)} \geq 2\, \eps^{1/2} + \mu(\lambda) \label{eqn:lambda_assumption_compact} \,.
\end{align}
Define 
\begin{align}
\Gamma = V^\dagger \left( \sum_{\vec v} \left( \bigotimes_{i} H^{\theta_i} \proj{v_i} H^{\theta_i}  \right)_Q \ot \Lambda_{\vec v} \right) V \,.
\end{align}
Since $\{\Lambda, \1 - \Lambda\}$ is an efficient measurement and $V$ an efficient isometry, $\{\Gamma, \1 - \Gamma\}$ is an efficient measurement, too.
We claim that, assuming \ref{eqn:lambda_assumption_compact} holds, we can use the measurement $\{\Gamma, \1 - \Gamma\}$ to distinguish $\sigma^{(\vec \theta)}$ from $\sigma^{(\vec \theta')}$ with advantage $\mu(\lambda)$, contradicting \ref{lem:states_indist}. 
To show this, we perform the following calculation:
\begin{align*}
&\tr{\Gamma \left( \sigma^{(\vec \theta)} - \sigma^{(\vec \theta')} \right)} \\
&= \tr{ \left( \sum_{\vec v} \left( \bigotimes_{i} H^{\theta_i} \proj{v_i} H^{\theta_i}  \right)_Q \ot \Lambda_{\vec v} \right) \left( V \sigma^{(\vec \theta)} V^\dagger  - V \sigma^{(\vec \theta')} V^\dagger  \right)} \\
\intertext{Now using \ref{eqn:state_explicit_eps} and \cite[Lemma 2.21(ii)]{self-testing} (in a non-asymptotic form, which is easily seen to hold from the proof of \cite[Lemma 2.21(ii)]{self-testing}):}
&\begin{multlined}
\geq \text{Tr} \Bigg[ \left( \sum_{\vec v} \left( \bigotimes_{i} H^{\theta_i} \proj{v_i} H^{\theta_i}  \right)_Q \ot \Lambda_{\vec v} \right) \\
\left( \sum_{\vec v' \in \bits^n} \left( \bigotimes_{i} H^{\theta_i} \proj{v'_i} H^{\theta_i}  \right)_Q \ot \alpha^{(\vec \theta, \vec v')}  - \sum_{\vec v' \in \bits^n} \left( \bigotimes_{i} H^{\theta'_i} \proj{v'_i} H^{\theta'_i}  \right)_Q \ot \alpha^{(\vec \theta', \vec v')}  \right) \Bigg]  - 2 \eps^{1/2}
\end{multlined} \\
\intertext{Since $\theta_i \neq \theta'_i$ for all $i$, we have that for all $\vec v$ and $\vec v'$, $|\bra{\vec v'} (\ot_i H^{\theta'_i} H^{\theta_i}) \ket{\vec v}|^2 = 2^{-n}$. Then using $\sum_{\vec v'} \alpha^{(\vec \theta' ,\vec v')} = \alpha^{(\vec \theta')}$:}
&= \sum_{\vec v} \tr{ \Lambda_{\vec v} \left( \alpha^{(\vec \theta, \vec v)} -  \frac{1}{2^n} \alpha^{(\vec \theta')} \right) } - 2 \eps^{1/2} \\
& \geq \mu(\lambda) \,,
\end{align*}
where the last inequality follows from \ref{eqn:lambda_assumption_compact}.
This yields the desired contradiction and finishes the proof.
\end{proof}

We are now ready to prove the main result of this part of the paper, namely that any efficient quantum prover that does not cause \ref{fig:protocol_multi_round} to abort must have prepared a product of BB84 states (up to an isometry and an additional state $\alpha$).

\begin{theorem} \label{thm:final_rigidity}
Suppose the verifier executes \ref{fig:protocol_multi_round} with an efficient quantum prover with parameters $n$ (number of BB84 states the verifier wishes to prepare), $N = M^2$ (maximum possible number of test rounds), and $\delta$ (error tolerance, i.e.~proportion of test rounds that a prover is allowed to fail) that satisfy $N \geq C \frac{1}{n^{5/4} \delta^{2+1/32}}$ for a sufficiently large constant $C$.
We denote by $\sigma_{SWDYR}^{(\vec\theta)}$ the verifier and prover's joint final state at the end of \ref{fig:protocol_multi_round}, where $\vec \theta$ is the basis choice recorded by the verifier in \ref{fig:protocol_multi_round}, $S$ is set to $\proj{\bot}$ by the verifier if the protocol aborts and $\proj{\top}$ otherwise, $W$ is the register in which the verifier records the string $\vec v$, and $DYR$ are the prover's registers.
Then, denoting the probability of success as $\pr{\top}$ and expanding 
\begin{align*}
\sigma_{SWDYR}^{(\vec\theta)} = \pr{\top} \proj{\top}_S \ot \sigma_{WDYR|\top}^{(\vec \theta)} + (1 - \pr{\top}) \proj{\bot}_S \ot \sigma_{WDYR|\bot}^{(\vec \theta)} 
\end{align*}
there exists an efficiently preparable state $\alpha'_{DYRA}$ such that for any $\vec \theta \in \bits^n$ the following holds (with $V: \H_{DYR} \to \H_{DYRAQ}$ the efficient isometry defined in \ref{def:iso}):
\begin{align}
\pr{\top} V \sigma_{WDYR|\top}^{(\vec \theta)} V^\dagger &\capprox_{n^{5/8} \delta^{1/64}} \pr{\top} \frac{1}{2^n} \sum_{\vec v \in \bits^n} \proj{\vec v}_W \ot \left( \bigotimes_{i} H^{\theta_i} \proj{v_i} H^{\theta_i}  \right)_Q   \ot \alpha'_{DYRA} \,. \label{eqn:final_rigidity_thm}
\end{align}
\end{theorem}

\begin{remark}
Computational indistinguishability in the above theorem is with respect to a distinguisher that does not know the trapdoors for the keys used in the protocol.
In particular, this includes the quantum prover with whom the verifier executed the protocol, i.e., the quantum prover itself cannot distinguish the joint state of the verifier's output $\vec v$ and its own final quantum state from the state on the r.h.s.~of  \ref{eqn:final_rigidity_thm} (up to the isometry $V$).
Also note that on the prover's systems $QDYRA$ on the r.h.s.~of \ref{eqn:final_rigidity_thm}, the only dependence on $\vec \theta$ and $\vec v$ is in form of the BB84 states $H^{\theta_i} \proj{v_i} H^{\theta_i}$, i.e.~the state $\alpha$ is the same for all $\vec \theta$ and $\vec v$.
\end{remark}

\begin{proof}
We will reduce this theorem to \ref{lem:bb84_same_alpha} in essentially the same manner as in \cite[Theorem 3.17]{rsp}, but we spell out the details for completeness.
Fix a prover strategy for \ref{fig:protocol_multi_round}. 
As a first step, we observe that we can model any such strategy as a sequence of devices $D_1, \dots, D_N$, where $D_i$ defines the prover's strategy for the $i$-th test round.
Note that these devices need not be independent, i.e.~device $D_{i+1}$ can depend on the final state of device $D_i$ after execution of the $i$-th test round.
Since the prover is computationally efficient, so are the devices $D_i$ (and the input to device $D_{i}$ from device $D_{i-1}$ is treated as an advice state as explained in \ref{rem:advice_state}).

Having fixed such sequence of devices $D_1, \dots, D_N$, we can now view an execution of \ref{fig:protocol_multi_round} as a random experiment with sample space $\Omega = \{0, \dots, M-1\} \times [M] \times \bits^n$.
We denote by $S$ and $R$ the random variables that output the first and second component of an element of $\Omega$, respectively, and by $F_i$ the (bit-valued) random variable that outputs the $i$-th entry of the the third component.
The probability measure on this space is induced by $\pr{S=s, R=r, (F_1=f_1, \dots, F_n=f_n)} = \frac{1}{M^2} \pr{F_1=f_1, \dots, F_n=f_n}$, where $\pr{F_1=f_1, \dots, F_n=f_n}$ is given by the fixed sequence of devices $D_1, \dots, D_N$ via the condition that $F_i = 1$ if the verifier sets $\texttt{flag = fail}_{\texttt{Pre}}$ or $\texttt{flag = fail}_{\texttt{Had}}$ when executing \ref{fig:protocol_test} with $D_i$, and $F_i = 0$ otherwise.

We also define random variables that denote whether this sequence of devices causes the verifier to abort on one of the $M$ blocks of rounds in \ref{fig:protocol_multi_round}: for $j \in [M]$, we define 
\begin{align*}
F^B_j = \begin{cases}
0 & \textnormal{if } \frac{1}{M} \sum_{i \in B_j} F_i \leq \delta \,,\\
1 & \textnormal{else,}
\end{cases} 
\end{align*}
where $B_j$ is defined as in \ref{fig:protocol_multi_round}.
We now denote by $\Omega' \subset \Omega$ the event that the verifier does not abort in \ref{fig:protocol_multi_round}, i.e., $\pr{\Omega'} = \pr{\top}$.
By the definition of the abort condition in \ref{fig:protocol_multi_round}, $\omega = (s, r, (f_1, \dots, f_n)) \in \Omega'$ if and only if $F^B_j(\omega) = 0$ for all $j \leq s$.
For any $\omega \in \Omega$ there can only exist one index $j^* \in [M]$ such that $F^B_{j^*}(\omega) = 1$ and $F^B_{j}(\omega) = 0$ for $j < j^*$.
Therefore, since $S$ is chosen uniformly at random by the verifier and $F_i$ are independent of $S$ (since $S$ is unknown to the prover),
\begin{align}
\pr{F^B_{S+1} = 1 | \Omega'} 
&= \frac{\pr{F^B_{S+1} = 1 \wedge \Omega'}}{\pr{\Omega'}} \\
&= \frac{\pr{F^B_{S+1} = 1 \wedge F^B_{S} = 0 \wedge \dots \wedge F^B_{1} = 0}}{\pr{\Omega'}} \\
&\leq \frac{1}{M \, \pr{\Omega'}} \,. \label{eqn:prob_F_bound}
\end{align}
Here, $F^B_{S}$ is a random variable defined in the obvious way, i.e., 
\begin{align}
F^B_{S}((s, r, (f_1, \dots, f_n))) \deq F^B_{s}((s, r, (f_1, \dots, f_n))) = \begin{cases}
0 & \textnormal{if } \frac{1}{M} \sum_{i \in B_s} f_i \leq \delta \,,\\
1 & \textnormal{else.}
\end{cases} 
\end{align}
We can now bound the probability of observing $F_{MS+R} = 1$ conditioned on not aborting: 
\begin{align*}
\pr{F_{MS+R} = 1 | \Omega'} &\leq \pr{F_{MS+R} = 1 \wedge F_{S+1}^B = 0 | \Omega'} + \pr{F_{MS+R} = 1 \wedge F_{S+1}^B = 1 | \Omega'} \\
&\leq  \pr{F_{MS+R} = 1 | \Omega' \wedge F_{S+1}^B = 0} + \pr{F_{S+1}^B = 1 | \Omega'} \\
&\leq \delta + \frac{1}{M \, \pr{\Omega'}} \eqqcolon \eta \,,
\end{align*}
where the last inequality holds because for any fixed $s$, conditioned on $F_{s+1}^B = 0$ and on average over $r$, $F_{Ms+r} = 1$ cannot occur with probability more than $\delta$; the second term in that inequality follows from \cref{eqn:prob_F_bound}.
We observe that we may assume $\pr{\Omega'} \geq \Omega(n^{5/8} \delta^{1/64})$, as otherwise the theorem statement is trivially satisfied.
Since we further assumed $N = M^2 \geq \frac{C}{n^{5/4} \delta^{2+1/32}}$ for sufficiently large $C$, we get that $\frac{1}{M \, \pr{\Omega'}} \leq O(\delta)$.

The above calculation implies that for a fraction $1 - O(\sqrt{\delta})$ of values $(s,r)$, we have 
\begin{align}
\pr{F_{Ms+r} = 0 | \Omega' \wedge S = s \wedge R = r} \geq 1 - O(\sqrt{\delta}) \,. \label{eqn:good_sr}
\end{align}
Conditioned on the verifier choosing an $(s,r)$ for which the above inequality holds, we get from \ref{lem:bb84_same_alpha} that the state $\sigma^{(\vec \theta)}$ at the end of \ref{fig:protocol_multi_round} satisfies
\begin{align}
\sum_{\vec v \in \bits^n} \proj{\vec v} \ot V \sigma^{(\vec \theta, \vec v)} V^\dagger &\capprox_{n^{5/8} \delta^{1/64}} \frac{1}{2^n} \sum_{\vec v \in \bits^n} \proj{\vec v} \ot \left( \bigotimes_{i} H^{\theta_i} \proj{v_i} H^{\theta_i}  \right)_Q   \ot \alpha'_{DYRA} \,, \label{eqn:rigidity_prf1}
\end{align}
for an efficiently preparable state $\alpha'_{DYRA}$.
On the other hand, if the verifier chooses a pair $(s, r)$ that does not satisfy \cref{eqn:good_sr}, then the computational distinguishability between the prover's final state and the state on the r.h.s.~of \ref{eqn:rigidity_prf1} is at most 1.
Since the verifier chooses $(s,r)$ uniformly at random, conditioned on the event $\Omega'$ this happens with probability at most $O(\sqrt{\delta})$, so on average over the verifier's choice of $(s,r)$ and conditioned on $\Omega'$  we get an additional $O(\sqrt{\delta})$ contribution to the approximation error.
The result follows because $\sqrt{\delta} \leq O(n^{5/8} \delta^{1/64})$.
\end{proof}

While \ref{thm:final_rigidity} does list explicit parameters, we emphasise that we have made no attempt at optimising these, so they can likely be improved.
For the proof-of-principle applications in \ref{part:applications}, we typically want to choose a number of BB84 states $n$ that we want to prepare and use parameters in the protocol that are polynomial in $n$.
For that use case, we provide the following corollary, which is a simplified version of \ref{thm:final_rigidity} in terms of asymptotically polynomial quantities.

\begin{corollary}\label{cor:conditional-state-non-abort}
For any $n \in \N$, there exist choices $N = \poly(n)$ and $\delta = 1/\poly(n)$ such that if the verifier executes \ref{fig:protocol_multi_round} with this choice of parameters, the following holds.
For any basis choice $\vec \theta$ by the verifier, denoting the verifier and prover's joint final state as 
\begin{align*}
\sigma_{SWDYR}^{(\vec\theta)} = \pr{\top} \proj{\top}_S \ot \sigma_{WDYR|\top}^{(\vec \theta)} + (1 - \pr{\top}) \proj{\bot}_S \ot \sigma_{WDYR|\bot}^{(\vec \theta)} 
\end{align*}
as in \ref{thm:final_rigidity}, the state $\sigma_{WDYR|\top}^{(\vec \theta)} = \sum_{\vec v} \proj{\vec v}_W \ot \sigma_{DYR|\top}^{(\vec \theta, \vec v)}$ conditioned on acceptance satisfies 
\begin{align}
\pr{\top} V \sigma_{WDYR|\top}^{(\vec \theta)} V^\dagger &\capprox_{1/\poly(n)} \pr{\top} \frac{1}{2^n} \sum_{\vec v \in \bits^n} \proj{\vec v}_W \ot \left( \bigotimes_{i} H^{\theta_i} \proj{v_i} H^{\theta_i}  \right)_Q   \ot \alpha'_{B} \label{eqn:rigidity_cor1}
\end{align}
for an efficiently preparable state $\alpha_{B}$ with $B \deq DYRA$ and the isometry $V: DYR \to QB$ as in \ref{thm:final_rigidity}.
If furthermore the prover's success probability is lower-bounded by an inverse polynomial, i.e., $\pr{\top} \geq 1/\poly(n)$, then we can choose $N = \poly(n)$ and $\delta = 1/\poly(n)$ such that an analogous statement also holds for the normalised conditional state:
\begin{align}
V \sigma_{WDYR|\top}^{(\vec \theta)} V^\dagger &\capprox_{1/\poly(n)} \frac{1}{2^n} \sum_{\vec v \in \bits^n} \proj{\vec v}_W \ot \left( \bigotimes_{i} H^{\theta_i} \proj{v_i} H^{\theta_i}  \right)_Q   \ot \alpha'_{B} \,. \label{eqn:rigidity_cor2}
\end{align}
\end{corollary}
\begin{proof}
We can choose $\delta = 1/\poly(n)$ such that $n^{5/8} \delta^{1/64} = 1/\poly(n)$.
Then, we can choose $N = O(1/(n^{5/4} \delta^{1 + 1/32})) = \poly(n)$ as in \ref{thm:final_rigidity} and obtain \ref{eqn:rigidity_cor1} directly from \ref{thm:final_rigidity}.
To show \ref{eqn:rigidity_cor2}, we note that by renormalising (and noting that in the definition of computational indistinguishability, the distinguishing advantage for subnormalised states scales proportionally to the normalisation factor), we immediately get 
\begin{align}
V \sigma_{WDYR|\top}^{(\vec \theta)} V^\dagger &\capprox_{1/(\pr{\top} \poly(n))} \frac{1}{2^n} \sum_{\vec v \in \bits^n} \proj{\vec v}_V \ot \left( \bigotimes_{i} H^{\theta_i} \proj{v_i} H^{\theta_i}  \right)_Q   \ot \alpha'_{B} \,. \label{eqn:rigcor1}
\end{align}
To finish the proof, we observe that for any $\pr{\top} \geq 1/\poly(n)$ we can choose a sufficiently small $\delta = 1/\poly(n)$ such that in the approximation error in \ref{eqn:rigcor1}, $\pr{\top} \poly(n) \geq \poly(n)$.
\end{proof}

\newpage
\part{Applications} \label{part:applications}

\section{Unclonable quantum encryption} 

As a first application of our parallel RSP protocol, we consider the notion of \emph{unclonable quantum encryption}. This cryptographic functionality was coined by Gottesman~\cite{gottesman2004uncloneable} and subsequently formalised by Broadbent and Lord~\cite{BroadbentL19}. Intuitively, the idea is to create a quantum ciphertext, corresponding to the encryption of a classical message, that cannot be duplicated. In other words, given the quantum ciphertext, no adversary can create two ciphertexts that decrypt to the same message as the original ciphertext. Let us formalise this intuition by introducing a few relevant definitions.

A private-key quantum encryption of classical messages (QECM) scheme is a procedure that takes as input a key and a plaintext in the form of a classical bit string in a quantum register, and produces a
ciphertext in the form of a quantum state. We formalise this notion in \ref{def:QECM}.

\begin{definition}[$\QECM$ scheme]\label{def:QECM}\ \\
Let $\lambda \in \N$ be the security parameter. A quantum encryption of classical messages (QECM) scheme with key space $\algo K$ and plaintext space $\algo M$ is a triplet $\QECM = (\KeyGen,\Enc,\Dec)$ consisting of $\QPT$ algorithms:
\begin{itemize}
    \item $\QECM.\KeyGen(1^\lambda) \rightarrow k$: takes as input the security parameter and outputs a key $k \in \algo K$.
    \item $\QECM.\Enc(k,m) \rightarrow \rho$: takes as input a key $k$ and a message $m \in \algo M$ and outputs $\rho \in \algo D(\algo H_A)$.
    \item $\QECM.\Dec(k,\rho) \rightarrow \sigma$: takes as input a key $k$ and a quantum ciphertext $\rho \in \algo D(\algo H_A)$, and outputs a plaintext in the form of a state $\sigma \in \algo D(\algo H_M)$, where $\algo H_M = \setft{span}\{ \ket{m} : m \in \algo M\}$.
\end{itemize}
We use the notation $\QECM.\Enc_k$ for the map $m \mapsto \QECM.\Enc(k, m)$, and likewise for $\QECM.\Dec_k$.
The map $\QECM.\Enc_k$ can be extended to quantum inputs $\sigma \in \mD(\H_M)$ with $\algo H_M = \setft{span}\{ \ket{m} : m \in \algo M\}$ via
$$
\Enc_k(\sigma) = \sum_{m \in \algo M} \tr{\proj{m}\sigma} \Enc_k(\proj{m}) \,.
$$
A QECM scheme is called \emph{correct} if, for all $m \in \algo M$ and $k \in \supp\, \QECM.\KeyGen(1^\lambda)$,
$$\tr{\proj{m} \QECM.\Dec_k \circ \QECM.\Enc_k(m)} \geq 1 - \negl(\lambda).$$
\end{definition}
Note that we typically assume that the key space is given by $\algo K = \bit^\lambda$ and that the plaintext space $\algo M$ and ciphertext space consist of inputs of at most $\poly(\lambda)$ many bits and qubits, respectively.

We use the following notion of indistinguishable encryptions for QECM schemes which is inspired by the definition of Alagic et al.~\cite{Alagic_2016}.
\begin{definition}[Indistinguishable security]\label{def:ind}
A QECM scheme $\Sigma=(\KeyGen,\Enc,\Dec)$ has \emph{indistinguishable encryptions} (or is $\IND$-secure) if, for every $\QPT$ adversary $\algo A = (\algo M_{\mA},\algo D)$ consisting of an (adversarial) message sampling procedure $\mM_\mA$ and a distinguisher $\mD$,
$$
\Big| \Pr\left[\algo D \big((\Enc_k \otimes \id_E)\rho_{ME}\big)=1 \right] - \Pr\left[\algo D \big((\Enc_k \otimes \id_E)(\proj{0}_M \otimes \rho_E)\big)=1 \right] \Big| \leq \negl(\lambda)\,,
$$
where we assume that $k \leftarrow \KeyGen(1^\lambda)$, $\rho_{ME} \leftarrow \algo M_{\mA}(1^\lambda)$ with $\rho_E = \Tr_M[\rho_{ME}]$, and where $\proj{0}_M$ is the all-0 string in the plaintext register $M$. 
\end{definition}

Informally, we say that a quantum encryption scheme is \emph{unclonable} if no $\QPT$ adversary can produce two ``copies'' of a quantum ciphertext which can each be decrypted with access to the private key. Before we make the notion of unclonable ciphertexts more precise, let us first introduce the following definition of a cloning attack due to Broadbent and Lord \cite{BroadbentL19}.

\begin{definition}[Cloning attack]\label{def:cloning-attack}

Let $\lambda \in \N$ be a parameter. A cloning attack $(\mathcal{A},\mathcal{B},\mathcal{C})$ against a QECM scheme $\Sigma = (\KeyGen,\Enc,\Dec)$ consists of the following $\QPT$ algorithms (which are parameterised by $\lambda$)
\begin{itemize}
    \item (cloning map) $\mathcal{A}: \mathcal{D}(\mathcal{H}_{A}) \rightarrow \mathcal{D}(\mathcal{H}_{B} \otimes \mathcal{H}_{C})$
    \item (1st decoder) $\mathcal{B}: \algo K \times \mathcal{D}( \mathcal{H}_{B}) \rightarrow \mathcal{D}(\mathcal{H}_{M})$
    \item (2nd decoder) $\mathcal{C}: \algo K \times \mathcal{D}( \mathcal{H}_{C}) \rightarrow \mathcal{D}(\mathcal{H}_{M})$
\end{itemize}
where $\algo K$ and $\mathcal{H}_{A}$ are defined by the scheme $\Sigma$, i.e.~$\mK$ is the set of keys and $\mathcal{H}_{A}$ is the ciphertext system.
\end{definition}

We remark that we only consider \emph{efficient} cloning attacks throughout this work. This is in contrast with the definition of Broadbent and Lord \cite{BroadbentL19} who consider $\CPTP$ maps more generally.

\begin{definition}[Cloning experiment]\label{def:cloning-experiment}
Let $\Sigma = (\KeyGen,\Enc,\Dec)$ be a QECM scheme and let $\lambda \in \N$ be the security parameter. We define the following security game, called the $\emph{cloning experiment}$, which takes place between a challenger and an efficient adversary who executes a cloning attack $(\algo A, \algo B,\algo C)$:
 \begin{enumerate}[label=\arabic*.]
     \item The challenger samples $k \leftarrow \KeyGen(1^\lambda)$ and $m \rand \algo M$, and sends $\rho_A \leftarrow \Enc_k(m)$ to $\algo A$.
     \item $\algo A$ maps $\rho_A$ into a bipartite state $\rho_{BC}$ on systems $BC$, and sends $\rho_{BC}$ to the challenger together with a pair of circuit descriptions of ensembles of efficient quantum algorithms $\{\algo B_{\kappa}\}_{\kappa \in \algo K}$ and $\{\algo C_{\kappa}\}_{\kappa \in \algo K}$.
     \item The challenger runs $\algo B_k$ on system $B$ and $\algo C_k$ on system $C$ of $\rho_{BC}$, measures the output states in the computational basis to obtain outcomes $m_B$ and $m_C$, and outputs $1$ if $m=m_B=m_C$, and $0$ otherwise.
 \end{enumerate}
We let the random variable $\mathsf{CloneExp}_{\Sigma}\big(1^\lambda, (\algo A,\algo B,\algo C)\big)$ denote the output bit of the challenger.
\end{definition}

Building on the cloning experiment in \ref{def:cloning-experiment}, we then define unclonable security as follows.

\begin{definition}[Unclonable Security]\label{def:unclonable_security} 
Let $\lambda \in \N$ be the security parameter. We say that a QECM scheme $\Sigma = (\KeyGen,\Enc,\Dec)$ with message space $\bit^{n(\lambda)}$ is $t(\lambda)$-unclonable secure if, for all cloning attacks $ (\mathcal{A},\mathcal{B},\mathcal{C})$, it holds that
$$
\Pr \left[\mathsf{CloneExp}_{\Sigma}\big(1^\lambda, (\algo A,\algo B,\algo C)\big) =1 \right] \leq 2^{- n(\lambda) + t(\lambda)} + \negl(\lambda).
$$
\end{definition}

\subsection{Conjugate coding and hybrid encryption}

In this section, we introduce a quantum encryption scheme based on the notion of \emph{conjugate coding} which was first proposed by Wiesner~\cite{wiesner} in the context of quantum money. 
The main idea behind the scheme is that it is possible to \emph{encrypt} a plaintext $m \in \bit^\lambda$ in the form of a quantum ciphertext by first applying a classical one-time pad, and then making $\lambda$ uniformly random choices of basis (either computational or Hadamard). Broadbent and Lord~\cite{BroadbentL19} considered this scheme in the context of unclonable encryption. 

\begin{construction}[Conjugate coding scheme]\label{cons:conjugate-coding} Let $\lambda \in \N$ be the security parameter. The conjugate coding encryption scheme $\Sigma = (\KeyGen,\Enc,\Dec)$ consists of the following $\QPT$ algorithms:
\begin{description}
 \item $\Sigma.\KeyGen(1^\lambda)$: takes as input the parameter $1^\lambda$ and outputs a key $k=(\vec r,\vec \theta) \rand \bit^\lambda \times \bit^\lambda$.
    \item $\Sigma.\Enc(k,m)$: takes as input a key $k=(\vec r,\vec \theta)$ and a message $m \in \bit^\lambda$ and outputs the cipherterxt
    $$
    \rho = \bigotimes_{i=1}^\lambda H^{\theta_i} \proj{r_i \oplus m_i} H^{\theta_i}.
    $$
    \item $\Sigma.\Dec(k,\rho)$: takes as input a key $k=(\vec r,\vec \theta)$ and decrypts a quantum ciphertext $\rho$ as follows:
    \begin{enumerate}
    \item apply $H^{\theta_1} \otimes \dots \otimes H^{\theta_\lambda}$ to the state $\rho$ and measure in the standard basis with outcome $\vec x$.
    \item output the plaintext $\vec m' = \vec x \oplus \vec r$.
\end{enumerate}
\end{description}

\end{construction}

Notice that the conjugate coding scheme in \ref{cons:conjugate-coding} is trivially correct and also $\IND$-secure according to \ref{def:ind}. Moreover, the scheme
has the following desirable property which was shown by Broadbent and Lord~\cite{BroadbentL19}. Namely, the quantum ciphertext is \emph{unclonable} in the following sense:

\begin{theorem}[\cite{BroadbentL19}, Theorem 15]\label{thm_broadbent-lord}
Let $\lambda \in \N$ be the security parameter. Then, the conjugate coding encryption scheme $\Sigma = (\KeyGen,\Enc,\Dec)$ defined in \ref{cons:conjugate-coding} is $t(\lambda)$-unclonable secure, where $t(\lambda) = \lambda \log(1 + 1/\sqrt{2})$.
In other words, for all cloning attacks $(\algo A,\algo B,\algo C)$,
$$
\underset{\vec m}{\mathbb{E}} \, \underset{\vec r}{\mathbb{E}}\underset{\vec \theta}{\mathbb{E}} \,\, \Tr\left[(\proj{\vec m} \otimes \proj{\vec m})(\mathcal{B}_{(\vec r,\vec \theta)} \otimes \mathcal{C}_{(\vec r ,\vec \theta)}) \circ \mathcal{A} \circ \Enc_{(\vec r,\vec \theta)}(\vec m) \right] \leq \left(\frac{1}{2} + \frac{1}{2\sqrt{2}}\right)^\lambda.
$$
\end{theorem}

\paragraph{Hybrid encryption.}

We conclude this section with the following simple \emph{hybrid encryption} scheme which we also make use of in our section on QCP schemes with a classical client. Our hybrid encryption scheme combines a standard QECM scheme with the classical \emph{one-time pad} $\mathsf{OTP} = (\KeyGen,\Enc,\Dec)$, where an encryption of a plaintext $\vec m \in \bit^\lambda$ via a key $\vec k \in \bit^\lambda$ is generated as $\mathsf{OTP}.\Enc_k(\vec m) = \vec k \oplus \vec m$.

\begin{construction}[Hybrid encryption scheme]\label{cons:hybrid-encryption} Let $\lambda \in \N$ be the security parameter. Let $\QECM= (\KeyGen,\Enc,\Dec)$ be a QECM scheme and let $\mathsf{OTP} = (\KeyGen,\Enc,\Dec)$ be the classical one-time pad. We define the hybrid encryption scheme $\mathsf{HE}= (\KeyGen,\Enc,\Dec)$ as follows:
\begin{description}
\item $\mathsf{HE}.\KeyGen(1^\lambda)$: this is the same procedure as $\QECM.\KeyGen(1^\lambda)$.
\item $\mathsf{HE}.\Enc(k,m)$: given as input a key $k$ and a plaintext $m \in \bit^\lambda$, encrypt as follows:
\begin{enumerate}
    \item sample a random string $r \rand \bit^\lambda$.
    \item output the ciphertext $(\QECM.\Enc_k(r),\mathsf{OTP}.\Enc_r(m))$.
\end{enumerate}
\item $\mathsf{HE}.\Dec(k,c)$: given as input a key $k$ and ciphertext $c = (c_0,c_1)$, decrypt as follows:
\begin{enumerate}
    \item compute
$r' = \QECM.\Dec_k(c_0)$
\item output the plaintext $\mathsf{OTP}.\Dec_{r'}(c_1)$. 
\end{enumerate}
\end{description}
\end{construction}

\begin{proposition}\label{prop:unclonable-HE}
Let $\lambda \in \N$ and let $\Sigma$ be a correct and $t(\lambda)$-unclonable secure $\QECM$ schme. Then, the hybrid encryption scheme $\mathsf{HE}$ in \ref{cons:hybrid-encryption} instantiated with $\Sigma$ is correct and $t(\lambda)$-unclonable secure.
\end{proposition}

\begin{proof} The correctness of $\mathsf{HE} = (\KeyGen,\Enc,\Dec)$ follows from the correctness of the scheme $\Sigma$.
Let us now show that $\mathsf{HE}$ is $t(\lambda)$-unclonable secure against any efficient cloning attack $(\algo A,\algo B,\algo C)$. To this end, we perform a security reduction and relate the cloning advantage against $\mathsf{HE}$ to the cloning advantage against the conjugate coding scheme in \ref{cons:conjugate-coding}, and then apply \ref{thm_broadbent-lord}. 

We define a cloning attack $(\mathcal{A}',\mathcal{B}',\mathcal{C}')$ for \ref{cons:conjugate-coding} which depends on $(\mathcal{A},\mathcal{B},\mathcal{C})$ as follows:
\begin{itemize}
    \item $\algo A'$ takes as input a $\lambda$-qubit state $\rho_Q$, samples a random string $\vec u \rand \bit^\lambda$ and then outputs a bipartite state in systems $B'C'$ which is generated by running $\algo A_{QQ' \rightarrow BC}$ as follows:
    $$
    \algo A'_{Q \rightarrow B'C'}(\rho_Q) = \algo A(\rho_{Q} \otimes \proj{\vec u}_{Q'}) \otimes \proj{\vec u}_{\bar B} \otimes \proj{\vec u}_{\bar C},
    $$
    where we let $B' = B \bar B$ and $C' = C \bar C$.
    \item $\algo B'$ takes as input a pair of keys $(\vec r,\vec \theta) \in \bit^\lambda \times \bit^\lambda$ and a $2\lambda$-qubit state $\rho_{B \bar B}$, measures system $\bar B$ with outcome $\vec u \in \bit^\lambda$, and then runs $\algo B_{\vec u \oplus \vec r,\vec \theta}$ on input $\rho_B$.
    \item $\algo C'$ takes as input a pair of keys $(\vec r,\vec \theta) \in \bit^\lambda \times \bit^\lambda$ and a $2\lambda$-qubit state $\rho_{C \bar C}$, measures system $\bar C$ with outcome $\vec u \in \bit^\lambda$, and then runs $\algo C_{\vec u \oplus \vec r,\vec \theta}$ on input $\rho_C$.    
\end{itemize}
Using \ref{eq:unclonable-enc-real-vs-ideal} and the attack $(\mathcal{A}',\mathcal{B}',\mathcal{C}')$, we can bound the cloning advantage as follows:
\begin{align*}
&\Pr \left[\mathsf{CloneExp}_{\mathsf{HE}}\big(1^\lambda, (\algo A,\algo B,\algo C)\big) =1 \right]\\
&= \underset{\vec m}{\mathbb{E}}  \underset{\vec v}{\mathbb{E}}\underset{\vec s}{\mathbb{E}} \underset{\vec \theta}{\mathbb{E}} \,\, \Tr\Bigg[(\proj{\vec m} \otimes \proj{\vec m})(\mathcal{B}_{(\vec s,\vec \theta)} \otimes \mathcal{C}_{(\vec s,\vec \theta)}) \,\circ \nonumber \\
& \quad\quad\quad \mathcal{A} \Big( \big( \bigotimes_{i} H^{\theta_i} \proj{v_i} H^{\theta_i}  \big)_Q 
  \ot \proj{\vec v \oplus \vec s \oplus \vec m}_{Q'} \Big)\Bigg]\\
 &= \underset{\vec m}{\mathbb{E}}  \underset{\vec r}{\mathbb{E}}\underset{\vec \theta}{\mathbb{E}} \,\, \Tr\left[(\proj{\vec m} \otimes \proj{\vec m})(\mathcal{B}'_{(\vec r,\vec \theta)} \otimes \mathcal{C}'_{(\vec r ,\vec \theta)}) \circ \mathcal{A}' \left(  \left( \bigotimes_{i} H^{\theta_i} \proj{r_i \oplus m_i} H^{\theta_i}  \right)_Q  \right) \right]\\
\intertext{We can now use the fact that \ref{thm_broadbent-lord} must hold for the cloning attack $(\algo A',\algo B',\algo C')$, and get:}
&\leq \left(\frac{1}{2} + \frac{1}{2\sqrt{2}}\right)^{\lambda}
= \quad 2^{-\lambda + t(\lambda)}.
\end{align*}
This proves the claim.
\end{proof}

\subsection{Unclonable quantum encryption with a classical client}

Let us now extend the notion of quantum encryption schemes of classical messages from \ref{def:QECM} to allow for interaction between a classical client and a quantum receiver. In this setting, the encryption algorithm is replaced by an interactive protocol between a client and a receiver which enables the classical delegation of a quantum ciphertext. Our definition is the following:

\begin{definition}[QECM scheme with a classical client]\label{def:QECM_cc}\ \\
Let $\lambda \in \N$ be a security parameter, let $\algo K$ be the key space and let $\algo M$ be the plaintext space. A QECM scheme with a classical client (i.e., a QECM$_{\text{CC}}$ scheme) is a tuple $\QECM_{\mathsf{CC}} = (\KeyGen,\Enc,\Dec)$, where
\begin{itemize}
    \item $\QECM_{\mathsf{CC}}.\KeyGen(1^\lambda) \rightarrow k$: takes as input the security parameter $1^\lambda$ and outputs a key $k \in \algo K$.
    \item $\QECM_{\mathsf{CC}}.\Enc(\algo C(1^\lambda,k,m),\algo R(1^\lambda)) \rightarrow \rho \, \textbf{or} \, \bot$: this is an interactive protocol between a classical client $\algo C$ which takes as input the security parameter $1^\lambda$, a key $k \in \algo K$ and message $m \in \algo M$, and a quantum receiver $\algo R$ which takes as input the parameter $1^\lambda$. The protocol takes place as follows:
    \begin{itemize}
        \item $\algo C$ and $\algo R$ exchange classical messages only. 
        Once the protocol is complete, $\algo C$ obtains a flag which is either $\top \,\text{(accept)}$ or $\bot \,\text{(reject)}$.
        Provided that the protocol is successful with the flag $\top$,  $\algo R$ is in possession of a state $\rho \in \algo D(\algo H_A)$ (which depends on $k$ and $m$). Otherwise, $\algo C$ outputs $\bot$.
    \end{itemize}
    \item $\QECM_{\mathsf{CC}}.\Dec(k,\rho)$: takes as input a secret key $k$, a state $\rho \in \algo D(\algo H_A)$ and outputs a plaintext in the form of a quantum state $\sigma \in \algo D(\algo H_M)$, where $\algo H_M = \setft{span}\{ \ket{m} : m \in \algo M\}$.
\end{itemize}
We say that $\QECM_{\mathsf{CC}} = (\KeyGen,\Enc,\Dec)$ is correct if, for all plaintexts $m \in \algo M$ and for all $k \in \algo K$,
$$
\tr{\proj{m} \QECM_{\mathsf{CC}}.\Dec_{k} (\rho)} \geq 1- \negl(\lambda),
$$
where $\rho \leftarrow \QECM_{\mathsf{CC}}.\Enc(\algo C(1^\lambda,k,m),\algo R(1^\lambda))$ is the final result of the interactive protocol between the client and the receiver, provided that the protocol is successful.
\end{definition}

Similar to \ref{def:cloning-attack}, we now define an analogous security experiment in the classical client setting. 

\begin{definition}[Cloning experiment with a classical client]\label{def:cloning-experiment_cc}
    Let $\Sigma = (\KeyGen,\Enc,\Dec)$ be a QECM$_{\text{CC}}$ scheme and let $\lambda \in \N$ be the security parameter. We define the security game (called the \emph{cloning experiment with a classical client}) which takes place between a challenger and an adversary $(\algo P,\algo A, \algo B,\algo C)$ consisting of an efficient interactive prover $\algo P$ and an efficient cloning attack $(\algo A, \algo B,\algo C)$ as follows:
 \begin{enumerate}[label=\arabic*.]
     \item The challenger samples a key $k \leftarrow \QECM_{\mathsf{CC}}.\KeyGen(1^\lambda)$ and message $m \rand \algo M$. To classically delegate a ciphertext according to the scheme $\Sigma$, the challenger takes the role of the classical client $\algo C$, whereas the prover $\algo P$ takes the role of the (possibly malicious) recipient in the interactive protocol specified by $\QECM_{\mathsf{CC}}.\Enc$. 
     Provided that the protocol succeeds with $\mathtt{flag}=\top$, $\algo P$ is in possession of a state $\rho \in \algo D(\algo H_A)$ (which depends on $m$ and $k$, as well as any additional information collected during the protocol). Otherwise, if $\mathtt{flag}=\bot$, the adversary loses and the challenger outputs 0.
     
     \item $\algo A$ maps $\rho_A$ into a bipartite state $\rho_{BC}$ on systems $BC$, and sends $\rho_{BC}$ to the challenger together with a pair of classical circuit descriptions of efficient quantum algorithms $\{\algo B_{\kappa}\}_{\kappa \in \algo K}$ and $\{\algo C_{\kappa}\}_{\kappa \in \algo K}$.
     
     \item The challenger applies the channel $(\algo B_{k} \otimes \algo C_{k})$ to the state $\rho_{BC}$, measures in the standard basis with outcomes $m_B$ and $m_C$, and outputs $1$ if $m=m_B=m_C$ and $\mathtt{flag} = \top$, and $0$ otherwise.
 \end{enumerate}
We let the random variable $\mathsf{CloneExp}_{\Sigma,\CC}\big(1^\lambda, (\algo P,\algo A,\algo B,\algo C)\big)$ denote the output bit of the challenger.
\end{definition}

\begin{definition}[Unclonable Security with a Classical Client]\label{def:unclonable_security_cc}\ \\
Let $\lambda \in \N$ be the security parameter. A QECM$_\text{CC}$ scheme $\Sigma = (\KeyGen,\Enc,\Dec)$ with message space $\{0,1\}^{n(\lambda)}$ is said to be $t(\lambda)$-unclonable secure if, for any (interactive) $\QPT$ adversary $(\algo P,\mathcal{A},\mathcal{B},\mathcal{C})$,
$$
\Pr \left[\mathsf{CloneExp}_{\Sigma,\CC}\big(1^\lambda, (\algo P,\algo A,\algo B,\algo C)\big) =1 \right] \leq 2^{- n(\lambda) + t(\lambda)} + \negl(\lambda).
$$
\end{definition}

\subsection{Conjugate coding encryption with a classical client}

In this section, we give a quantum encryption scheme with a classical client based on the notion of \emph{conjugate coding}, formally defined in \ref{cons:unclonable-encryption-classical-client}. \footnote{This scheme is similar to the transformation from single-decryptor schemes to unclonable encryption in~\cite{georgiou2020unclonable}. 
Indeed, our RSP protocol can also be used in a natural way to dequantise the single-decryptor scheme from~\cite{georgiou2020unclonable}, but we focus on unclonable encryption here. We thank an anonymous reviewer for pointing this out to us.}
Let us first show the correctness of our scheme.

\begin{table}[ht!]
\begin{longfbox}[breakable=false, padding=1em, padding-right=1.8em, padding-top=1.2em, margin-top=1em, margin-bottom=1em, background-color=gray!20]
\begin{protocol} {\bf Conjugate Coding Encryption with a Classical Client} \label{cons:unclonable-encryption-classical-client} \end{protocol}
Let $\lambda \in \N$ and $n=n(\lambda) \in \N$. Consider the scheme $\Sigma= (\KeyGen,\Enc,\Dec)$ defined by:
\begin{description}
\item $\Sigma.\KeyGen(1^\lambda)$: takes as input $1^\lambda$ and samples a key $k = (\vec s,\vec \theta) \rand \bit^n \times \bit^n$.
\item $\Sigma.\Enc(\algo C(1^\lambda,k, m),\algo R(1^\lambda)) \rightarrow (\sigma,\vec c) \, \textbf{or} \, \bot$: this is the following interactive protocol between a classical client $\algo C$ (which takes as input the security parameter $1^\lambda$, a key $k = (\vec s,\vec \theta)$ and a plaintext $\vec m \in \bit^n$) and a quantum receiver $\algo R$ (which takes as input the parameter $1^\lambda$):
\begin{enumerate}
    \item $\algo C$ and $\algo R$ run the parallel RSP protocol (\ref{fig:protocol_multi_round}) with parameters $\lambda$ and $n$ with choice of basis $\vec \theta \in \bit^n$. If the protocol is successful, $\algo C$ obtains $\vec v \rand \bit^n$ and $\algo R$ obtains an $n$-qubit state $\sigma$ (which depends on $\vec v,\vec{\theta}$). Otherwise, if the protocol fails, $\algo C$ outputs $\bot$.
    \item $\algo C$ sends the classical string $\vec c = \vec v \oplus \vec s \oplus \vec m$ to the receiver $\algo R$.
\end{enumerate}
\item $\Sigma.\Dec(k,\rho) \rightarrow \vec m'$: takes as input a key $k=(\vec s,\vec \theta)$ and a state $\rho = (\sigma,\vec c)$ and decrypts as follows:
\begin{enumerate}
    \item apply $H^{\theta_1} \otimes \dots \otimes H^{\theta_n}$ to the state $\sigma$ and measure in the standard basis with outcome $\vec x$.
    \item output the plaintext $\vec m' = \vec x \oplus \vec c \oplus \vec s$.
\end{enumerate}
\end{description}
\end{longfbox}
\end{table}

\begin{proposition}[Correctness]
\label{thm:correctness-CC-CC}
$\Sigma = (\KeyGen,\Enc,\Dec)$ in \ref{cons:unclonable-encryption-classical-client} is a correct QECM$_{\text{CC}}$ scheme.
\end{proposition}
\begin{proof}
Let $\lambda \in \N$ and $n(\lambda) \in \N$ be parameters and let $\Sigma = (\KeyGen,\Enc,\Dec)$ be the QECM$_{\text{CC}}$ scheme in \ref{cons:unclonable-encryption-classical-client}. To prove correctness, we have to show that, for all $m \in \bit^{n(\lambda)}$,
\begin{align}\label{eq:correctness-CC-CC}
\tr{\proj{m} \Dec_{k} (\sigma)} \geq 1- \negl(\lambda),
\end{align}
where $(\sigma,\vec c) \leftarrow \Enc(\algo C(1^\lambda,k,m),\algo R(1^\lambda))$ is the final result of the interactive protocol between the client $\algo C$ and the receiver $\algo R$. Note that in this scenario we can assume that the receiver $\algo R$ is honest throughout the protocol.
According to \ref{prop:honest-prover-correctness},
$\algo R$ is accepted in \ref{fig:protocol_multi_round} with probability negligibly close to $1$ in the security parameter $\lambda$ (for parameter choices $n$ at most polynomial in $\lambda$ and $\delta$ at least inverse polynomial in $\lambda$).
Furthermore, provided that \ref{fig:protocol_multi_round} is successful, the final state of $\algo R$ is given by
\begin{align}
\sigma = \bigotimes_{i =1}^n H^{\theta_i} \proj{v_i} H^{\theta_i} \,, 
\end{align}
where $\vec v$ and $\vec \theta$ are the strings recorded by the client $\algo C$. Hence, the correctness property in \ref{eq:correctness-CC-CC} of $\Sigma = (\KeyGen,\Enc,\Dec)$ in \ref{cons:unclonable-encryption-classical-client} immediately follows from the correctness of the conjugate coding scheme in \ref{cons:conjugate-coding}. This proves the claim.
\end{proof}

Since the quantum state $\sigma$ that $\algo R$ receives in \cref{cons:unclonable-encryption-classical-client} is independent of the message $m$ (without access to the secret key), it immediately follows that \cref{cons:unclonable-encryption-classical-client} is $\IND$-secure in the sense of \cref{def:ind}.
The following theorem shows that \cref{cons:unclonable-encryption-classical-client} also satisfies the unclonable security property from \cref{def:unclonable_security_cc}.

\begin{theorem}[Unclonable security]\label{thm:unclonable-security-CC} Let $\lambda \in \N$ and $n(\lambda) \in \N$ be parameters. Then, there exist $\eps(\lambda) = 1/\poly(\lambda)$ and $t(\lambda) = n(\lambda) \log(1 + 1/\sqrt{2})$ such that the QECM$_{\text{CC}}$ scheme $\Sigma = (\KeyGen,\Enc,\Dec)$ defined in \ref{cons:unclonable-encryption-classical-client} is $\widetilde{t}(\lambda)$-unclonable secure, where we let
$\widetilde{t}(\lambda) \deq t(\lambda) + \log\big( 1 + 2^{n(\lambda) - t(\lambda) + \log \eps(\lambda)} \big).$
In other words, for every efficient interactive quantum prover $\algo P$ and for all $\QPT$ cloning attacks $ (\mathcal{A},\mathcal{B},\mathcal{C})$,
$$
\Pr \left[\mathsf{CloneExp}_{\Sigma,\CC}\big(1^\lambda, (\algo P,\algo A,\algo B,\algo C)\big) =1 \right] \leq 2^{- n(\lambda) + t(\lambda)} + \eps(\lambda) + \negl(\lambda).
$$
\begin{proof}
Let us analyze the success probability in the cloning experiment in \ref{def:cloning-experiment}. Let $(\algo P,\mathcal{A},\mathcal{B},\mathcal{C})$ be an efficient adversary consisting of an efficient interactive prover $\algo P$ and a triplet of $\QPT$ algorithms $(\algo A,\algo B,\algo C)$. Let
$\sigma_{WDYR|\top}^{(\vec \theta)}$ be the final state at the end of the interactive protocol conditioned on the event that the challenger outputs $\mathtt{flag} = \top$ as in \ref{cor:conditional-state-non-abort}. 
Recall that the challenger has access to system $W$ whereas the prover $\algo P$ has access to systems $DYR$.
\ref{cor:conditional-state-non-abort} then considers the state $V \sigma_{WDYR|\top}^{(\vec \theta)} V^\dagger$ for an efficient isometry $V: DYR \to QS$ with $S = DYRA$.
We can argue that w.l.og.~we can assume that the prover's state is in fact $V \sigma_{WDYR|\top}^{(\vec \theta)} V^\dagger$, not $\sigma_{WDYR|\top}^{(\vec \theta)}$.
The reason is that the prover can recover the latter from the former by considering the (efficient) unitary extension $U: QS \to QS$ of $V$, applying $U^\dagger$ to $V \sigma_{WDYR|\top}^{(\vec \theta)} V^\dagger$, and then tracing out systems $QA$.
Hence, slightly abusing notation, we define 
\begin{align*}
\sigma_{WQS|\top}^{(\vec \theta)} \deq V \sigma_{WDYR|\top}^{(\vec \theta)} V^\dagger 
\end{align*}
and assume that the prover has prepared the state $\sigma_{WQS|\top}^{(\vec \theta)}$.
We can expand this state as a cq-state as in \ref{cor:conditional-state-non-abort}, and write
\begin{align}
\sigma_{WQS|\top}^{(\vec \theta)} = \sum_{\vec v \in \bit^n} \proj{\vec v}_W \otimes \sigma_{QS|\top}^{(\vec v,\vec \theta)}\,,
\end{align}
where $\sigma_{QS|\top}^{(\vec v,\vec \theta)}$ are subnormalised according to the probability of the verifier receiving outcome $\vec v$.
In \ref{protocol-CP-classical-client}, the verifier makes the basis choice $\vec \theta$ uniformly at random, and in the cloning experiment in \ref{def:cloning-experiment} the verifier chooses a uniformly random message $\vec m$.
Therefore, the success probability in the cloning experiment can be expressed as
\begin{align*}
&\Pr \left[\mathsf{CloneExp}_{\Sigma,\CC}\big(1^\lambda, (\algo P,\algo A,\algo B,\algo C)\big) =1 \right]\\
&\qquad= \sum_{\vec v} \underset{\vec m}{\mathbb{E}} \underset{\vec s}{\mathbb{E}} \underset{\vec \theta}{\mathbb{E}} \,\, \Tr\left[(\proj{\vec m} \otimes \proj{\vec m})(\mathcal{B}_{(\vec s,\vec \theta)} \otimes \mathcal{C}_{(\vec s,\vec \theta)}) \circ \mathcal{A} \left(  \sigma_{QS|\top}^{(\vec v,\vec \theta)} \otimes \proj{\vec v \oplus \vec s \oplus \vec m}_{Q'} \right)\right]. \numberthis
\end{align*}
(Here, we sum over $\vec v$ because $\sigma_{QS|\top}^{(\vec v,\vec \theta)}$ is subnormalised. Alternatively, one can also renormalise these states and take an expectation. Both methods are equivalent since the expression in the trace scales linearly in the normalisation factor.)
We can assume that the acceptance probability of $\algo P$ in the interactive protocol satisfies $\Pr[\top]\geq 1/\poly(\lambda)$ as otherwise, the success probability of the adversary $(\algo P,\mathcal{A},\mathcal{B},\mathcal{C})$ in the cloning experiment would decay faster than an inverse polynomial in $\lambda$, in which case \ref{thm:unclonable-security-CC} is trivially satisfied.
Under this assumption, it follows from \ref{cor:conditional-state-non-abort} that there exists $\eps(\lambda) = 1/\poly(\lambda)$ such that for any $\vec \theta \in \bit^n$,
\begin{align}
\sum_{\vec v \in \bit^n} \proj{\vec v}_W \otimes \sigma_{QS|\top}^{(\vec v,\vec \theta)} &\capprox_{\eps(\lambda)}  \frac{1}{2^n} \sum_{\vec v \in \bit^n} \proj{\vec v}_W \otimes   \left( \bigotimes_{i} H^{\theta_i} \proj{v_i} H^{\theta_i}  \right)_Q   \ot \alpha'_{S} \,, \label{eqn:rigidity_cor_unclonable}
\end{align}
where $\alpha'_{S}$ is a (normalised) quantum state which is independent of $\vec v$ and $\vec \theta$ (and $\vec m$). 
From this it follows that for any $\vec m, \vec s$ and $\vec \theta$,
\begin{align}
&\sum_{\vec v} \Tr\left[(\proj{\vec m} \otimes \proj{\vec m})(\mathcal{B}_{(\vec s,\vec \theta)} \otimes \mathcal{C}_{(\vec s,\vec \theta)}) \circ \mathcal{A} \left( \sigma_{QS|\top}^{(\vec v,\vec \theta)}  \otimes \proj{\vec v \oplus \vec s \oplus \vec m}_{Q'}\right) \right] \nonumber\\
&\leq \frac{1}{2^n} \sum_{\vec v} \,\, \Tr\Bigg[(\proj{\vec m} \otimes \proj{\vec m})(\mathcal{B}_{(\vec s,\vec \theta)} \otimes \mathcal{C}_{(\vec s,\vec \theta)}) \,\circ \nonumber \\
& \quad\quad\quad \mathcal{A} \Big( \big( \bigotimes_{i} H^{\theta_i} \proj{v_i} H^{\theta_i}  \big)_Q 
 \ot \alpha'_{S} \ot \proj{\vec v \oplus \vec s \oplus \vec m}_{Q'} \Big)\Bigg] + \eps(\lambda) \,. \label{eq:unclonable-enc-real-vs-ideal}
\end{align}
This is the case because the adversary $(\algo P,\mathcal{A},\mathcal{B},\mathcal{C})$ is efficient, so given the state on the l.h.s.~or r.h.s.~of \ref{eqn:rigidity_cor_unclonable}, one can efficiently estimate the quantity on the l.h.s.~or r.h.s.~of \ref{eq:unclonable-enc-real-vs-ideal}, respectively.

Since $\alpha_S'$ is a fixed and efficiently preparable state, we can consider a modified cloning map $\widetilde{\algo A}$ that only acts on systems $QQ'$ (rather than on $QQ'S$) and is defined by 
\begin{align*}
\widetilde{\algo A}(\rho_{QQ'}) = \mA \left( \rho_{QQ'} \ot \alpha_S' \right).
\end{align*}
To complete the proof, we use \ref{prop:unclonable-HE}. This allows us to argue that the hybrid encryption scheme $\mathsf{HE}$ in \ref{cons:hybrid-encryption} (which is implicit in $\Sigma$) is $t(\lambda)$-unclonable secure when instantiated with the $t(\lambda)$-unclonable secure conjugate coding scheme in \ref{cons:conjugate-coding}, where $t(\lambda) = n(\lambda) \log(1 + 1/\sqrt{2})$.

Using \ref{prop:unclonable-HE}, we can bound the cloning advantage as follows:
\begin{align*}
&\Pr \left[\mathsf{CloneExp}_{\Sigma,\CC}\big(1^\lambda, (\algo P,\algo A,\algo B,\algo C)\big) =1 \right]\\
&\leq \underset{\vec m}{\mathbb{E}}  \underset{\vec v}{\mathbb{E}}\underset{\vec s}{\mathbb{E}} \underset{\vec \theta}{\mathbb{E}} \,\, \Tr\Bigg[(\proj{\vec m} \otimes \proj{\vec m})(\mathcal{B}_{(\vec s,\vec \theta)} \otimes \mathcal{C}_{(\vec s,\vec \theta)}) \,\circ \nonumber \\
& \quad\quad\quad \mathcal{A} \Big( \big( \bigotimes_{i} H^{\theta_i} \proj{v_i} H^{\theta_i}  \big)_Q 
 \ot \alpha'_{S} \ot \proj{\vec v \oplus \vec s \oplus \vec m}_{Q'} \Big)\Bigg] + \eps(\lambda)\\
&= \underset{\vec m}{\mathbb{E}}  \underset{\vec v}{\mathbb{E}}\underset{\vec s}{\mathbb{E}} \underset{\vec \theta}{\mathbb{E}} \,\, \Tr\Bigg[(\proj{\vec m} \otimes \proj{\vec m})(\mathcal{B}_{(\vec s,\vec \theta)} \otimes \mathcal{C}_{(\vec s,\vec \theta)}) \,\circ \nonumber \\
& \quad\quad\quad \widetilde{\mathcal{A}} \Big( \big( \bigotimes_{i} H^{\theta_i} \proj{v_i} H^{\theta_i}  \big)_Q 
  \ot \proj{\vec v \oplus \vec s \oplus \vec m}_{Q'} \Big)\Bigg] + \eps(\lambda)\\
&\leq \left(\frac{1}{2} + \frac{1}{2\sqrt{2}}\right)^{n(\lambda)} + \eps(\lambda)
\quad = \quad 2^{-n(\lambda) + \widetilde{t}(\lambda)},
\end{align*}
 where we defined $\widetilde{t}(\lambda) \deq t(\lambda) + \log\big( 1 + 2^{n(\lambda) - t(\lambda) + \log \eps(\lambda)} \big)$ and $t(\lambda) = n(\lambda) \log(1 + 1/\sqrt{2})$.
\end{proof}

\end{theorem}

\section{Quantum copy-protection}

Our second application of our parallel RSP protocol for BB84 states is the task of quantum copy-protection -- an idea first proposed by Aaronson~\cite{Aar2009}. Here, we imagine that a vendor wishes to encode a program into a quantum state in a way
that enables a recipient to run the program, but not to create functionally equivalent ``pirated'' copies. Let us begin by introducing a few relevant definitions.

\begin{definition}[Quantum copy-protection scheme]\label{def:copy-protection-scheme}
Let $ \mathscr{F} = \bigcup_{\lambda \in \N} \mathcal{F}_\lambda$ be a class of efficiently computable functions $f: \algo X \rightarrow \algo Y$ with domain $\algo X$ and range $\algo Y$. A quantum copy-protection (QCP) scheme for the class $\mathscr{F}$ is a pair of $\QPT$ algorithms $\QCP = (\mathsf{Protect}, \mathsf{Eval})$ defined as follows:
\begin{description}
    \item $\mathsf{QCP.Protect}(1^\lambda,d_f) \to \rho:$ takes as input the security parameter $1^\lambda$ and a classical description $d_f$ of a function $f \in \algo F_\lambda$, and outputs a (possibly mixed) quantum state $\rho$.
    \item $\mathsf{QCP.Eval}(1^\lambda,\rho,x) \to \rho' \otimes \ketbra{y}{y}:$ takes as input the security parameter $1^\lambda$, a quantum state $\rho$ and an input $x \in \algo X$, and outputs a bipartite state $\rho' \otimes \ketbra{y}{y}$ with $y \in \algo Y$.
\end{description}
\end{definition}

Slightly abusing notation, we occasionally ignore the post-evaluation state $\rho'$ and simply identify the output of the procedure $\mathsf{QCP.Eval}(1^\lambda,\rho,x)$ with a classical outcome denoted by $y \in \algo Y$.

We say that a QCP scheme is $\delta$-\emph{correct} if, for any $\lambda \in \mathbb{N}$, any  $f \in \algo{F}_\lambda$, and any input $x \in \algo X$ to $f$:
$$\Pr\bigg[ \mathsf{QCP.Eval}(1^\lambda,\rho, x)=f(x) \, : \, \rho \leftarrow \mathsf{QCP.Protect}(1^{\lambda}, d_f) \bigg] \geq 1-\delta(\lambda).$$
Note that the probability above comes from the procedure $\mathsf{QCP.Eval}$ of the QCP scheme. If $\delta(\lambda) = \negl(\lambda)$, we simply call a copy-protection scheme \emph{correct}. By the Gentle Measurement Lemma~\cite{winter1999} it is easy to see that a $\delta$-correct scheme is reusable in the following sense: after performing $\QCP.\mathsf{Eval}$ to $\rho$ it is possible to rewind the procedure to obtain a state that is within trace distance $\sqrt{\delta}$ of the original state $\rho$.
 
Informally, we say that a QCP scheme $\QCP = (\mathsf{Protect}, \mathsf{Eval})$ is \emph{secure} if no $\QPT$ adversary can produce two ``copies'' of a copy-protected program $\rho \leftarrow \QCP.\mathsf{Protect}(1^\lambda,d_f)$ that can both be used to evaluate $f$. We formalise the security of copy-protection schemes by means of the following security experiment.
 
\begin{definition}[Piracy experiment]\label{def:CP-security} Let $\QCP = (\mathsf{Protect}, \mathsf{Eval})$ be a copy-protection scheme for a class of functions $ \mathscr{F} = \bigcup_{\lambda \in \N} \mathcal{F}_\lambda$ with domain $\algo X$ and range $\algo Y$. Let $\algo D_{\mathscr{F}} = \{\algo D_{\algo F_\lambda}\}_{\lambda  \in \N}$ be an ensemble of distributions over $\algo F_\lambda$ and let $\algo D_{\algo X} = \{\algo D_{\algo X}(f)\}_{f \in \algo F_\lambda}$ be an ensemble of distributions over function inputs $\algo X$. The security game (which we call the \emph{piracy experiment}) takes place between a challenger and an adversary consisting of a triplet of $\QPT$ algorithms $(\algo A, \algo B,\algo C)$:
 \begin{enumerate}[label=\arabic*.]
     \item The challenger samples $f \leftarrow \algo D_{\algo F_\lambda}$ and sends the program $\rho \leftarrow \QCP.\mathsf{Protect}(1^\lambda,d_f)$ to $\algo A$.
     \item $\algo A$ applies an efficient $\CPTP$ map to map $\rho$ into a bipartite state $\rho_{BC}$ on systems $BC$, and sends system $B$ to $\algo B$ and system $C$ to $\algo C$ (who are not allowed to communicate from this step onward).
     
     \item The challenger samples a pair $(x_B,x_C) \leftarrow \algo D_{\algo X}(f) \times \algo D_{\algo X}(f)$, and sends $x_B$ to $\algo B$ and $x_C$ to $\algo C$.
     \item $\algo B$ and $\algo C$ output values $y_B \in \algo Y$ and $y_C \in \algo Y$, respectively, and send them to the challenger. The challenger outputs $1$, if $y_B = f(x_B)$ and $y_C = f(x_C)$ (i.e., the adversary has succeeded) and $0$, otherwise (i.e., the adversary has failed).
 \end{enumerate}
We let the random variable $\mathsf{PiracyExp}_{\algo D_{\mathscr{F}}, \algo D_{\algo X}}^{\QCP}\big(1^\lambda, (\algo A,\algo B,\algo C)\big)$ denote the output bit of the challenger.
\end{definition}

\begin{definition}[Secure quantum copy-protection]\label{def:secure-CP-scheme} Let $\QCP = (\mathsf{Protect}, \mathsf{Eval})$ be a QCP scheme for a class of functions $ \mathscr{F} = \bigcup_{\lambda \in \N} \mathcal{F}_\lambda$. Let $\algo D_{\mathscr{F}} = \{\algo D_{\algo F_\lambda}\}_{\lambda  \in \N}$ be an ensemble of distributions over $\algo F_\lambda$ and let $\algo D_{\algo X} = \{\algo D_{\algo X}(f)\}_{f \in \algo F_\lambda}$ be an ensemble of distributions over $\algo X$. Then, $\QCP = (\mathsf{Protect}, \mathsf{Eval})$ is called $(\algo D_{\mathscr{F}},\algo D_{\algo X},\gamma)$-secure if, for any triplet of $\QPT$ algorithms $(\algo A,\algo B,\algo C)$, it holds that
$$
\Pr \left[ \mathsf{PiracyExp}_{\algo D_{\mathscr{F}}, \algo D_{\algo X}}^{\QCP}\big(1^\lambda, (\algo A,\algo B,\algo C)\big) =1\right] \leq  p^{\setft{triv}}_{\algo D_{\mathscr{F}},\algo D_{\algo X}} + \gamma(\lambda),
$$
where $p^{\setft{triv}}_{\algo D_{\mathscr{F}},\algo D_{\algo X}} $ is the trivial winning probability that is always possible due to correctness: $\algo A$ forwards the original copy-protected program to one of the parties, say $\algo B$ (who then evaluates it to obtain the correct output), and the other party, say $\algo C$, has to guess at random~\cite{coladangelo2020quantum}.

\end{definition}
 
\subsection{Quantum encryption with wrong-key detection}

Recent work by Coladangelo, Majenz and Poremba~\cite{coladangelo2020quantum} proposed the first construction of QCP for compute-and-compare programs in the quantum random oracle model (QROM) as well as a scheme for multi-bit point functions in the QROM based on unclonable encryption with wrong-key detection (WKD) -- a property which enables the decryption procedure to recognise incorrect keys. Because the QCP constructions in~\cite{coladangelo2020quantum} rely on BB84 states, we can immediately use our parallel RSP protocol to dequantise the communication between the sender and receiver.
We remark that the former construction in~\cite{coladangelo2020quantum}, however, only guarantees a constant adversarial advantage for creating pirated copies. 

In this work, we focus on the latter QCP construction for multi-bit point functions, which is closely related to \emph{unclonable encryption} and enjoys stronger security guarantees. Our QCP scheme for multi-bit point functions (formally defined in \ref{sec:QCP-MBPF}) combines our unclonable hybrid encryption scheme in \ref{cons:unclonable-encryption-classical-client} with a generic WKD transformation in the QROM as in~\cite{coladangelo2020quantum}. 
The notion of WKD has previously been used in the context of classical obfuscation of point functions by Canetti et al.~\cite{cryptoeprint:2010:049} and was shown to be achievable unconditionally using pair-wise independent permutations~\cite{NaorReingold97}. This seems to suggest that one may be able to copy-protect multi-bit point functions unconditionally (i.e., without the QROM assumption) by combining a generic WKD transformation based on pair-wise independent permutations with an unclonable encryption scheme similar to the construction of Coladangelo et al.~\cite{coladangelo2020quantum}. Unfortunately, it seems difficult to justify that the resulting encryption scheme remains unclonable when relying on pair-wise independent permutations, rather than the random oracle heuristic.
We choose to rely on the generic WKD transformation in the QROM from~\cite{coladangelo2020quantum} as it is compatible with search-to-decision reductions via the one-way-to-hiding paradigm~\cite{cryptoeprint:2018/904}. The latter is crucial when proving that the WKD transformation preserves the notion of unclonability (see \cite{coladangelo2020quantum}). Despite being uninstantiable in principle~\cite{10.1145/1008731.1008734}, modeling
hash functions as quantum-accessible random functions (in the QROM) is considered a standard assumption in cryptography~\cite{https://doi.org/10.48550/arxiv.1008.0931}.

We use the following definition of wrong-key detection for QECM schemes~\cite{coladangelo2020quantum} who consider quantum secret-key encryption schemes, more generally.

\begin{definition}[Wrong-key detection for QECM schemes]\label{def:wkd} A QECM scheme $\Sigma  = (\KeyGen,\Enc,\Dec)$ satisfies the wrong-key detection ($\WKD$) property if, for every $k'\neq k \leftarrow \KeyGen(1^\lambda)$ and plaintext $m$,
$$
\tr{(\1 - \proj{\bot}) \Dec_{k'} \circ \Enc_k(m)} \leq \negl(\lambda).
$$

\end{definition}

We use the following generic transformation due to Coladangelo, Majenz and Poremba~\cite{coladangelo2020quantum} which allows us to turn any QECM scheme into another QECM scheme with wrong-key detection.

\begin{construction}[Generic transformation for $\WKD$ in the QROM]\label{cons:generic-trf} Let $\lambda \in \N$ be the security parameter and let $\Sigma = (\KeyGen,\Enc,\Dec)$ be a QECM scheme with key space $\algo K = \bit^\lambda$. Let $h: \bit^\lambda \rightarrow \bit^\kappa$ be a function. Consider the following QECM scheme $\Sigma_h = (\KeyGen',\Enc',\Dec')$ defined by:
\begin{description}
\item $\KeyGen'(1^\lambda) \rightarrow k:$ this is the same as $ \KeyGen(1^\lambda)$.
\item $\Enc'(k,m) \rightarrow \ct$: takes as input a secret key $k \in \bit^\lambda$ and plaintext $m \in \bit^\lambda$ and outputs the quantum ciphertext given by $\ct = \Enc_k(m) \ot \proj{h(k)}$.
\item $\Dec'(k, \ct) \rightarrow m' \text{ \textbf{or} } \bot:$ takes as input a key $k \in \bit^\lambda$ and ciphertext $\ct$, and decrypts as follows:
\begin{enumerate}[label=\arabic*.]
    \item Parse the ciphertext $\ct$ as $\rho \otimes  \proj{z}$.
    \item If $h(k) \neq z$, output $\bot$. Else, if $h(k)=z$, output $m'=\Dec_k(\rho)$.
\end{enumerate}
\end{description}
\end{construction}

We use the following lemma which is shown in \cite[Lemma 8]{coladangelo2020quantum}.

\begin{lemma}[\cite{coladangelo2020quantum}]\label{lem:generic_transf}
Let $\lambda \in \N$ and let $\Sigma = (\KeyGen,\Enc,\Dec)$ be any $t(\lambda)$-unclonable secure $\QECM$ scheme. Let $h: \{0,1\}^\lambda \rightarrow \{0,1\}^\kappa$ be a hash function, for $ \kappa = \omega(\log \lambda)$. Then, \ref{cons:generic-trf} yields a correct and $t(\lambda)$-unclonable secure $\QECM$ scheme $\Sigma_h$ with WKD in the QROM.
\end{lemma}

Finally, we conclude with the following result.

\begin{corollary}\label{cor:WKD_HE}
Let $\lambda \in \N$ and let $h: \bit^\lambda \rightarrow \bit^\kappa$ be a hash function with $\kappa = \omega(\log \lambda)$. Let $\Sigma$ be a $t(\lambda)$-unclonable secure $\QECM$ schme and let $\Sigma_h$ be the associated WKD transformation in \ref{cons:generic-trf}. Then, the hybrid encryption scheme $\mathsf{HE}= (\KeyGen,\Enc,\Dec)$ in \ref{cons:hybrid-encryption} instantiated with the scheme $\Sigma_h$ is correct, $t(\lambda)$-unclonable secure and has the WKD property in the QROM.
\end{corollary}

\begin{proof}
The correctness of the scheme $\mathsf{HE}= (\KeyGen,\Enc,\Dec)$ is implied by the correctness of $\Sigma_h$. To argue that $\mathsf{HE}$ is $t(\lambda)$-unclonable secure, we can use the fact that $\Sigma_h$ is $t(\lambda)$-unclonable secure in the QROM (see \ref{lem:generic_transf}) and then apply \ref{prop:unclonable-HE}. The fact that $\mathsf{HE}$ inherits the WKD property from $\Sigma_h$ is clear, as the decryption procedure of the scheme $\mathsf{HE}$ first checks whether a given ciphertext decrypts correctly via $\Sigma_h$ under a particular secret key. This proves the claim.
\end{proof}

\subsection{Quantum copy-protection with a classical client}

Let us now define the notion of quantum copy-protection with a classical client. Similar to \ref{def:copy-protection-scheme}, it enables a \emph{classical} client to delegate a copy-protected program to a \emph{quantum} receiver in a way that allows the receiver to execute the program, but prevents it from creating additional ``pirated'' copies.

\begin{definition}[Quantum copy-protection scheme with a classical client]\label{def:copy-protection-scheme_cc}
Let $ \mathscr{F} = \bigcup_{\lambda \in \N} \mathcal{F}_\lambda$ be a class of efficiently computable functions $f: \algo X \rightarrow \algo Y$. A QCP scheme with a classical client (or, QCP$_{\text{CC}}$ scheme) for the class $\mathscr{F}$ is a pair of efficient interactive quantum algorithms $\QCP_{\mathsf{CC}} = (\mathsf{Protect}, \mathsf{Eval})$:
\begin{description}
    \item $\mathsf{QCP_\CC.Protect}(\algo C(1^\lambda,d_f),\algo R(1^\lambda)) \rightarrow \rho \, \textbf{or} \, \bot$: this is an interactive protocol between a classical client $\algo C$ (who takes as input the parameter $1^\lambda$ and a classical description $d_f$ of a function $f \in \algo F_\lambda$) and a quantum receiver $\algo R$ (who takes as input $1^\lambda$).
    Provided that the protocol is successful with the flag $\top$, $\algo R$ is in possession of a quantum state $\rho$ (which depends on $f$). Otherwise, $\algo C$ outputs $\bot$.
    \item $\mathsf{QCP_\CC.Eval}(1^\lambda,\rho,x):$ takes as input the security parameter $1^\lambda$, a quantum state $\rho$ and an input $x \in \algo X$, and outputs a bipartite state $\rho' \otimes \ketbra{y}{y}$ for an outcome $y \in \algo Y$.
\end{description}
\end{definition}

Similar to \ref{def:copy-protection-scheme}, we occasionally ignore the post-evaluation state $\rho'$ and simply identify the output of the procedure $\mathsf{QCP_\CC.Eval}(1^\lambda,\rho,x)$ with a classical outcome $y \in \algo Y$.
We say that a QCP$_{\text{CC}}$ scheme is $\delta$-correct, if, for any $\lambda \in \mathbb{N}$, any  $f \in \algo{F}_\lambda$, and any input $x \in \algo X$ to $f$:
 $$\Pr\bigg[ \mathsf{QCP_\CC.Eval}(1^\lambda,\rho,  x) = f(x) \, : \, \rho \leftarrow \mathsf{QCP_\CC.Protect}(\algo C(1^\lambda,d_f,\algo),R(1^\lambda)) \bigg] \geq 1- \delta(\lambda).$$

Let us now define a piracy experiment with a classical client, similar to \ref{def:CP-security}.
 
 \begin{definition}[Piracy experiment with a classical client] Let $\QCP_\CC = (\mathsf{Protect}, \mathsf{Eval})$ be a QCP$_{\text{CC}}$ scheme for a class of functions $ \mathscr{F} = \bigcup_{\lambda \in \N} \mathcal{F}_\lambda$ with domain $\algo X$ and range $\algo Y$. Let $\algo D_{\mathscr{F}} = \{\algo D_{\algo F_\lambda}\}_{\lambda  \in \N}$ be an ensemble of distributions over $\algo F_\lambda$ and let $\algo D_{\algo X} = \{\algo D_{\algo X}(f)\}_{f \in \algo F_\lambda}$ be an ensemble of distributions over $\algo X$. The security game (called $\emph{piracy experiment with a classical client}$) takes place between a challenger and an adversary consisting of an efficient interactive quantum prover $\algo P$ and triplet of $\QPT$ algorithms $(\algo A, \algo B,\algo C)$:
 \begin{enumerate}
     \item The challenger samples a function $f \leftarrow \algo D_{\algo F_\lambda}$. To classically delegate a program for the functionality $f$ according to $\QCP_\CC$, the challenger takes the role of the classical client $\algo C$ and the prover $\algo P$ takes the role of the (possibly malicious) receiver $\algo R$ in the interactive protocol specified by $\QCP_\CC.\mathsf{Protect}$. Provided that the protocol succeeds with $\mathtt{flag}=\top$, $\algo P$ is in possession of a quantum state $\rho_A$ (which depends on $f$, as well as any additional information the adversary has collected throughout the protocol). Otherwise, if $\mathtt{flag}=\bot$, the adversary loses.
     \item $\algo A$ applies an efficient $\CPTP$ map to map $\rho_A$ into a bipartite state $\rho_{BC}$ on systems $BC$, and sends system $B$ to $\algo B$ and system $C$ to $\algo C$ (who are not allowed to communicate from this step onward).
     
     \item The challenger samples a pair $(x_B,x_C) \leftarrow \algo D_{\algo X}(f) \times \algo D_{\algo X}(f)$, and sends $x_B$ to $\algo B$ and $x_C$ to $\algo C$.
     \item $\algo B$ and $\algo C$ output values $y_B \in \algo Y$ and $y_C \in \algo Y$, respectively, and send them to the challenger. The challenger outputs $1$ if it holds that $y_B = f(x_B)$, $y_C = f(x_C)$ and $\mathtt{flag} = \top$ (i.e., the adversary has succeeded), and $0$ otherwise (i.e., the adversary has failed).
 \end{enumerate}
We let the random variable $\mathsf{PiracyExp}_{\algo D_{\mathscr{F}}, \algo D_{\algo X}}^{\QCP_\CC}\big(1^\lambda, (\algo P,\algo A,\algo B,\algo C)\big)$ denote the output bit of the challenger.
\end{definition}

 We formalise the notion of security with a classical client similar to \ref{def:secure-CP-scheme}.

\begin{definition}[Secure quantum copy-protection with a classical client] Let $\QCP_\CC = (\mathsf{Protect}, \mathsf{Eval})$ be a QCP$_{\text{CC}}$ scheme for a class of functions $ \mathscr{F} = \bigcup_{\lambda \in \N} \mathcal{F}_\lambda$ with domain $\algo X$ and range $\algo Y$. Further, let $\algo D_{\mathscr{F}} = \{\algo D_{\algo F_\lambda}\}_{\lambda  \in \N}$ be an ensemble of distributions over $\algo F_\lambda$ and let $\algo D_{\algo X} = \{\algo D_{\algo X}(f)\}_{f \in \algo F_\lambda}$ be an ensemble of distributions over $\algo X$. Then, $\QCP_\CC$ is called $(\algo D_{\mathscr{F}},\algo D_{\algo X},\gamma)$-secure if, for any efficient interactive quantum prover $\algo P$ and any triplet of $\QPT$ algorithms $(\algo A,\algo B,\algo C)$, it holds that
$$
\Pr \left[ \mathsf{PiracyExp}_{\algo D_{\mathscr{F}}, \algo D_{\algo X}}^{\QCP_\CC}\big(1^\lambda, (\algo P,\algo A,\algo B,\algo C)\big) =1\right] \leq p^{\setft{triv}}_{\algo D_{\mathscr{F}},\algo D_{\algo X}} + \gamma(\lambda).
$$

\end{definition}

\subsection{Quantum copy-protection of multi-bit point functions with a classical client}\label{sec:QCP-MBPF} In this section, we combine our parallel RSP protocol for BB84 states with a QECM scheme with wrong-key detection to construct a QCP$_\text{CC}$ scheme for multi-bit point functions, which is described in \ref{protocol-CP-classical-client}.
We begin with a formal definition. For integers $\mu \in \N$ and $\ell \in \N$, the class of multi-bit point functions $\mathscr{P}_{\mu,\ell}$ consists of functions $P_{y,m}\in \mathscr{P}_{\mu,\ell}$, where for a marked input $\vec y \in \bit^\mu$ and output $\vec m \in \bit^\ell$,
$$
P_{y,m}(\vec x) = \begin{cases}
\vec m, & \text{ if } \vec x=\vec y,\\
\vec 0^\ell,  & \text{ if } \vec x\neq \vec y.
\end{cases}
$$
Let $\lambda \in \N$ be the security parameter and let $\mu = 2 \lambda$ and $\ell = \lambda$. Our construction in \ref{protocol-CP-classical-client} enables a classical client to remotely prepare a copy-protected program for a multi-bit point function of the form $P_{y,m}$, where $y \rand \bit^\mu$ and $m \rand \bit^\ell$ are uniformly random. We first prove the \emph{correctness} of \ref{protocol-CP-classical-client}.

\begin{table}[p!]
\begin{longfbox}[breakable=false, padding=1em, padding-right=1.8em, padding-top=1.2em, margin-top=1em, margin-bottom=1em, background-color=gray!20]
\begin{protocol} {\bf Quantum Copy-Protection of Multi-bit Point Functions with a Classical Client} \label{protocol-CP-classical-client} \end{protocol}

Let $\lambda \in \N$ be the security parameter and let $\mu = 2 \lambda$ and $\ell = \lambda$. Let $\mathscr{P}_{\mu,\ell}$ be the class of multi-bit point functions with input length $\mu$ and output length $\ell$.
Let $h: \bit^\lambda \rightarrow \bit^\kappa$ be a hash function.
We define the QCP$_{\text{CC}}$ scheme consisting of a pair of efficient (interactive) algorithms $\QCP_{\CC,h}=(\mathsf{Protect},\mathsf{Eval})$ for the class $\mathscr{P}_{\mu,\ell}$ as follows:
\begin{itemize}
\item $\mathsf{QCP}_{\CC,h}.\mathsf{Protect}(\algo C(1^{\lambda},d_{P_{y,m}}),\algo R(1^\lambda))$: this is an interactive protocol between a classical client $\algo C$ and a quantum receiver $\algo R$.
The client $\algo C$ takes as input
the security parameter $1^\lambda$ and a succinct classical description $d_{P_{y,m}}$ of a multi-bit point function $P_{y,m}\in \mathscr{P}_{\mu,\ell}$, whereas the receiver $\algo R$ takes as input the security parameter $1^\lambda$. The protocol takes place as follows:
\begin{enumerate}
    \item The client $\algo C$ parses the marked input as $\vec y$ as $(\vec s || \vec \theta)$ with $\vec s,\vec \theta \in \bit^\lambda$.
    \item The client $\algo C$ and the receiver $\algo R$ run the interactive RSP protocol for BB84 states (\ref{fig:protocol_multi_round}) with security parameter $1^\lambda$, $n=\lambda$ and choice of basis $\vec \theta \in \bit^{\lambda}$.
    \item Provided the protocol is successful, $\algo C$ obtains $\vec v \rand \bit^{\lambda}$ and the receiver $\algo R$ obtains a $\lambda$-qubit state $\sigma$ (which depends on $\vec v$ and $\vec{\theta}$). Otherwise, $\algo C$ outputs $\bot$.
    \item $\algo C$ sends the classical strings $\vec c = \vec v \oplus \vec s \oplus \vec m$ and $\vec z = h(\vec y)$ to the receiver $\algo R$.
\end{enumerate}

\item $\mathsf{QCP}_{\CC,h}.\mathsf{Eval}(1^{\lambda}, \sigma ; x)$: takes as input $1^\lambda$, an alleged program $\sigma$, and a string $x \in \{0,1\}^{2 \lambda}$ (the input on which the program is to be evaluated), and runs the program as follows:
\begin{enumerate}
\item Parse the program as $\rho \otimes  \proj{\vec c} \otimes \proj{\vec z} \leftarrow \rho$.
\item Parse the input $\vec x$ as $(\vec s_x || \vec \theta_x)$ with $\vec s_{x},\vec \theta_{x} \in \bit^\lambda$.
    \item Output $\vec 0^\ell$, if $h(\vec x) \neq \vec z$. Else, if $h(\vec x) =\vec z$, apply Hadamards $H^{\theta_{x,1}} \otimes \dots \otimes H^{\theta_{x,2\lambda}}$ to $\rho$ and measure in the computational basis with outcome $\vec w$. Then, output $\vec m' =\vec w \oplus \vec s_x \oplus \vec c$.
\end{enumerate}
\end{itemize}
\end{longfbox}
\end{table}

\begin{proposition}[Correctness]
Let $\lambda \in \N$ be the security parameter and let $\mu =2 \lambda$ and $\ell = \lambda$. Let $h: \bit^\lambda \rightarrow \bit^\kappa$ be a hash function with $\kappa = \omega(\log \lambda)$. Then, the scheme
$\mathsf{QCP}_{\CC,h} = (\mathsf{Protect}, \mathsf{Eval})$ in \ref{protocol-CP-classical-client} is a correct QCP$_\text{CC}$ scheme for the class $\mathscr{P}_{\mu,\ell}$ in the QROM.
\end{proposition}

\begin{proof}
To prove correctness, we show that for any $\lambda \in \mathbb{N}$, $P_{y,m} \in \mathscr{P}_{\mu,\ell}$ and any input $\vec x \in \bit^{2 \lambda}$ to $P_{y,m}$,
 $$\Pr\bigg[ \mathsf{QCP}_{\CC,h}.\mathsf{Eval}(1^\lambda,\sigma, \vec x) \neq P_{y,m}(\vec x) \, : \, \sigma \leftarrow \mathsf{QCP}_{\CC,h}.\mathsf{Protect}(\algo C(1^\lambda,d_{P_{y,m}}),\algo R(1^\lambda)) \bigg] \leq \negl(\lambda).$$
Note that in this scenario we can assume that the receiver $\algo R$ is honest throughout \ref{protocol-CP-classical-client}.
According to \ref{prop:honest-prover-correctness},
$\algo R$ is accepted in the interactive protocol $\mathsf{QCP}_{\CC,h}.\mathsf{Protect}(\algo C(1^\lambda,d_{P_{y,m}}),\algo R(1^\lambda))$ with probability negligibly close to $1$ in the parameter $\lambda$ 
(for parameter choices $n$ at most polynomial in $\lambda$ and $\delta$ at least inverse polynomial in $\lambda$). Provided that the protocol successful, the final state of $\algo R$ is given by
\begin{align}\label{eq:sigma-HE-ctxt}
\sigma =  \bigotimes_{i=1}^{\lambda} H^{\theta_i} \proj{v_i} H^{\theta_i}  \ot \proj{\vec c} \otimes \proj{\vec z}.
\end{align}
where $\vec v \rand \bit^\lambda$ is the random string recorded by the client $\algo C$ during the parallel RSP protocol, and where $\vec c = \vec v \oplus \vec s \oplus \vec m$ and $z= h(\vec y)$ with respect to $(\vec s||\vec \theta) \leftarrow \vec y$.

Let $\Sigma_h = (\KeyGen,\Enc,\Dec)$ be the QECM scheme with WKD in \ref{cons:generic-trf} instantiated with the conjugate coding encryption scheme in \ref{cons:conjugate-coding} and hash $h$. Further, let $\mathsf{HE}= (\KeyGen,\Enc,\Dec)$ be the hybrid encryption scheme in \ref{cons:hybrid-encryption} instantiated with $\Sigma_h$. Notice that the state $\sigma$ in \ref{eq:sigma-HE-ctxt} corresponds precisely to a ciphertext generated by $\mathsf{HE}.\Enc_{\vec y}(\vec m)$ with randomness $\vec u = \vec v \oplus \vec s$, where
\begin{align}
\mathsf{HE}.\Enc_{\vec y}(\vec m) = \Sigma_h.\Enc_{\vec y}(\vec u) \otimes \proj{\mathsf{OTP}.\Enc_{\vec u}(\vec m)}.
\end{align}
Note that, if $\vec x = \vec y$, the correctness of $\QCP_{\CC,h}$ follows immediately from the decryption correctness of $\mathsf{HE}$.
Else, if $\vec x \neq \vec y$, the correctness of $\QCP_{\CC,h}$ follows from the WKD property of $\mathsf{HE}$ (in particular, of $\Sigma_h$). Both properties therefore follow from \ref{cor:WKD_HE}.
This proves the claim.
\end{proof}

Let us now prove the security of our QCP$_\text{CC}$ scheme \ref{protocol-CP-classical-client} for the class of point functions $\mathscr{P}_{\mu,\ell}$ of the form $P_{y,m}$, where the marked strings $y \rand \bit^\mu$ and $m \rand \bit^\ell$ are chosen uniformly random.

\begin{proposition}[Security]
Let $\lambda \in \N$ be the security parameter and let $\mu = 2 \lambda$ and $\ell = \lambda$. Let $\mathscr{P}_{\mu,\ell}$ be the class of multi-bit point functions with input $\algo X = \bit^{\mu}$ output $\algo Y = \bit^{\ell}$. Let
$\algo D_{\algo X}$ be an arbitrary ensemble of challenge distributions over $\algo X$ and let $\algo D_{\mathscr{P}_{\mu,\ell}}$ be the ensemble of multi-bit point functions that samples both the marked input and the marked output uniformly at random. Let $h: \bit^\lambda \rightarrow \bit^\kappa$ be a hash function with $\kappa = \omega(\log \lambda)$.
Then, there exists $\gamma(\lambda) = 1/\poly(\lambda)$ such that the scheme
$\QCP_{\mathsf{CC},h} = (\mathsf{Protect}, \mathsf{Eval})$ in \ref{protocol-CP-classical-client} is a
 $(\algo D_{\mathscr{P}_{\mu,\ell}},\algo D_{\algo X},\gamma)$-secure QCP$_\text{CC}$ scheme for the class $\mathscr{P}_{\mu,\ell}$ in the QROM. In other words, for any (interactive) $\QPT$ algorithms $(\algo P,\algo A,\algo B,\algo C)$, it holds that
$$
\Pr \left[ \mathsf{PiracyExp}_{\algo D_{\mathscr{P}_{\mu,\ell}}, \algo D_{\algo X}}^{\QCP_{\CC,h}}\big(1^\lambda, (\algo P,\algo A,\algo B,\algo C)\big) =1\right] \leq p^{\setft{triv}}_{\algo D_{\mathscr{P}_{\mu,\ell}},\algo D_{\algo X}} + \gamma(\lambda).
$$

\end{proposition}

\begin{proof}
Let $\Adv = (\algo A,\algo B,\algo C)$ denote the adversary in the challenge phase of the piracy experiment $\mathsf{PirExp}_{\algo D_{\mathscr{P}_{\mu,\ell}}, \algo D_{\algo X}}^{\QCP_{\CC,h}}$. We consider two cases, namely when $p^{\setft{triv}}_{\algo D_{\mathscr{P}_{\mu,\ell}},\algo D_{\algo X}}=1$ and when $p^{\setft{triv}}_{\algo D_{\mathscr{P}_{\mu,\ell}},\algo D_{\algo X}} <1$ (depending on the challenge distribution $\algo D_{\algo X}$). In the former case, the scheme is trivially secure by definition and we are done. Hence, we will assume that $p^{\setft{triv}}_{\algo D_{\mathscr{P}_{\mu,\ell}},\algo D_{\algo X}}<1$ for the remainder of the proof. Note that in this case, the distribution $\algo D_{\algo X}$ has non-zero weight on the marked input $y$ of a random point function $P_{y,m}$.

Let $(x_B,x_C) \leftarrow \algo D_{\algo X}$ denote the inputs received by the algorithms $\algo B$ and $\algo C$ during the challenge phase, and let $y_B$ and $y_C$ be the outputs returned by $\algo B$ and $\algo C$, respectively.
We can express the probability that $\Adv$ succeeds at the challenge phase (i.e., that $y_B = P_{y,m}(x_B)$ and $y_C=P_{y,m}(x_C)$) as follows:
\begin{align*}
\Pr[ \Adv \text{ wins}]
&\! =\! \Pr[ \Adv \text{ wins} \,|\, x_B \!\neq\! y \!\neq\! x_C] \!\cdot\!\Pr[x_B \!\neq\! y \!\neq\! x_C]\!+\!\Pr[ \Adv \text{ wins} \,|\, x_B \!=\! y \!\neq\! x_C] \!\cdot\!\Pr[x_B \!=\! y \!\neq\! x_C]\nonumber\\
&\, \!+\!\Pr[ \Adv \text{ wins} \,|\, x_B \!\neq\! y \!=\! x_C] \!\cdot\!\Pr[x_B \!\neq\! y \!=\! x_C]\!+\!\Pr[ \Adv \text{ wins} \,|\, x_B \!=\! y \!=\! x_C] \!\cdot\!\Pr[x_B \!=\! y \!=\! x_C].
\end{align*}
Without loss of generality, we assume that $\Pr[x_B = y \!\neq\! x_C] \leq \Pr[x_B \!\neq\! y = x_C]$. Hence,
\begin{align}
\Pr[ \Adv \text{ wins}] &\leq  \Pr[ \Adv \text{ wins} \,|\, x_B \!\neq\! y \!\neq\! x_C] \cdot \Pr[x_B \!\neq\! y \!\neq\! x_C]\nonumber\\
& +\big(\Pr[ \Adv \text{ wins} \,|\, x_B = y \!\neq\! x_C]+\Pr[ \Adv \text{ wins} \,|\, x_B \!\neq\! y = x_C]\big) \cdot \Pr[x_B \!\neq\! y = x_C]\nonumber\\
& +\Pr[ \Adv \text{ wins} \,|\, x_B = y = x_C] \cdot \Pr[x_B = y = x_C]. \label{eq:A_wins}
\end{align}
Let us now state the following simple inequality. By first applying the union bound and then using that $\algo B$ and $\algo C$ are non-signalling, we find that:
\begin{align}
&\Pr[ \Adv \text{ wins} \,|\, x_B = y =  x_C]\nonumber\\
&= \Pr[\algo B \text{ succeeds} \land \algo C \text{ succeeds} \,|\, x_B = y = x_C]\nonumber\\
&\geq \Pr[\algo B \text{ succeeds} \,|\, x_B = y = x_C] + \Pr[ \algo C \text{ succeeds} \,|\, x_B = y = x_C] -1\nonumber\\
&= \Pr[\algo B \text{ succeeds} \,|\, x_B = y \neq x_C] + \Pr[ \algo C \text{ succeeds} \,|\, x_B \neq y = x_C] -1\nonumber\\
&\geq \Pr[\Adv \text{ wins} \,|\, x_B = y \neq x_C] + \Pr[ \Adv \text{ wins} \,|\, x_B \neq y = x_C] -1.
\end{align}
Plugging this into~\ref{eq:A_wins}, we obtain the following upper bound:
\begin{align}
\Pr[\Adv \text{ wins}] & \leq  \Pr[ \Adv \text{ wins} \,|\, x_B \neq y \neq x_B] \cdot \Pr[x_B \neq y \neq x_C]\nonumber\\
& +\big(1 + \Pr[ \Adv \text{ wins} \,|\, x_B = y =  x_C]\big) \cdot \Pr[x_B \neq y = x_C]\nonumber\\
& +\Pr[ \Adv \text{ wins} \,|\, x_B = y = x_C] \cdot \Pr[x_B = y = x_C]\nonumber\\
& \leq  \Pr[x_B \neq y \neq x_C] + \Pr[x_B \neq y = x_C]+2\Pr[ \Adv \text{ wins} \,|\, x_B = y = x_C]\nonumber\\
&= p^{\setft{triv}}_{\algo D_{\mathscr{P}_{\mu,\ell}},\algo D_{\algo X}} + 2\Pr[ \Adv \text{ wins} \,|\, x_B = y = x_C].\label{eq:A_wins3}
\end{align}
In the last line, we used the assumption that $\Pr[x_B = y \neq x_C] \leq \Pr[x_B \neq y = x_C]$ together with the following simple identity for the trivial guessing probability:
$$
p^{\setft{triv}}_{\algo D_{\mathscr{P}_{\mu,\ell}},\algo D_{\algo X}} = \Pr[x_B \neq y \neq x_C] + \max\big\{\Pr[x_B \neq y = x_C],\Pr[x_B = y \neq x_C]\big\}.
$$
We complete the proof by showing that $\Pr[ \Adv \text{ wins} \,|\, x_B = y = x_C] \leq \eps(\lambda)+\negl(\lambda)$, for some function $\eps(\lambda) = 1/\poly(\lambda)$. This implies that
\begin{align}
\Pr[\Adv \text{ wins}] \leq p^{\setft{triv}}_{\algo D_{\mathscr{P}_{\mu,\ell}},\algo D_{\algo X}} + \eps(\lambda) + \negl(\lambda).\label{eq:CP_security_bound}
\end{align}

Let $\sigma_{|\top} \leftarrow \QCP_{\CC,h}.\mathsf{Protect}(\algo C(1^\lambda,d_{P_{y,m}}),\algo R(1^\lambda))$ denote the final state of the interactive protocol conditioned on the event that challenger outputs the flag $\bot$. As in the proof of \ref{thm:unclonable-security-CC}, we assume that the acceptance probability of $\algo P$ in the interactive protocol satisfies $\Pr[\top]\geq 1/\poly(\lambda)$. Using a similar careful analysis as in the proof of \ref{thm:unclonable-security-CC} we can argue that, without loss of generality (and up to an isometry), the final state of $\algo R$ is $\eps(\lambda)=1/\poly(\lambda)$-computationally indistinguishable from
\begin{align}\label{eq:final-state-CP}
\sigma =  \bigotimes_{i=1}^{\lambda} H^{\theta_i} \proj{v_i} H^{\theta_i}  \ot \proj{\vec c} \otimes \proj{\vec z},
\end{align}
where $\vec v \rand \bit^\lambda$ is the random string recorded by the client $\algo C$ during the parallel RSP protocol, and where $\vec c = \vec v \oplus \vec s \oplus \vec m$ and $z= h(\vec y)$ with respect to $(\vec s||\vec \theta) \leftarrow \vec y$, for a random marked input $\vec y \in \rand \bit^{2\lambda}$.

Let $\Sigma_h = (\KeyGen,\Enc,\Dec)$ be the QECM scheme with WKD in \ref{cons:generic-trf} instantiated with the conjugate coding encryption scheme in \ref{cons:conjugate-coding} and hash $h$. Further, let $\mathsf{HE}= (\KeyGen,\Enc,\Dec)$ be the hybrid encryption scheme in \ref{cons:hybrid-encryption} instantiated with $\Sigma_h$.
To complete the proof, we now relate the success probability of the adversary $(\algo P,\algo A,\algo B,\algo C)$ against $\QCP_\CC$ in the experiment $\mathsf{PirExp}_{\algo D_{\mathscr{P}_{\mu,\ell}}, \algo D_{\algo X}}^{\QCP_\CC}$ with challenge pair $x_B =y$ and $x_C=y$ to a particular adversary $(\algo A',\algo B',\algo C')$ in the experiment $\mathsf{CloneExp}_{\mathsf{HE}}$.

We now proceed as in the proof of \ref{thm:unclonable-security-CC}. Using the indistinguishability property in \ref{eq:final-state-CP} together with \ref{cor:WKD_HE}, we find that there exists $\eps(\lambda) = 1/\poly(\lambda)$ such that
\begin{align*}
&\Pr[ \Adv \text{ wins} \,|\, x_B = y = x_C]\\
&=\underset{\vec m}{\mathbb{E}}\underset{\vec y}{\mathbb{E}}\,\, \Tr\left[(\proj{\vec m} \otimes \proj{\vec m})(\mathcal{B}_{\vec y} \otimes \mathcal{C}_{\vec y}) \circ \mathcal{A} \circ \sigma_{|\top} \right]\\
&\leq \underset{\vec m}{\mathbb{E}} \underset{\vec s}{\mathbb{E}}\underset{\vec \theta}{\mathbb{E}} \underset{\vec v}{\mathbb{E}} \,\, \Tr\Bigg[(\proj{\vec m} \otimes \proj{\vec m})(\mathcal{B'}_{(\vec s ,\vec \theta)} \otimes \mathcal{C'}_{(\vec s,\vec \theta)}) \,\circ \\
& \quad\quad\quad\quad\quad\mathcal{A'} \left( \bigotimes_{i=1}^{\lambda} H^{\theta_i} \proj{v_i} H^{\theta_i}  \otimes \proj{\vec v \oplus \vec s \oplus \vec m} \ot \proj{h(\vec s || \vec \theta)} \right) \Bigg] + \eps(\lambda)\\
&= \Pr \left[\mathsf{CloneExp}_{\mathsf{HE}}\big(1^\lambda, (\algo A',\algo B',\algo C')\big) =1 \right] + \eps(\lambda)\\
&\leq 2^{-\lambda + t(\lambda)} + \eps(\lambda) + \negl(\lambda),
\end{align*}
where the last line follows from \ref{cor:WKD_HE} with $t(\lambda) = \lambda \log(1 + 1/\sqrt{2})$. Putting everything together, we find that there exists $\gamma = 1/\poly(\lambda)$ such that the advantage in the piracy experiment is at most
\begin{align*}
\Pr \left[ \mathsf{PiracyExp}_{\algo D_{\mathscr{P}_{\mu,\ell}}, \algo D_{\algo X}}^{\QCP_{\CC,h}}\big(1^\lambda, (\algo P,\algo A,\algo B,\algo C)\big) =1\right]  \leq p^{\setft{triv}}_{\algo D_{\mathscr{P}_{\mu,\ell}},\algo D_{\algo X}} + \gamma(\lambda).
\end{align*}
\end{proof}

\section{Quantum computing on encrypted data} \label{sec:qced}

The next application of parallel RSP we consider is quantum computing on encrypted data. At a high level, this is a protocol between a weakly-quantum client, $\mathcal{C}$, and a quantum server, $\mathcal{S}$. By ``weakly-quantum'' we mean that the client can perform some quantum operations (preparing, storing and applying unitaries on quantum states) but does not have the ability to evaluate a generic quantum circuit, $C$. For this reason, it will interact with $\mathcal{S}$ in order to delegate the application of $C$ on some input quantum state, $\rho_{\mathcal{C}}$. Before doing so, however, the client will encrypt $\rho_{\mathcal{C}}$, yielding a quantum ciphertext $\ct$. In addition, it will also generate a quantum state, $\pi$, which will aid in evaluating the circuit $C$ on the encrypted input, but does not require full quantum capabilities to prepare. Both $\ct$ and $\pi$ are sent to the server and the two parties exchange classical (or quantum) messages for a number of rounds that is polynomial in the size of the input. At the end of the interaction, the client obtains a state $\ct^*$, which, ideally, is an encryption of $C \rho_{\mathcal{C}} C^{\dagger}$. Decrypting this state, should result in the client's desired output.
The fundamental property of such a scheme is that, throughout the protocol, the server learns nothing about the client's input state, $\rho_{\mathcal{C}}$ (or the output state). This should be true even if the server deviates from the client's instructions and behaves maliciously.

Formalizing this intuition, we have the following:

\begin{definition}[Quantum Computing on Encrypted Data]\label{def:qc-enc}
A protocol for quantum computing on encrypted data is a tuple of $\QPT$ algorithms $\QCED = (\Setup, \Enc, \Evaluate, \Dec)$, where

\begin{description}
\item $\QCED.\Setup(C, 1^n) \to (\sk=(\sk_{in}, \sk_{comp}), \pi)$: takes as input a quantum circuit $C$, acting on $n$ qubits, and outputs two classical keys $\sk_{in}, \sk_{comp} \in \{0, 1\}^{\poly(n)}$, jointly denoted as $\sk$, and a quantum state $\pi$, on $\poly(n)$ qubits.

\item $\QCED.\Enc(\sk, \rho_{\mathcal{C}}) \to \ct$: takes as input a key $\sk$ and an $n$-qubit state $\rho_{\mathcal{C}}$ and outputs an $n$-qubit state, $\ct$, representing an encrypted version of $\rho_{\mathcal{C}}$.

\item $\QCED.\Evaluate(\mathcal{C}(C, \sk, \ct, \pi), \mathcal{S}(\rho_{\mathcal{S}})) \to 
(\sk^*, \ct^*, \tau_{\mathcal{S}})$: is an interactive protocol between a client, $\mathcal{C}$, and a server, $\mathcal{S}$. The client's input consists of the circuit $C$, a classical key $\sk$, and two quantum states $\ct$ and $\pi$. The server's input is the quantum state\footnote{While in what follows we take $\rho_{\mathcal{S}}$ to be $\proj{0}^{\otimes \poly(n)}$, it should be noted that $\rho_{\mathcal{S}}$ can in principle be any state that depends only on $n$. It can be thought of as a \emph{quantum advice} state of the prover (in the sense of~\cite{nishimura2004polynomial}), so that the overall security of the protocol holds against non-uniform adversaries.} $\rho_{\mathcal{S}}$. At the end of the interaction, the client obtains an output key $\sk^*$ and the $n$-qubit quantum state $\ct^*$. The server obtains the quantum state $\tau_{\mathcal{S}}$.

\item $\QCED.\Dec(\sk^*, \ct^*) \to \tau_{\mathcal{C}}$: takes as input a key, $\sk^*$ and a state $\ct^*$ and either outputs an $n$-qubit state $\tau_{\mathcal{C}}$ or sets $\tau_\mathcal{C}$ to an abort symbol $\proj{\bot}$.
\end{description}

We now define the following properties associated with a $\QCED$ protocol:
\begin{enumerate}
    \item \emph{$\delta$-correctness:} for every quantum circuit $C$ and every $n$-qubit client input $\rho_{\mathcal{C}}$, provided the server follows the instructions of the interactive protocol, $\QCED.\Evaluate$, we have that
    $$\| \QCED.\Dec(\sk^*, \ct^*) - C \rho_{\mathcal{C}} C^{\dagger} \|_{1} \leq \delta(n)$$
where,    
	$$ (\sk^*, \ct^*) \leftarrow \QCED.\Evaluate \left( \mathcal{C}(C, \sk, \ct, \pi), \mathcal{S}\left( \ketbra{0}{0}^{\otimes \poly(n)} \right) \right)$$
	$$ \ct \leftarrow \QCED.\Enc(\sk, \rho_{\mathcal{C}})$$
	$$ (\sk, \pi) \leftarrow \QCED.\Setup(C, 1^n) $$
    \item \emph{$\eps$-security:} for each quantum circuit $C$ and each QPT server $\mathcal{S}$, there exists a QPT algorithm known as a simulator, denoted $\Sim$, such that for every $n$-qubit client input $\rho_{\mathcal{C}}$ and every server input $\rho_\mS$ it is the case that
    $$ \| \Sim(\rho_\mS) - \QCED.\Evaluate(\mathcal{C}(C, \sk, \ct, \pi), \mathcal{S}(\rho_\mS)) \|_1 \leq \eps(n)$$
    with
  	$$ \ct \leftarrow \QCED.\Enc(\sk, \rho_{\mathcal{C}})$$
	$$ (\sk, \pi) \leftarrow \QCED.\Setup(C, 1^n) $$
	where $\Sim$ and $\QCED.\Evaluate$ are CPTP maps of the form $\rho_{\mathcal{S}} \mapsto (\ct^*, \tau_{\mathcal{S}})$, mapping states from the server's input register to its output space (which comprises the quantum ciphertext returned to the client, $\ct^*$, and any side information the server retains, $\tau_{\mathcal{S}}$)\footnote{Note that the states $\ct^*$ and $\tau_{\mathcal{S}}$ can be entangled.}.
    
\end{enumerate}
\end{definition}

A number of $\QCED$ protocols have been proposed~\cite{arrighi2006blind, broadbent2009universal, fitzsimons2017private}. Most of these have as a primary objective to not only hide the input to the computation $C$, but to also hide the computation itself. For this reason, such protocols are referred to as \emph{blind quantum computing} protocols.
Of course, it should be clear that a $\QCED$ protocol, as defined above, also allows for hiding of the computation $C$. This is done by delegating to the server a universal circuit $C_U$ and with the encrypted input consisting of not just the client's state $\rho_{\mathcal{C}}$ but also a description of the circuit $C$, to be applied on $\rho_{\mathcal{C}}$.

It should also be noted that the definition of a $\QCED$ can be easily satisfied by a protocol in which the client performs the quantum computation itself and never interacts with the server. Naturally, this requires the client to have the capability of performing universal quantum computations. Instead, the $\QCED$ protocols that have been proposed are such that the client only needs to prepare or measure single qubit states and otherwise has no quantum capabilities. 
Here we focus on one such protocol from~\cite{broadbent2015delegating}. For a formal description of the protocol, we refer to~\cite{broadbent2015delegating}.
For our purposes, the important feature of this protocol that we use is that the client encrypts its input using a \emph{quantum one-time pad} (QOTP) and the state $\pi$ it generates is a tensor product of BB84 states. More formally:

\begin{theorem}[\cite{broadbent2015delegating}] \label{thm:broadbent_proto}
There exists a $(\delta, \eps)$-$\QCED$ protocol, with $\delta(n) = \eps(n) = 0$, for which:
\begin{enumerate}
\item[1.] $\QCED.\Setup(C, 1^n) \to ((\sk_{in}, \sk_{comp}), \pi)$ with 
 $$\sk_{in} = (\vec{a}, \vec{b}), \;\; \vec{a}, \vec{b} \rand \bit^{n}$$ 
 $$\sk_{comp} = (\vec{v}, \vec{\theta}), \;\; \vec{v}, \vec{\theta} \rand \bit^{|C|}$$
 $$\pi =  \bigotimes_{i=1}^{|C|} H^{\theta_i} \proj{v_i} H^{\theta_i}.$$
 where $|C|$ denotes the size, or number of gates\footnote{Strictly speaking, in the protocol of~\cite{broadbent2015delegating}, the length of $\sk_{comp}$ and the number of qubits in $\pi$ are equal to the number of $T$ gates in the circuit $C$.}, in $|C|$.
\item[2.] $\QCED.\Enc(((\vec{a}, \vec{b}), \sk_{comp}), \rho_{\mathcal{C}}) \to \ct$ with
 $$ \ct =  \sigma_{X}(\vec{a}) \sigma_{Z}(\vec{b}) \; \rho_{\mathcal{C}} \; \sigma_{Z}(\vec{b}) \sigma_{X}(\vec{a}).$$
\item[3.] $\QCED.\Dec(\sk^*, \ct^*) \to \sigma_C$ with 
 $$ \sigma_C =\sigma_{X}(\vec{a}^*) \sigma_{Z}(\vec{b}^*) \; \ct^{*} \; \sigma_{Z}(\vec{b}^*) \sigma_{X}(\vec{a}^*),$$
 whenever $\sk^* = (((\vec{a}^*, \vec{b}^*), \sk^*_{comp})$ and otherwise the client sets $\sigma_{\mathcal{C}} = \proj{\bot}$.
\end{enumerate}
\end{theorem}

\noindent For the special case in which the client's input (the state $\rho_{\mathcal{C}}$) is entirely classical, note that it needs to prepare only BB84 states and the quantum one-time pad becomes a classical one-time pad. This will be the case of interest for us when switching to a completely classical client in the next subsection.

The security achieved by schemes for $\QCED$ is known as \emph{simulation security}, owing its name to the existence of the simulator, $\textsf{Sim}$, which reproduces the behavior of $\mathcal{S}$. This notion of security is equivalent to \emph{composable security}~\cite{canetti2001universally, maurer2011abstract}, allowing for the protocol to be composed in sequence or in parallel with other protocols.  

Let us also consider another version of security that is standard in the cryptography literature, namely that of indistinguishability under chosen-plaintext-attacks: 

\begin{definition}[(Adaptive) $\eps$-\textsf{IND-CPA}-security of $\QCED$] \label{def:indcpaqced}
Let $\QCED = (\Setup, \Enc, \Evaluate, \Dec)$ be a protocol for quantum computing on encrypted data. For some quantum circuit $C$ acting on $n$ qubits, consider the following security experiment between a computationally unbounded adversary $\algo A$ and a referee, $\algo R$:

\begin{enumerate}
\item[1.] $\algo A$ generates two $n$-qubit quantum states $\rho_0$ and $\rho_1$. It sends these states to the referee.
\item[2.] The referee chooses $b \rand \bit$ and keeps $\rho_b$. It then performs $\QCED.\Setup(C, 1^n) \to (\sk, \pi)$ and $\QCED.\Enc(\sk, \rho_b) \to \ct$.
\item[3.] The referee and $\algo A$ engage in the interactive protocol defined by $\QCED.\Evaluate(\mathcal{R}, \mathcal{A})$.
\item[4.] The adversary $\algo A$ outputs a bit $b' \in \bit$.
\end{enumerate} 

It should be the case that
\begin{equation}
\pr { b' = b } \leq \frac{1}{2} + \eps(n)
\end{equation}
where $\eps(n)$ is referred to as the adversary's advantage in the security experiment.

\end{definition}

Note that in the above definition the adversary is assumed to be computationally unbounded. When we transition to the setting of poly-time-bounded adversaries, we can simply modify this definition so as to consider only $\QPT$ adversaries $\algo A$.
It should be clear that simulation security is the stronger notion of security, in the sense that simulation security implies \textsf{IND-CPA}-security, while the converse is not true. However, for the purpose of designing quantum protocols that use only classical communication, we will use \textsf{IND-CPA}-security, as in the previous schemes. The reason for this stems from the difficulties in achieving simulation security in the setting where the client and server interact classically, as explained in ~\cite{rsp,badertscher2020security}.
However, we expect that with ideas from~\cite{rsp} and the additional assumption of a measurement buffer introduced in that paper, one could strengthen our security proof to satisfy the simulation based criterion.

\subsection{Quantum computing on encrypted data with a classical client}

We start by defining a $\QCED$ protocol that only involves classical communication between the client and the server (and where, additionally, the client is entirely classical).

\begin{definition}[Quantum Computing on Encrypted Data with a Classical Client]\label{def:qcedcc}
A protocol for quantum computing on encrypted data with a classical client is a tuple of polynomial-time classical and quantum algorithms $\QCED = (\Setup, \StatePrep, \Evaluate, \Dec)$, where

\begin{description}
\item $\QCEDCC.\Setup(C, 1^n) \to (\sk, \pk)$: takes as input a quantum circuit $C$, acting on $n$ qubits and with $|C| = \poly(n)$, and outputs two classical keys $\sk, \pk \in \{0, 1\}^{\poly(n)}$, referred to as the secret and public key, respectively.

\item $\QCEDCC.\StatePrep(\mathcal{C}(\sk, \pk), \mathcal{S}(\mathcal{\pk})) \to (\tau, \sigma_{\mathcal{S}})$: is an interactive protocol between a $\PPT$ client, $\mathcal{C}$, and a $\QPT$ server $\mathcal{S}$. The client's input consists of a secret key $\sk$ and a public key $\pk$. The server's input is the public key $\pk$. At the end of the interaction both parties either obtain the classical transcript $\tau \in \bit^{\poly(n)}$ and the server also obtains the quantum state $\sigma_{\mathcal{S}}$ or the client aborts.

\item $\QCEDCC.\Evaluate(\mathcal{C}(C, m_{\mathcal{C}}, \sk, \pk, \tau), \mathcal{S}(\pk, \tau, \sigma_{\mathcal{S}})) \to 
(\sk^*, \ct^*, \gamma_{\mathcal{S}})$: is an interactive protocol between a $\PPT$ client, $\mathcal{C}$, and a $\QPT$ server, $\mathcal{S}$. The client's input consists of the circuit $C$, an input for that circuit denoted as the string $m_{\mathcal{C}} \in \bit^n$, the secret key $\sk$, the public key $\pk$ and the classical transcript $\tau$. The server's input is the quantum state $\sigma_{\mathcal{S}}$ as well as the public key $\pk$ and the transcript $\tau$. At the end of the interaction, the client obtains an output key $\sk^* \in \bit^{\poly(n)}$ and a classical ciphertext $\ct^* \in \bit^{\poly(n)}$. The server obtains the quantum state $\gamma_{\mathcal{S}}$.

\item $\QCEDCC.\Dec(\sk^*, \ct^*) \to r$: takes as input a key, $\sk^*$ and ciphertext $\ct^*$ and outputs the string $r \in \bit^{\poly(n)}$. If the string $\ct^*$ is not of the appropriate size, the client aborts.
\end{description}

We now define the following properties associated with a $\QCEDCC$ protocol:
\begin{enumerate}
    \item \emph{$\delta$-correctness:} for every quantum circuit $C$, with $|C| = \poly(n)$, and every input $m_{\mathcal{C}} \in \bit^n$, provided the server follows the instructions of the interactive protocol, $\QCEDCC.\Evaluate$, we have that
    $$\TVD [ \QCEDCC.\Dec(\sk^*, \ct^*), M(C \ket{m_{\mathcal{C}}}\bra{m_{\mathcal{C}}} C^{\dagger}) ] \leq \delta(n)$$
where $M(\cdot)$ denotes the process of measuring a quantum state in the computational basis and where,    
\begin{align*}
(\sk^*, \ct^*) &\leftarrow \QCEDCC.\Evaluate
(\mathcal{C}(C, m_{\mathcal{C}}, \sk, \pk, \tau), \mathcal{S}(\pk, \tau, \sigma_{\mathcal{S}}))\,,\\
(\tau, \sigma_{\mathcal{S}}) &\leftarrow \QCEDCC.\StatePrep(\mathcal{C}(\sk, \pk), \mathcal{S}(\mathcal{\pk})) \,,\\
(\sk, \pk) &\leftarrow \QCEDCC.\Setup(C, 1^n)\,.
\end{align*}
\item \emph{(Adaptive) $\eps$-\textsf{IND-CPA} security:} For a quantum circuit $C$ acting on $n$ qubits, consider the following security experiment between a $\QPT$ adversary $\algo A$ and a referee, $\algo R$:

\begin{enumerate}
\item[1.] $\algo A$ generates two $n$-bit strings $m_0$ and $m_1$. It sends these to the referee.
\item[2.] The referee chooses $b \rand \bit$ and keeps $m_b$. It then performs $\QCEDCC.\Setup(C, 1^n, 1^{\lambda}) \to (\sk, \pk)$.
\item[3.] The referee and the adversary engage in the protocol $\QCEDCC.\StatePrep(\mathcal{R}(\sk, \pk), \mathcal{A}(\mathcal{\pk})) \to (\tau, \sigma_{\mathcal{A}})$.
\item[4.] Provided the referee did not abort in step 3, the referee and $\algo A$ engage in the interactive protocol defined by $\QCEDCC.\Evaluate(\mathcal{R}(C, m_b, \sk, \pk, \tau), \mathcal{A}(\pk, \tau, \sigma_{\mathcal{A}}))$.
\item[5.] The adversary $\algo A$ outputs a bit $b' \in \bit$.
\end{enumerate} 

It should be the case that
\begin{equation}
| \pr{ b' = 0 \text{ $\wedge$ $\algo R$ did not abort} | b=0 } - \pr{ b' = 0 \text{ $\wedge$ $\algo R$ did not abort} | b=1} | \leq  \eps(n)
\end{equation}
where $\eps(n)$ is referred to as the adversary's advantage in the security experiment.
    
\end{enumerate}

\end{definition}
Note that unlike~\ref{def:indcpaqced}, the security definition for \QCEDCC includes a condition about whether the referee (or client) aborted in the protocol or not. This is because the state preparation subroutine, $\QCEDCC.\StatePrep$, that is meant to mimic the sending of quantum states, can fail if the adversary (or server) behaves maliciously. As such, the security definition in this case states that the adversary cannot both convince the referee to accept in $\QCEDCC.\StatePrep$ and learn something about the referee's chosen input.

We now turn our attention to constructing a protocol that satisfies \ref{def:qcedcc}.
At a high level, the idea is to start from the protocol from~\cite{broadbent2015delegating} described in \ref{thm:broadbent_proto} for the special case of classical inputs $m_C$ and replace the step of sending random BB84 states with our parallel RSP protocol (\ref{fig:protocol_multi_round})
We give a formal description of this protocol in~\ref{fig:QCEDCC_protocol}.
The following theorem asserts that~\ref{fig:QCEDCC_protocol} indeed satisfies the conditions in \ref{def:qcedcc}.

\begin{table}[ht!]
\begin{longfbox}[breakable=false, padding=1em, padding-right=1.8em, padding-top=1.2em, margin-top=1em, margin-bottom=1em, background-color=gray!20]
\begin{protocol} {\bf Protocol for Quantum Computing on Encrypted Data with a Classical Client} \label{fig:QCEDCC_protocol}\end{protocol}
Let $C$ be a circuit acting on $n$ qubits. Additionally, take $\delta = 1/\poly(n)$ to be some error tolerance parameter.
Define the following algorithms.

\begin{description}
\item $\QCEDCC.\Setup(C, 1^n) \to (\sk, \pk)$:
Run the procedure to sample keys and trapdoors $(k_1, t_{k_1}; \dots k_M, t_{k_M})$, as in~\ref{fig:protocol_test} and~\ref{fig:protocol_prep}, $M=\poly(n)$ is the total amount of keys used, as determined by~\ref{fig:protocol_multi_round}, with error tolerance $\delta$.
In addition, take $\sk_{in} \rand \bit^n$. Set $\pk = \{ k_i \}_i, \; \sk = \{ \{ t_{k_i} \}_i, \sk_{in} \}$.
\item $\QCEDCC.\StatePrep(\mathcal{C}(\sk, \pk), \mathcal{S}(\mathcal{\pk})) \to (\tau, \sigma_{\mathcal{S}})$: The client runs~\ref{fig:protocol_multi_round} with security parameter $\lambda=\poly(n)$ with the server. Here, $\tau$ will denote the transcript of this protocol and $\sigma_{\mathcal{S}}$ will be the server's state upon completing the protocol (in the case where the client accepts).
\item $\QCEDCC.\Evaluate(\mathcal{C}(C, m_{\mathcal{C}}, \sk, \pk, \tau), \mathcal{S}(\pk, \tau, \sigma_{\mathcal{S}})) \to 
(\sk^*, \ct^*, \gamma_{\mathcal{S}})$: The client and server perform $\QCED.\Evaluate$ from Broadbent's protocol~\cite{broadbent2015delegating}, mentioned in the previous section. Specifically, the client uses $\sk$ and $\tau$ to determine the BB84 states the server should have (upon completing $\QCEDCC.\StatePrep$). The key $\sk_{in}$ is used as a one-time pad for the input $m_{\mathcal{C}}$. The client then performs the classical interaction from Broadbent's protocol with the server, as if the BB84 states determined by $\sk$ and $\tau$ were sent via a quantum channel.
\item $\QCEDCC.\Dec(\sk^*, \ct^*) \to r$: Same as the $\QCED.\Dec$ operation in Broadbent's protocol, but restricted to classical outputs.
\end{description}
\end{longfbox}
\end{table}

\begin{theorem}
\ref{fig:QCEDCC_protocol} is a protocol for Quantum Computing on Encrypted Data with a Classical Client having $\negl(n)$ correctness and $1/\poly(n)$ adaptive \textsf{IND-CPA} security.
\end{theorem}
\begin{proof}
Correctness follows from the completeness of~\ref{fig:protocol_multi_round} (see \ref{prop:honest-prover-correctness}) together with the correctness of Broadbent's $\QCED$ protocol.

To show \textsf{IND-CPA} security, let us consider the security experiment in~\ref{def:qcedcc}. The adversary (or the server) chooses two inputs, $m_0$ and $m_1$ and sends them to the referee (or the client). The referee picks a random bit $b \in \{0, 1\}$ and then selects $m_b$. This is followed by $\QCEDCC.\Setup$ and the interactive protocol $\QCEDCC.\StatePrep$ between the referee and the adversary. 
As follows from~\ref{cor:conditional-state-non-abort} we have that the shared state between the two is:
\begin{align*}
\sigma_{SWDYR}^{(\vec\theta)} = \pr{\top} \proj{\top}_S \ot \sigma_{WDYR|\top}^{(\vec \theta)} + (1 - \pr{\top}) \proj{\bot}_S \ot \sigma_{WDYR|\bot}^{(\vec \theta)} 
\end{align*}
where we recall that $S$ is the register denoting acceptance or rejection in the protocol, $W$ is a register of the referee and $DYR$ are registers of the adversary.
As in the proof of~\ref{thm:unclonable-security-CC}, we can assume w.l.o.g.~that conditioned on acceptance, the prover in fact has prepared the state $V \sigma_{WDYR|\top}^{(\vec \theta)} V^\dagger$, not $\sigma_{WDYR|\top}^{(\vec \theta)}$, where $V$ is the isometry from~\ref{cor:conditional-state-non-abort}.
Slightly abusing notation, we denote this state as 
\begin{align*}
\sigma_{WQB|\top}^{(\vec \theta)} \deq V \sigma_{WDYR|\top}^{(\vec \theta)} V^\dagger \,,
\end{align*}
with $B \deq DYRA$ and $V: DYR \to QB$ as in~\ref{cor:conditional-state-non-abort}.
Then, \ref{cor:conditional-state-non-abort} states that for any $\vec \theta$,
\begin{align*}
\pr{\top} \sigma_{WQB|\top}^{(\vec \theta)} &\capprox_{1/\poly(n)} \pr{\top} \frac{1}{2^n}  \sum_{\vec v \in \bits^n} \proj{\vec v}_W \ot \left( \bigotimes_{i} H^{\theta_i} \proj{v_i} H^{\theta_i}  \right)_Q   \ot \alpha'_{B}
\end{align*}
for some efficiently preparable state $\alpha'_{B}$

Next, conditioned on having accepted in $\QCEDCC.\StatePrep$, the referee $\algo R$ and the adversary $\algo A$ engage in the protocol $\QCEDCC.\Evaluate$. As mentioned, this procedure is identical to the one in Broadbent's protocol, except for the fact that the referee, in this case, does not send BB84 states to the adversary, instead assuming those states were prepared via $\QCEDCC.\StatePrep$. 
Upon completing this step, the adversary will output the bit $b'$. We are interested in the adversary's advantage in guessing the referee's chosen input. We proceed via an argument by contradiction. Suppose the adversary guesses the referee's input with advantage greater than $\eps(n)$, i.e.
\begin{equation}
| \pr{ b' = 0 \text{ $\wedge$ $\algo R$ did not abort} | b=0} - \pr{ b' = 0 \text{ $\wedge$ $\algo R$ did not abort} | b=1} | >  \eps(n). \label{eqn:assume_eps}
\end{equation}
We will consider two cases. First, suppose the referee's probability of not aborting is less than $\eps(n)$, or, using the notation from~\ref{cor:conditional-state-non-abort}, $\pr{\top} < \eps(n)$. Then, using the fact that the referee aborting in $\QCEDCC.\StatePrep$ is independent of the choice of $b$, we immediately find a contradiction to \ref{eqn:assume_eps}.

In what follows, we will assume $\pr{\top} \geq \eps(n)$ and take $\eps(n) = 1/\poly(n)$. In this case, we get from~\ref{cor:conditional-state-non-abort} that
\begin{align}
\sigma_{WQB|\top}^{(\vec \theta)} &\capprox_{1/\poly(n)} \frac{1}{2^n} \sum_{\vec v \in \bits^n} \proj{\vec v}_W \ot \left( \bigotimes_{i} H^{\theta_i} \proj{v_i} H^{\theta_i}  \right)_Q   \ot \alpha'_{B} \,,
\end{align} \label{eqn:qced_indist_prf}
where $\alpha'_{B}$ is an efficiently preparable (normalised) quantum state.

We would now like to construct a new adversary, $\algo A'$, that breaks the security of Broadbent's \QCED protocol. First, suppose $\algo A'$ interacts with the referee in the security experiment of~\ref{def:indcpaqced}, up to step 3.\footnote{We are assuming that possible inputs of that security experiment are the classical strings $m_0$ and $m_1$ in order to match the experiment where the two interact classically.} At this point, the adversary has received random BB84 states from the referee. After adding to this the (fixed and efficiently preparable) state $\alpha'_B$, adversary $\algo A'$ has a state that is $1/\poly(n)$ computationally indistinguishable from $\sigma_{WQB|\top}^{(\vec \theta)}$. Now, $\algo A'$ runs the original adversary $\algo A$.
We know that $\algo A$ achieves advantage greater than $\eps(n)$ in the security experiment when run on the state $\sigma_{WQB|\top}^{(\vec \theta)}$. 
Since $\algo A$ is computationally efficient, it follows that its distinguishing advantage can change by at most $1/\poly(n)$ when run on the state prepared by $\algo A'$.
By choosing appropriate parameters for the execution of \ref{fig:protocol_multi_round}, this $1/\poly(n)$ error can be made sufficiently small, e.g.~less than $\eps(n)/2$.
This immediately implies that $\algo A'$ obtains advantage greater than $\eps(n)/2$ in the security experiment of~\ref{def:indcpaqced}, which contradicts the adaptive \textsf{IND-CPA} security of Broadbent's \QCED protocol.
It follows that $\algo A$'s advantage must be upper bounded by $\eps(n)$, concluding the proof.
\end{proof}

\section{Verifiable delegated blind quantum computation}

Our final application of parallel RSP is verifiable delegated blind quantum computation (VDBQC). We start with an overview of VDBQC and then show how one can obtain such a protocol which uses only classical communication. Much like quantum computing on encrypted data, VDBQC protocols are interactive protocols between two parties that are here referred to as \emph{verifier} and \emph{prover}. The verifier delegates some quantum computation, in the form of a quantum circuit $C$, to a $\QPT$ prover. The ``blindness'' of VDBQC protocols refers to the fact that the verifier's input should remain hidden from the prover\footnote{As mentioned in~\ref{sec:qced}, the original definition of blind quantum computing~\cite{arrighi2006blind, broadbent2009universal} refers to hiding both the input to the circuit and the circuit itself from the prover. However, this is equivalent to keeping only the input hidden, as the verifier can always delegate some universal circuit $C_U$, which takes as input a description of a circuit $C$ and an input string $x$.}, just like in QCED protocols. 
In addition to this condition, it must also be the case that the verifier accepts with high probability when the output it obtains after interacting with the prover is the correct one (corresponding to having performed the computation $C$), and otherwise it rejects with high probability (this is called \emph{verifiability}).

Formalizing this intuition, we have the following:

\begin{definition}[Verifiable Delegated Blind Quantum Computing]\label{def:vdbqc}
A protocol for verifiable delegated blind quantum computing is a tuple of $\QPT$ algorithms $\VDBQC = (\Setup, \Enc, \Evaluate, \Dec)$, where

\begin{description}
\item $\VDBQC.\Setup(C, 1^n) \to (\sk, \pi)$: takes as input a quantum circuit $C$, acting on $n$ qubits, and outputs a classical key $\sk \in \{0, 1\}^{\poly(n)}$, and a quantum state $\pi$, on $\poly(n)$ qubits.

\item $\VDBQC.\Enc(\sk, \rho_{\mathcal{V}}) \to \ct$: takes as input a key $\sk$ and an $n$-qubit state $\rho_{\mathcal{V}}$ and outputs an $n$-qubit state, $\ct$, representing an encrypted version of $\rho_{\mathcal{V}}$.

\item $\VDBQC.\Evaluate(\mathcal{V}(C, \sk, \ct, \pi), \mathcal{P}(\rho_{\mathcal{P}})) \to 
(\sk^*, \ct^*, \tau_{\mathcal{P}})$: is an interactive protocol between a verifier, $\mathcal{V}$, and a prover, $\mathcal{P}$. The verifier's input consists of the circuit $C$, a classical key $\sk$, and two quantum states $\ct$ and $\pi$. The prover's input is the quantum state $\rho_{\mathcal{P}}$. At the end of the interaction, the verifier obtains an output key $\sk^*$ and a quantum state $\ct^*$.
The prover obtains the quantum state $\tau_{\mathcal{P}}$.

\item $\VDBQC.\Dec(\sk^*, \ct^*) \to \tau_{\mathcal{V}}$: takes as input a key $\sk^*$ and an $n$-qubit state $\ct^*$ and outputs an $poly(n)$-qubit state $\tau_{\mathcal{V}}$ such that $\tau_{\mathcal{V}} = \sigma \otimes \phi$, and either $\ct^* \leftarrow \VDBQC.\Enc(\sk^*, \sigma)$ and $\phi = \proj{\top}$ or $\sigma$ is an arbitrary state and $\phi = \proj{\bot}$.

\end{description}

We now define the following properties associated with a $\VDBQC$ protocol:
\begin{enumerate}
    \item \emph{$\delta$-correctness:} for every quantum circuit $C$ and every $n$-qubit verifier input $\rho_{\mathcal{V}}$, provided the prover follows the instructions of the interactive protocol, $\VDBQC.\Evaluate$, we have that
    $$\| \VDBQC.\Dec(\sk^*, \ct^*) - C \rho_{\mathcal{V}} C^{\dagger} \otimes \proj{\top} \|_{1} \leq \delta(n),$$
where,    

\begin{align*}
(\sk^*, \ct^*) &\leftarrow \VDBQC.\Evaluate \left( \mathcal{V}(C, \sk, \ct, \pi), \mathcal{P}\left( \ketbra{0}{0}^{\otimes \poly(n)} \right) \right)\\
\ct &\leftarrow \VDBQC.\Enc(\sk, \rho_{\mathcal{V}})\\
(\sk, \pi) &\leftarrow \VDBQC.\Setup(C, 1^n)
\end{align*}

    \item \emph{$\eps$-blind-verifiability:} for each quantum circuit $C$ and each QPT prover $\mathcal{P}$, there exists a QPT algorithm known as a simulator, denoted $\Sim$, such that for every $n$-qubit verifier input $\rho_{\mathcal{V}}$ and every prover input $\rho_{\mathcal{P}}$ it is the case that 
    $$ \| \Sim(\rho_{\mathcal{P}}) - \VDBQC.\Evaluate(\mathcal{V}(C, \sk, \ct, \pi), \mathcal{P}(\rho_{\mathcal{P}})) \|_1 \leq \eps(n)$$
and there exist $0 \leq p \leq 1$ and a state $\rho_{inc}$ on $\poly(n)$ qubits such that 
	$$ \| \tau_{\mathcal{V}} - \rho_{ver} \|_{1} \leq \eps(n) $$
    with 
    $$\rho_{ver} = p C \rho_{\mathcal{V}} C^{\dagger} \otimes \ketbra{\top}{\top} + (1 - p) \rho_{inc} \otimes \ketbra{\bot}{\bot}, $$
    $$ \tau_{\mathcal{V}} \leftarrow \VDBQC.\Dec(\sk^*, \ct^*)$$
    $$(\sk^*, \ct^*, \tau_{\mathcal{P}}) \leftarrow \VDBQC.\Evaluate \left( \mathcal{V}(C, \sk, \ct, \pi), \mathcal{P}\left( \rho_{\mathcal{P}} \right) \right)$$     
  	$$ \ct \leftarrow \VDBQC.\Enc(\sk, \rho_{\mathcal{V}})$$
	$$ (\sk, \pi) \leftarrow \QCED.\Setup(C, 1^n) $$
where $\Sim$ and $\VDBQC.\Evaluate$ are CPTP maps of the form $\rho_{\mathcal{P}} \mapsto (\ct^*, \tau_{\mathcal{P}})$, mapping states from the prover's input register to its output space (which comprises the quantum ciphertext returned to the verifier, $\ct^*$, and any side information the prover retains, $\tau_{\mathcal{P}}$).
    
\end{enumerate}

\end{definition}

There are a number of protocols that satisfy the above definition and we refer the reader to~\cite{gheorghiu2019verification} for an overview. We should also note that the above definition does not apply to all $\VDBQC$ protocols in the literature. In a bit more detail, using terminology from~\cite{gheorghiu2019verification}, the above definition generally refers to single-prover prepare-and-send protocols,\footnote{The definition also applies for single-prover receive-and-measure protocols~\cite{gheorghiu2019verification}.} where the verifier sends quantum states to the prover.
This is precisely the setting we are interested in.
Specifically, we would like a $\VDBQC$ protocol in which the verifier has to prepare and send only BB84 states to the prover. 
As far as we are aware, there is no protocol that achieves this directly. However, it can be achieved indirectly using a protocol of Morimae~\cite{morimae2018blind}. To explain, we need to give a high-level description of that protocol. We start by first explaining the \emph{history state construction}~\cite{kitaev2002classical, kempe2006complexity}.

\subsection{History state construction and post-hoc verification} \label{subsec:histstate}

The history state construction, introduced in~\cite{kitaev2002classical, kempe2006complexity}, was used to show that the \emph{local Hamiltonian problem} is complete for the complexity class \textsf{QMA} (the quantun analogue of \textsf{NP}). For our purposes, we forego the definition of \textsf{QMA} and instead explain how the history state construction can be used to verify general quantum computations.

Given a circuit $C = U_{T} U_{T-1} ... U_{1}$ acting on $n$ qubits and an input $x \in \{0, 1\}^n$, where $T = |C|$ and $U_i$ denotes the $i$'th gate in $C$ (and with $U_0 = \id$), consider the state
\begin{equation}
\ket{\psi_{C, x}} = \frac{1}{\sqrt{T+1}} \sum \limits_{t=0}^{T}  \ket{t} U_{t} U_{t-1} ... U_{0} \ket{x}.
\end{equation}
This is referred to as the history state of $C$ on input $x$. Importantly, $\ket{\psi_{C, x}}$ is the ground state of a $k$-local Hamiltonian denoted $H_{C,x}$, with $k = O(1)$. Whenever $C$ accepts $x$ (i.e. measuring the first qubit of $C \ket{x}$ yields outcome $1$ with probability greater than $2/3$) the state $\ket{\psi_{C, x}}$ has low energy with respect to $H_{C, x}$ otherwise it has high energy. More precisely, the gap between accepting and rejecting instances (referred to as the \emph{promise gap}) is at least $1/\poly(n)$.
It should also be noted that $\ket{\psi_{C, x}}$ can be prepared efficiently, with a circuit of size $\poly(|C|, |x|)$ applied to a $\ket{00...0}$ state.

This yields a simple protocol for verifying quantum computations: a quantum prover is instructed to prepare $\ket{\psi_{C, x}}$ and send the state to the verifier, who will then measure it with one of the local terms of $H_{C, x}$. The local term will be chosen with a probability proportional to its ``weight'' in $H_{C, x}$. In other words, it should be the case that the expectation of the verifier's measurement is $\bra{\psi_{C, x}} H_{C, x} \ket{\psi_{C,x}}$. This provides an estimate for the state's energy which can be used to decide acceptance or rejection.
Thus, if $C$ accepts $x$, the history state (when measured according to $H_{C, x}$) will have low energy. In contrast, whenever $C$ rejects $x$ \emph{all} states will have high energy with respect to $H_{C, x}$ so that a malicious prover cannot convince the verifier to accept.
It should also be noted that one can construct local Hamiltonians that are $2$-local and for which the non-trivial components of the Hamiltonian terms are either Pauli $\sigma_X$ or $\sigma_Z$~\cite{biamonte2008realizable}.

The protocol we sketched above is referred to as \emph{post-hoc} verification of quantum computation and was introduced in~\cite{morimae2016post, fitzsimons2018post}. We can summarise this result with the following theorem.

\begin{theorem}[\cite{fitzsimons2018post}] \label{thm:posthoc}
Let $C$ be a $T$-gate quantum circuit acting on $n$ qubits such that $C = U_{T} U_{T-1} ... U_{1}$ where $T = |C|$ and $U_i$ denotes the $i$'th gate in $C$ (and where we take $U_0 = \id$). Suppose it is the case that for any input $x \in \{0, 1\}^n$ either $|\bra{0} C \ket{x}|^2 \geq 2/3$ or $|\bra{0} C \ket{x}|^2 \leq 1/3$. Then, there exists a collection of $n$-qubit observables $\{ O_i \}_{i \leq m}$, with $m = \poly(n)$, of the form $O_i = \bigotimes_{j = 1}^{n} P_j$, with $P_j \in \{ \sigma_X, \sigma_Z, \id \}$ (and with at most two $P_j$'s not equal to $\id$), a probability distribution $p : [m] \to [0, 1]$, and values $\alpha, \beta > 0$ with $\beta - \alpha > 1/\poly(n)$ (referred to as the \emph{promise gap}), such that
\begin{itemize}
\item if $|\bra{0} C \ket{x}|^2 \geq 2/3$ then $\sum_{i=1}^m p(i) \bra{\psi_{C, x}} O_i \ket{\psi_{C, x}} \leq \alpha$,
\item if $|\bra{0} C \ket{x}|^2 \leq 1/3$ then for all states $\ket{\psi'}$, $\sum_{i=1}^m p(i) \bra{\psi'} O_i \ket{\psi'} \geq \beta$,
\end{itemize}
where
\begin{equation}
\ket{\psi_{C, x}} = \frac{1}{\sqrt{T+1}} \sum \limits_{t=0}^{T}  \ket{t} U_{t} U_{t-1} ... U_{0} \ket{x}.
\end{equation}
\end{theorem}
Here the observables $O_i$ are the local terms of the Hamiltonian $H_{C, x}$, whose groundstate is the history state $\ket{\psi_{C, x}}$. The probability distribution $p(i)$ is simply the weight of each local term in the Hamiltonian, suitably normalised. In the post-hoc protocol, the verifier samples an observable according to that distribution and measures the state received from the prover with that observable.

\begin{remark}[Instance independence] \label{rem:instindep}
An important observation about the post-hoc verification protocol, first made in~\cite{alagic2020non}, is that the sampling of an observable $O_i$ and the measurement performed by the verifier can be made independent of each other. This is done as follows. The verifier will measure each qubit of the state received from the prover uniformly at random with either the $\sigma_X$ or $\sigma_Z$ observable and then sample an observable $O_i$ according to the distribution $p$. Assuming the underlying Hamiltonian is $2$-local, let $1 \leq j, k \leq n$ denote the indices of the qubits acted upon by non-trivial operators in $O_i$ (i.e. operators not equal to $\id$). If these operators match the observables the verifier used when measuring qubits $i$ and $j$ (which happens with probability $1/4$), the verifier keeps the measurement outcome, discarding it otherwise (and rejecting). As shown in~\cite{alagic2020non}, the effect of this is to reduce the promise gap by a constant factor of $1/4$, so that overall the promise gap is still $1/\poly(n)$.
\end{remark}

\subsection{Morimae's verification protocol} \label{subsec:morimae}

The post-hoc protocol described in the previous subsection is a simple protocol for verifying general quantum computations. Note, however, that it is not a VDBQC protocol according to our definition. This is because the protocol is not blind---the prover is told both the computation $C$ and the input $x$ so as to prepare the history state. 

While the post-hoc protocol does not satisfy the blindness property, Morimae recently proposed a modified version of the protocol that is blind~\cite{morimae2018blind}.
This main idea of Morimae's protocol is to instruct the prover to prepare the history state using a $\QCED$ protocol.\footnote{Strictly speaking, the preparation should be done with a blind quantum computing protocol, i.e. a protocol that explicitly hides the computation from the prover. However, as mentioned previously, a $\QCED$ protocol also yields blind quantum computation by having the $\QCED$ computation be a universal circuit, $U$, with the property that that $U \ket{C, x} = C \ket{x} \otimes \ket{00...0}$, for some quantum circuit $C$ and input $x$.} In fact, the verifier delegates to the prover the preparation of the state $\ket{\psi_{C, x}}$ as well as $\ket{\psi_{\bar{C}, x}}$, where $\bar{C}$ is the same as $C$ except the output qubit is negated (that is, an additional $\sigma_X$ gate is added to the output qubit of $C$). 
As the states are prepared via a $\QCED$ scheme, they will be encrypted and the prover learns nothing about the input $x$ or the circuit $C$.
The prover then sends these states to the verifier who proceeds to decrypt and measure them using the observables from~\ref{thm:posthoc}.
In this way, the verifier is able certify the outcome of the computation $C$ on input $x$ while also keeping $C$ and $x$ hidden from the prover.
We can also see that if the $\QCED$ scheme that is used in the first part of the protocol is Broadbent's scheme~\cite{broadbent2015delegating}, mentioned in the previous section, the verifier will only be required to send BB84 states to the prover. Moreover, in that scheme, the encryption of the states is done via a quantum one-time pad. As this consists of the application of Pauli $\sigma_X$ and $\sigma_Z$ operations, these can be absorbed into the observables that the verifier measures, so that the overall protocol requires the verifier to send BB84 states and measure qubits in the standard and Hadamard bases (or with the observables $\sigma_Z$, $\sigma_X$, respectively).

Our goal is to dequantise this protocol. As we can see, however, this requires one to not only dequantise the sending of quantum states, but also the process of performing $\sigma_X$ and $\sigma_Z$ measurements. For this, we make use of the measurement subprotocol from Mahadev's verification protocol~\cite{mahadev}. This is described in the next subsection.

\subsection{Mahadev's measurement protocol} \label{subsec:mahadev}

In~\cite{mahadev}, Mahadev introduced the first protocol allowing a classical verifier to delegate and verify arbitrary quantum computations to a quantum prover. At the core of that protocol is a primitive for classically instructing a quantum prover to perform measurements in the standard and Hadamard bases. Importantly, this is done in a ``verifiable way'', so that if the prover does not perform the correct measurements, the verifier can detect this and abort.
Combined with the idea of post-hoc verification, we can see how this yields a protocol for verifying general quantum computations: simply instruct the prover to measure observables from \ref{thm:posthoc} on the history state of the computation.

For our purposes, we are only interested in the measurement subprotocol. For a detailed description of that protocol, we refer the reader to~\cite{mahadev}. At a high level, the protocol works in a similar way to our RSP protocol (which was directly inspired by Mahadev's result): the verifier uses an ENTCF family (based on the quantum intractability of LWE) to delegate the chosen measurement to the prover. As a result of interacting with the prover, the verifier obtains a purported measurement outcome and can decide whether to accept or reject. This decision is performed using the ENTCF family's trapdoor, which allows the verifier to check whether the prover performed the measurements as instructed or not.
The main result of this protocol is summarised in the following two theorems, taken from~\cite[Sections 5.4, 7]{mahadev}, that capture the completeness and soundness of the protocol:

\begin{theorem}[Completeness of measurement protocol~\cite{mahadev}] \label{thm:completemeas}
For all $n$-qubit states $\rho$ and all basis choices $h \in \{0, 1\}^n$ (where $h_i=0$ denotes measurement of qubit $i$ in the standard basis and $h_i=1$ denotes measurement in the Hadamard basis), an honest QPT prover, $\algo P$, is accepted by the verifier in the measurement protocol with probability $1 - \negl(n)$. Additionally, it is the case that
\begin{equation}
\TVD [D_{\algo P, h}, D_{\rho, h} ]  \leq \negl(n),
\end{equation}
where $D_{\algo P, h}$ is the probability distribution resulting from the verifier having delegated the measurement $h$ to $\algo P$ and having accepted in the protocol; $D_{\rho, h}$ is the distribution resulting from measuring the state $\rho$ in the bases determined by $h$.
\end{theorem}

\begin{theorem}[Soundness of measurement protocol~\cite{mahadev}] \label{thm:soundmeas}
Assuming the quantum intractability of LWE, for a QPT prover $\algo P$, let $1 - p$ be the probability that the verifier accepts the prover in the measurement protocol.\footnote{Strictly speaking, the measurement protocol has two types of rounds (a test round and a Hadamard round) and one should consider the probabilities of acceptance in each of these rounds. See~\cite[Section 7]{mahadev} for more details. Here we take $p$ to be the maximum of the two probabilities.} Then, there exists an $n$-qubit state $\rho$ and a prover $\algo P'$ such that for all basis choices $h$,
\begin{equation}
\TVD [ D_{\algo P, h}, D_{\algo P', h} ] \leq p + \sqrt{p} + \negl(n),
\end{equation}
and with $D_{\algo P', h}$ being computationally indistinguishable from $D_{\rho, h}$.
\end{theorem}

\subsection{Verifiable delegated blind quantum computation with a classical verifier}

Let us start by defining $\VDBQC$ with a classical verifier. Throughout this section we are considering $\VDBQC$ protocols for decision problems, i.e.~problems where the output is a single classical bit. In other words, for some computation $C$ and input $x$, we are only interested in the result of measuring the first qubit of $C \ket{x}$ in the standard basis.

\begin{definition}[Verifiable Delegated Blind Quantum Computing with a Classical Verifier]\label{def:vdbqccc}
A protocol for verifiable delegated blind quantum computing (of decision problems) with a classical verifier is a tuple of $\QPT$ algorithms $\VDBQCCC = (\StatePrep, \Evaluate, \Dec)$, where

\begin{description}

\item $\VDBQCCC.\StatePrep(\mathcal{V}(C, 1^n), \mathcal{P}(\rho_{\algo P})) \to (\sk, \pk, \tau, \sigma_{\mathcal{P}})$: is an interactive protocol between a $\PPT$ verifier, $\mathcal{V}$, and a $\QPT$ prover $\mathcal{P}$. The verifier's input consists of a description of a circuit $C$ acting on $n$ qubits. The prover has as input some $\poly(n)$-qubit quantum state $\rho_{\algo P}$. The protocol either aborts or succeeds\footnote{Here we are assuming that there is a global abort flag for the entire protocol.}. If the protocol succeeds the verifier obtains as an output a secret key $\sk \in \bit^{\poly(n)}$ and the prover obtains a $\poly(n)$-qubit quantum state $\sigma_{\algo P}$. Additionally, they both obtain a public key $\pk \in \bit^{\poly(n)}$ and a classical transcript $\tau \in \bit^{\poly(n)}$. 

\item $\VDBQCCC.\Evaluate(\mathcal{V}(C, m_{\mathcal{V}}, \sk, \pk, \tau), \mathcal{P}(\pk, \tau, \sigma_{\mathcal{P}})) \to 
(\sk^*, \ct^*, \gamma_{\mathcal{P}})$: is an interactive protocol between a $\PPT$ verifier, $\mathcal{V}$, and a $\QPT$ prover, $\mathcal{P}$. The verifier's input consists of the circuit $C$, an input for that circuit denoted as the string $m_{\mathcal{V}} \in \bit^n$, the secret key $\sk$, the public key $\pk$ and the classical transcript $\tau$. The prover's input is the quantum state $\sigma_{\mathcal{P}}$ as well as the public key $\pk$ and the transcript $\tau$. At the end of the interaction, the verifier obtains an output key $\sk^* \in \bit^{\poly(n)}$ and a classical ciphertext $\ct^* \in \bit^{\poly(n)}$. The prover obtains the quantum state $\gamma_{\mathcal{P}}$.

\item $\VDBQCCC.\Dec(\sk^*, \ct^*) \to r$: takes as input a key, $\sk^*$, and a ciphertext $\ct^*$ and either outputs the bit $r \in \bit$ or the protocol aborts.
\end{description}

We now define the following properties associated with a $\VDBQCCC$ protocol:
\begin{enumerate}
\item \emph{$\delta$-correctness:} for every quantum circuit $C$, with $|C| = \poly(n)$, and every input $m_{\mathcal{V}} \in \bit^n$, provided the prover follows the instructions of the interactive protocol, $\VDBQCCC.\Evaluate$, we have that
\begin{equation*}
\Pr[r = M(C \proj{m_{\mathcal{V}}} C^{\dagger}) \land \top ] \geq 1 - \delta(n)
\end{equation*}
where $M(\cdot)$ denotes the process of measuring the first qubit of a quantum state in the standard basis and outputting that result, $\bot$ denotes the event of the protocol aborting and where,
\begin{align*}
r &\leftarrow \VDBQCCC.\Dec(\sk^*, \ct^*) \\
(\sk^*, \ct^*, \gamma_{\algo P}) &\leftarrow \VDBQCCC.\Evaluate
(\mathcal{V}(C, m_{\mathcal{V}}, \sk, \pk, \tau), \mathcal{P}(\pk, \tau, \sigma_{\mathcal{P}}))	\\
(\sk, \pk, \tau, \sigma_{\mathcal{P}}) &\leftarrow \VDBQCCC.\StatePrep(\mathcal{V}(C, 1^n), \mathcal{P}(\proj{0}^{\otimes \poly(n)}))\\	
\end{align*}
\item \emph{$\eps$-blind-verifiability:} the protocol satisfies the $\eps$-\textsf{IND-CPA} security condition of~\ref{def:qcedcc} and in addition, for each quantum circuit $C$, each verifier input $m_{\mathcal{V}}$, each $\QPT$ prover $\mathcal{P}$ and each state $\rho_{\algo P}$ it is the case that 
$$ \Pr [ r \neq M(C \proj{m_{\mathcal{V}}} C^{\dagger}) \land \top ] \leq \eps(n)$$
where $M(\cdot)$ denotes the process of measuring the first qubit of a quantum state in the standard basis and outputting that result, $\top$ denotes the event of the protocol \emph{not aborting} and where,
\begin{align*}
r &\leftarrow \VDBQCCC.\Dec(\sk^*, \ct^*) \\
(\sk^*, \ct^*, \gamma_{\algo P}) &\leftarrow \VDBQCCC.\Evaluate
(\mathcal{V}(C, m_{\mathcal{V}}, \sk, \pk, \tau), \mathcal{P}(\pk, \tau, \sigma_{\mathcal{P}}))	\\
(\sk, \pk, \tau, \sigma_{\mathcal{P}}) &\leftarrow \VDBQCCC.\StatePrep(\mathcal{V}(C, 1^n), \mathcal{P}(\rho_{\algo P}))\\	
\end{align*}
\end{enumerate}
\end{definition}

To obtain a protocol satisfying the above definition, for some computation $C$ on input $x$, we follow the ideas in Morimae's protocol and have the verifier first perform a $\QCEDCC$ protocol with the prover, instructing it to prepare the states $\ket{\psi_{C, x}}$ and $\ket{\psi_{\bar{C}, x}}$, as described in~\ref{subsec:morimae}. 
The verifier must then instruct the prover to perform the measurement of local terms from the Hamiltonian $H_{C, x}$. We do this by making use of Mahadev's measurement protocol from~\ref{subsec:mahadev}. An important point to note here is that because the states $\ket{\psi_{C, x}}$ and $\ket{\psi_{\bar{C}, x}}$ were prepared via the $\VDBQCCC$ protocol, they will be encrypted and so it's not a priori clear whether the subsequent measurement is useful for the verifier. But this will indeed be the case, as the encryption of the states is given a by a quantum one-time pad (that is known to the verifier) which can be absorbed into the measurements by suitably flipping measurement outcomes.
The protocol is described more formally in~\ref{fig:VDBQCCC_protocol} and the following theorem shows that it satisfies~\ref{def:vdbqccc}

\begin{table}[ht!]
\begin{longfbox}[breakable=false, padding=1em, padding-right=1.8em, padding-top=1.2em, margin-top=1em, margin-bottom=1em, background-color=gray!20]
\begin{protocol} {\bf Protocol for Verifiable Delegated Blind Quantum Computation with a Classical Verifier} \label{fig:VDBQCCC_protocol} \end{protocol}
Let $C$ be a circuit acting on $n$ qubits. Additionally, take $\delta = 1/\poly(n)$ to be some error tolerance parameter.
Define the following algorithms.
\begin{description}

\item $\VDBQCCC.\StatePrep(\mathcal{V}(C, 1^n), \mathcal{P}(\rho_{\algo P}) \to (\sk, \pk, \tau, \sigma_{\mathcal{P}})$: The verifier runs~\ref{fig:protocol_multi_round} with the prover, taking the security parameter $\lambda = \poly(n)$. Here, $\pk$ and $\sk$ denote the public and secret keys generated in that protocol, $\tau$ denotes the transcript of the protocol and $\sigma_{\mathcal{P}}$ is the prover's state upon completing the protocol (in the case where the verifier accepts).

\item $\VDBQCCC.\Evaluate(\mathcal{V}(C, m_{\mathcal{V}}, \sk, \pk, \tau), \mathcal{P}(\pk, \tau, \sigma_{\mathcal{P}})) \to 
(\sk^*, \ct^*, \gamma_{\mathcal{P}})$: The verifier first performs the $\QCEDCC.\Evaluate(\mathcal{V}(C', m_{\mathcal{V}}, \{\sk, \sk_{in}\}, \pk, \tau), \mathcal{P}(\pk, \tau, \sigma_{\mathcal{P}}))$ subprotocol, where the circuit $C'$ is one that prepares the history states $\ket{\psi_{C, m_{\mathcal{V}}}}$ and $\ket{\psi_{\bar{C}, m_{\mathcal{V}}}}$ and where $\sk_{in} \rand \bit^{\poly(n)}$ (see~\ref{fig:QCEDCC_protocol} for details). It then performs Mahadev's measurement protocol, instructing the prover to measure each qubit of the (encrypted) history states in either the standard basis or the Hadamard basis, chosen uniformly at random. The string $\ct^*$ is set to be the transcript of that protocol. The verifier also samples a local term from each of the Hamiltonians $H_{C, m_{\mathcal{V}}}$ and $H_{\bar{C}, m_{\mathcal{V}}}$ with probability proportional to the weight of that term (as in the post-hoc protocol mentioned in~\ref{subsec:histstate}). The string $\sk^*$ is set to the choices of these local terms, together with the strings $\sk$ and $\sk_{in}$.

\item $\VDBQCCC.\Dec(\sk^*, \ct^*) \to r$: The verifier uses $\sk^*$ to decrypt the results in $\ct^*$, corresponding to the measurements of local terms of the Hamiltonians $H_{C, m_{\mathcal{V}}}$ and $H_{\bar{C}, m_{\mathcal{V}}}$ on the states $\ket{\psi_{C, m_{\mathcal{V}}}}$ and $\ket{\psi_{\bar{C}, m_{\mathcal{V}}}}$\footnote{Strictly speaking, this happens with probability $1/4$ as we are following the instance-independent setup described in~\ref{rem:instindep}.}. If the tests of the Mahadev's measurement protocol failed, the verifier aborts. Otherwise, if, according to the measurements, $\ket{\psi_{C, m_{\mathcal{V}}}}$ has low energy with respect to $H_{C, m_{\mathcal{V}}}$ and $\ket{\psi_{\bar{C}, m_{\mathcal{V}}}}$ has high energy with respect to $H_{\bar{C}, m_{\mathcal{V}}}$\footnote{``Low energy'' and ``high energy'' meaning that the expected outcome is below a threshold $\alpha$ or above a threshold $\beta$, as in~\ref{thm:posthoc}.}, it sets $r=0$ and accepts. If $\ket{\psi_{C, m_{\mathcal{V}}}}$ has high energy with respect to $H_{C, m_{\mathcal{V}}}$ and $\ket{\psi_{\bar{C}, m_{\mathcal{V}}}}$ has low energy with respect to $H_{\bar{C}, m_{\mathcal{V}}}$, it sets $r=1$ and accepts. Otherwise, it aborts.
\end{description}
\end{longfbox}
\end{table}

\begin{theorem}
\ref{fig:VDBQCCC_protocol} is a protocol for Verifiable Delegated Blind Quantum Computation with a Classical Verifier having $\negl(n)$ correctness and $1/\poly(n)$ blind-verifiability.
\end{theorem}

\begin{proof}
We start by showing the protocol's correctness.
The prover will first engage in the $\QCEDCC$ protocol with the verifier, resulting in the preparation of BB84 states in its memory, followed by the application of the unitary for preparing the (encrypted) history states $\ket{\psi_{C, m_{\mathcal{V}}}}$ and $\ket{\psi_{\bar{C}, m_{\mathcal{V}}}}$. As explained in~\ref{sec:qced}, these states are encrypted via a quantum one-time pad known to the verifier. The verifier then uses Mahadev's measurement protocol to instruct the prover to measure the states according to randomly chosen terms from the Hamiltonians $H_{C, x}$ and $H_{\bar{C}, x}$, respectively (using the instance-independent idea from~\ref{rem:instindep}). Again, as the prover is assumed to behave correctly, this will result in the correct measurements being performed, as follows from~\ref{thm:completemeas}.
Recall that the terms of the Hamiltonian comprise only Pauli $\sigma_X$ and $\sigma_Z$ operators (and so the measurements performed via Mahadev's protocol will also be a measurements of $\sigma_X$ or $\sigma_Z$ observables). Since the history states have a quantum one-time pad applied, we can ``absorb'' the pad into the measurement by suitably flipping the measurement outcomes. More precisely, if a qubit has a $\sigma_X$ ($\sigma_Z$, respectively) applied in the one-time pad and the qubit is measured with the $\sigma_Z$ ($\sigma_X$, respectively) observable, then the outcome should be flipped. Otherwise, the outcome is not flipped.
Once the verifier has decrypted the measurement outcomes in this way, it will then test the energy condition to estimate the energies of the two states relative to the Hamiltonians $H_{C, x}$ and $H_{\bar{C}, x}$. 
The correctness of the history state construction (\ref{thm:posthoc}) ensures that the verifier obtains the correct outcome with high probability\footnote{We are assuming here that the verifier and prover are performing the parallel-repeated version of these protocols, as described in~\cite{alagic2020non}, hence why correctness is negligible.}. 

We would now like to show that the protocol also satisfies $1/\poly(n)$ blind-verifiability.
The blindness condition, which is identical to that of $\QCEDCC$, is clearly satisfied up to the point where the verifier performs Mahadev's measurement protocol. This is because until that point, the protocol is simply an application of $\QCEDCC$ for a specific computation (that of preparing the two history states). Using the instance-independent version of the history state construction (see~\ref{rem:instindep}), Mahadev's measurement protocol also satisfies the \textsf{IND-CPA} condition. This is due to the fact that the instructions the verifier sends to the prover in the measurement protocol are independent of the input.

Verifiability follows from the soundness of the measurement protocol (\ref{thm:soundmeas}) together with the properties of the history state construction (\ref{thm:posthoc}). The proof is identical to the proof of Theorem 8.6 from~\cite{mahadev}, which we will not restate here. It follows that
\begin{equation*}
|  \Pr [ r \neq M(C \ket{m_{\mathcal{V}}}\bra{m_{\mathcal{V}}} C^{\dagger}) \land \top] - 1/\poly(n)| \leq \negl(n),
\end{equation*}
which then implies that 
\begin{equation*}
    \Pr [ r \neq M(C \ket{m_{\mathcal{V}}}\bra{m_{\mathcal{V}}} C^{\dagger}) \land \top] \leq 1/\poly(n) + \negl(n).
\end{equation*}
This yields the desired condition of $1/\poly(n)$ verifiability that, when combined with the blindness condition (i.e. the \textsf{IND-CPA} security) yields $1/\poly(n)$ blind-verifiability, concluding the proof.
\end{proof}

\bibliographystyle{alpha}
\bibliography{main}

\end{document}